%% file: arXiv.tex
\documentclass[journal]{IEEEtran}
\input{header}

\title{Self-Creating Random Walks for  Decentralized Learning under Pac-Man Attacks}

\author{
Xingran Chen\IEEEauthorrefmark{1},~\IEEEmembership{Member,~IEEE}, 
\and Parimal Parag\IEEEauthorrefmark{2},~\IEEEmembership{Senior~Member,~IEEE},\\
\and Rohit Bhagat\IEEEauthorrefmark{3},~\IEEEmembership{Student~Member,~IEEE},~
\and Salim El Rouayheb\IEEEauthorrefmark{3},~\IEEEmembership{Senior~Member,~IEEE}
\IEEEcompsocitemizethanks 	 
{
\IEEEcompsocthanksitem  Xingran Chen is with Engineering Systems and Design Pillar, Singapore University of Technology and Design, 8 Somapah Road, Singapore 487372	(E-mail: 
xingranc@ieee.org).
\IEEEcompsocthanksitem Parimal Parag is with the Department of Electrical Communication
Engineering, Indian Institute of Science, Bangalore, Karnataka 560012, India (E-mail: parimal@iisc.ac.in).
\IEEEcompsocthanksitem  Rohit Bhagat, and Salim El Rouayheb are with Department of Electrical and Computer Engineering,  Rutgers University, Piscataway Township, NJ 08854, USA	(E-mail: \{rb1395, sye8\}@scarletmail.rutgers.edu).

}
\IEEEcompsocitemizethanks{
The work is supported 
in part by National Science Foundation of China (Grant No. 62401111); 
in part by Scheme for Promotion of Academic and Research Collaboration under Grant SPARC/2025-2026/P4239; 
in part by Qualcomm Inc. under Qualcomm University Relations 6G India; 
in part by Anusandhan National Research Foundation under Grant CRG/2023/008854; in part by Centre for Networked Intelligence (a Cisco Corporate Social Responsibility (CSR) Initiative), IISc; in
part by the Army Research Lab (ARL) under Grant W911NF-24-2-0172. 
Any opinions, findings, conclusions, and recommendations expressed in this material are those of the authors and do not necessarily reflect the views of funding agencies. 
}
\IEEEcompsocitemizethanks{This work has been accepted for publication in the IEEE Journal on Selected Areas in Communications (JSAC) in 2026.}
}
\begin{document}
\maketitle
\begin{abstract}
Random walk (RW)-based algorithms have long been popular in distributed systems due to low overheads and scalability, with recent growing applications in decentralized learning. However, their reliance on local interactions makes them inherently vulnerable to malicious behavior. 
In this work, we investigate a termination-based  attack, termed the ``Pac-Man'' attack, in which a malicious node probabilistically terminates any RW that visits it. 
This stealthy behavior gradually eliminates active RWs from the network, effectively halting the learning process without triggering failure alarms. 
To counter this threat, we propose the \CWL (\cwl) algorithm, which is a fully decentralized mechanism that is resilient to termination-based Pac-Man attacks by enabling the self-creation of RWs and preventing RW extinction under such attacks.
Our theoretical analysis shows that the \cwl algorithm guarantees several desirable properties, such as (i) non-extinction of the RW population, (ii) almost sure boundedness of the RW population, and %and the algorithm parameter can be theoretically selected to control the expected peak RW population below any prescribed level. 
(iii) convergence of RW-based stochastic gradient descent even in the presence of Pac-Man with a quantifiable deviation from the true optimum. 
Moreover, the learning process experiences at most a linear time delay due to Pac-Man interruptions and RW regeneration.
Our extensive empirical results on both synthetic and public benchmark datasets validate our theoretical findings. 
\end{abstract}

\section{Introduction}\label{sec:Introduction}
Decentralized algorithms are becoming increasingly important in modern large-scale networked applications. 
These algorithms enable a network of nodes/agents, each with access only to local data and a limited local view of the connection graph, to collaborate in solving global computational tasks without relying on centralized coordination. 
Among the most widely studied approaches are random-walk (RW)-based algorithms \cite{lovasz1993random, levin2017markov} and gossip-based algorithms \cite{tsitsiklis2003distributed, pmlr-v119-koloskova20a}. 
RW-based algorithms operate by letting one or more tokens perform RWs over the graph; local updates are carried out only at the nodes currently holding tokens\footnote{Throughout this paper, we use the terms ``random walk'' and ``token'' interchangeably.}. 
Gossip-based algorithms, by contrast, rely on frequent broadcast communication in the local neighborhood of the nodes where they exchange model snapshots and update locally based on received messages.

While consensus-based methods are powerful, their repeated local broadcasts can lead to large communication overhead in large-scale systems \cite{tsitsiklis2003distributed, pmlr-v119-koloskova20a}. 
This motivates the study of RW-based algorithms, which offer a communication-efficient alternative for large networks due to their simplicity and scalability \cite{lovasz1993random, levin2017markov}. 
By propagating information through tokens that traverse the network, RW-based algorithms require only local updates, no global coordination, and low communication overhead. 
These features have led to their successful application across a wide range of domains \cite{page1999pagerank, fouss2007random, liu2016smartwalk, backstrom2011supervised, zhang2019heterogeneous}. 
A particularly compelling application is decentralized machine learning, where RWs are used to aggregate learning on data distributed across networked agents \cite{johansson2007simple, ayache2023walk}. While our focus in this article is on decentralized machine learning, the techniques we develop are broadly applicable to other domains that rely on RW-based methods.

Despite their advantages, RW learning algorithms are susceptible to failure and security threats \cite{10693599, 9714881}. In particular, they are vulnerable to a class of attacks that evade classical defenses against poisoning and Byzantine behavior \cite{10633251,selfdup} by executing locally correct computations while silently terminating random walks, ultimately halting the learning process. In this work, we consider a specific instance of such behavior, which we term the \emph{Pac-Man attack}, in which a malicious node probabilistically terminates any random walk that visits it. We refer to this adversary as \emph{Pac-Man}, evoking the arcade character known for devouring everything in its path.

A Pac-Man node (malicious node) \textit{mimics} honest behavior under standard fault-tolerance mechanisms \cite{schneider1990implementing, lamport1998part, lamport1982byzantine, elnozahy2002survey, chandra1996unreliable, 10.5555/320673}, such as retransmissions, replication, or timeout-based recovery, thereby making it difficult to detect or isolate via conventional means. 
In fact, it can terminate any RW that visits it,  effectively removing it from the network. Even when multiple redundant RWs are used to perform learning across the network, a single Pac-Man node can eventually lead to the extinction of all walks---effectively halting the learning process without triggering any explicit failure or alarm. In this work, we focus on the case of a single Pac-Man node to facilitate clarity of exposition and present the key insights. 
This restriction does not limit the generality of our results. 
Both the proposed algorithm and the theoretical framework can be extended to settings with multiple Pac-Man nodes, as illustrated in Fig.~\ref{fig:DeCacomplete}(a), and discussed in Remark~\ref{remark:MultiplePacMan} and Appendix~G.

\begin{figure}[t!]
 \centering
\includegraphics[width=0.35\textwidth]{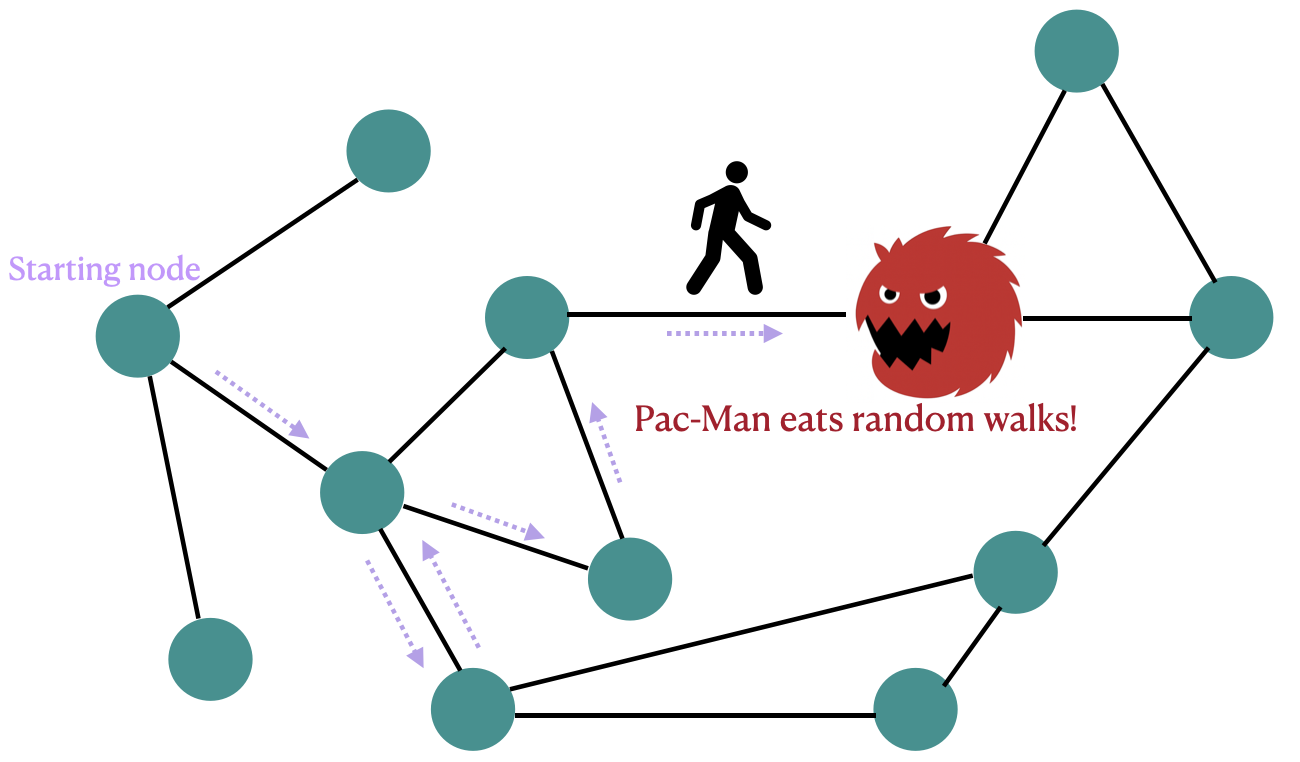}
\vspace{2mm}
\caption{An illustration of the Pac-Man attack: A malicious node (shown in red) intercepts and  ``eats'' (terminates) any visiting RW with a positive probability. Despite behaving like a benign node to its neighbors, it prevents the random walk from continuing, leading to eventual extinction of all walks in the network.}
 \end{figure}

Motivated by this limitation, we propose a new decentralized self-regulation mechanism based on creation rather than duplication. 
Our algorithm, termed \CWL (\cwl), is designed to maintain long-term resilience against Pac-Man attacks without relying on system parameter estimation or global coordination (see Algorithm~\ref{alg:LC}). 
In \cwl, each benign node independently decides whether to create a new RW using only local timing information, namely the time elapsed since the last visit of any RW. If this interval exceeds a predetermined threshold, the node infers a possible RW termination and probabilistically creates a new RW by replicating the most recently visited one. This locally triggered mechanism enables resilient RW creation in a fully decentralized manner.

A novel aspect of our analysis lies in the fact that, although the RW population process in our work, bears resemblance to classical branching processes, such as the Galton-Watson process~\cite[Chapter~I]{harris1964branching}, the standard assumptions underlying these models do not hold in our setting. 
In particular, the RW creation rate is density-dependent and depends on the graph topology as well as the joint system state, including the locations of all RWs and the visitation histories of honest nodes. 
These features introduce intricate dependencies that violate the independence assumptions required by classical branching process theory, and as such standard tools and results do not directly apply.

%%%%%%%%%%%%%%%%%%%%%%%%%%%%%%%%%%%%%%%%%%%%%%%%%%
\subsection{Contributions}
%%%%%%%%%%%%%%%%%%%%%%%%%%%%%%%%%%%%%%%%%%%%%%%%%%
In the following, we summarize the main contributions of this article. %are summarized as follows:
\begin{compactenum}[(i)]
\item \textbf{Resilience to termination-based Pac-Man attacks through local decisions, without the need to estimate system parameters}. Unlike duplication-based methods such as \DeCa, \cwl guarantees long-term survivability of the RW population and prevents permanent extinction (see Fig.~\ref{fig:Boundedness} in Section~\ref{sec: simulations}). 
Consequently, computational tasks carried by RWs remain operational even in an  adversarial environment and do not depend on an accurate estimation of the number of alive RWs.
\item \textbf{Novel theoretical framework and analytical results in an adversarial setting}. 
We contribute to developing a novel and rigorous theoretical framework to analyze the evolution of the number of active RWs under \cwl, explicitly accounting for the strong interdependence among RWs. 
We establish that the RW population under \cwl is almost surely bounded (Theorem~\ref{thm:FiniteRWs}) and provide a theoretical regime for selecting the algorithm parameter such that the expected peak population stays below any prescribed level (see Theorem~\ref{thm:Peak}). 
Together, these results guarantee both the long-term stability and controllability of the RW population in the presence of adversarial attacks.
\item \textbf{Integration with decentralized learning and proof of convergence}. 
The \cwl mechanism can be seamlessly integrated into decentralized random walk stochastic gradient descent (RW-SGD) \cite{johansson2007simple, sun2018markov, mao2020walkman}. 
We prove that RW-SGD converges under the \cwl mechanism, even in the presence of an adversary (see Theorem~\ref{thm:Effectiveness}).  
We further characterize the bias induced by premature RW termination, showing how adversarial behavior skews the optimizer and leads to a bounded deviation from the true global optimum (Proposition~\ref{pro:Bounds}). 
Moreover, we establish that the inactive time introduced by the adversary is at most a constant fraction of the current time
(Proposition~\ref{pro:NumIter}), ensuring that the effective learning iterations remain at least linearly increasing in time.  
Numerical experiments on both synthetic and public benchmark datasets validate our theoretical results.
\end{compactenum}

\noindent\textbf{Notation:} 
We denote the set of first $N$ consecutive positive integers as $[N] \triangleq \set{1, \dots, N}$. 
For any finite set $\cV$, we denote the collection of probability measures on set $\cV$ by $\cM(\cV) \triangleq\set{\alpha \in [0,1]^\cV: \sum_{v \in \cV}\alpha_v = 1}$.

\subsection{Related Work}

The closest related works are RW-based decentralized optimization algorithms with self-duplication mechanisms in a fully decentralized manner \cite{selfdup, egger2024self, AC}. In \DeCa \cite{selfdup, egger2024self}, each node records RW visit times and estimates the number of active RWs. When the estimated population falls below a predefined threshold, duplication is triggered to replenish the RW population. Similarly, in \ac \cite{AC} each node records the time interval between consecutive RW visits and duplicates a RW when the observed interval exceeds a prescribed threshold. 

Both approaches aim to maintain a sufficient RW population using only local observations. However, the underlying setting and robustness objective are different. The original \DeCa framework is motivated by non-adversarial catastrophic events, rather than by an adversary that strategically terminates incoming RWs. In addition, the robustness of \DeCa and \ac is highly sensitive to parameter choices and the RW population may eventually collapse (see Fig. \ref{fig:DeCa} in Section \ref{sec: simulations} and Fig.~2(b) in \cite{AC}). That is, \DeCa and \ac do not remain consistently resilient under the Pac-Man attack. 
%In other words, duplication-based mechanisms cannot fundamentally prevent the extinction of the RW population under adversarial termination, and are very fragile to parameter estimation errors.  

The Pac-Man attack shares similarities with other attacks in distributed communication and learning systems; however, it differs in both the entity it targets and its effect on information propagation. The closest analogy is a blackhole attack \cite{4448606,7007762,edemacu2014packet}, in which a malicious intermediate node drops packets intended for a particular destination without changing the underlying network topology. In contrast, an RW has no predetermined destination and propagates model updates through its continued movement across the network. Thus, a Pac-Man node terminates an entire persistent computational process, eliminating its future trajectory and all subsequent model updates it would have performed, rather than merely preventing a single packet from reaching its destination. Moreover, because the Pac-Man node remains in the graph and may otherwise behave normally, it can be difficult to detect using only local observations.

The Pac-Man attack also differs from random and localized node-removal attacks \cite{PALMIERI2025108250,chujyo2026impact,berezin2015localized}, which alter the communication topology and may fragment the network. Pac-Man nodes remain active in the graph and instead disrupt learning by terminating visiting RWs. Another related disruptive behavior is that of stragglers, which arise from heterogeneous computation or communication resources and cause some nodes to respond substantially more slowly than others \cite{10534807, Reisizadeh2020StragglerResilientFL,tziotis2023stragglerresilient,10.1145/3641512.3690036,he2025straggler,liang2026towards}. Stragglers delay the dissemination of information, whereas Pac-Man nodes eliminate the RWs that carry and propagate it.

\section{Problem Formulation}\label{sec:ProblemFormulation}
%%%%%%%%%%%%%%%%%%%%%%%%%%%%%%%%%%%%%%%%%%%%%%%%%%
We consider a decentralized system consisting of $N$ agents that collaboratively perform a computational task without any central coordination. 
%%%%%%%%%%%%%%%%%%%%%%%%%%%%%%%%%%%%%%%%%%%%%%%%%%
\subsection{Graph and Random Walks}\label{subsec:GraphRW}
%%%%%%%%%%%%%%%%%%%%%%%%%%%%%%%%%%%%%%%%%%%%%%%%%%
These agents are modeled as nodes in a {\it connected} undirected graph, where each agent possesses its own local data and can communicate only with neighboring agents.\footnote{Without loss of generality, we restrict our analysis to connected graphs. For disconnected graphs, the proposed theoretical framework can be applied separately to each connected component, and the analysis proceeds analogously.} Without loss of generality, we assume that the Pac-Man agent is indexed as node~$1$. 
The communication topology graph is constructed as a finite graph $\cG \triangleq (\cV, E)$, where the set of agents/nodes is $\cV\triangleq [N]$, and $E \subseteq \binom{\cV}{2}$ denotes the set of edges. In addition, let $\cB\triangleq[N]\backslash\set{1}$ denote the set of benign nodes.

In this system, computation is carried out via RWs on the graph. 
We assume that there is a fixed and identical transition probability matrix $P$ for each RW on this graph. Each RW carries a token message that is processed and passed along the network. 
At each discrete time step, only the node currently holding a token performs local computation and updates the message. After computation, the current node forwards the token to a randomly selected neighbors based on the transition probability matrix $P$. 
This process continues until a predefined stopping criterion is satisfied. 
When multiple RWs are present, each is uniquely identifiable (e.g., via an index $j$) to track their individual progress through the network.

%%%%%%%%%%%%%%%%%%%%%%%%%%%%%%%%%%%%%%%%%%%%%%%%%%
\subsection{Threat Model: The Pac-Man Attack}\label{subsec:Pack-ManAttack}
%%%%%%%%%%%%%%%%%%%%%%%%%%%%%%%%%%%%%%%%%%%%%%%%%%
We are interested in ensuring the resiliency of RW-based algorithms---that is, their ability to prevent total  extinction and continue operating effectively even when a malicious node attempts to terminate the RWs. Specifically, we focus on a threat we refer to as the \emph{Pac-Man attack}: a malicious node, termed \emph{Pac-Man}, terminates  all incoming RWs without performing the required computation or forwarding the result to a neighbor. We restrict our attention in this paper to the case of a single Pac-Man node to facilitate the presentation of the proposed algorithm and its theoretical analysis. This setting already captures the core difficulty of the problem: even one such adversary is sufficient to cause eventual extinction of all random walks with probability one, thereby completely halting the learning process. Extensions to settings with multiple Pac-Man nodes are discussed in subsequent sections.

The Pac-Man can be particularly dangerous due to its ability to remain hidden through a deceptive behavior. 
\begin{compactenum}[(i)]
\item  The Pac-Man node can reply positively to all network-level fault-tolerance checks and controls, for example retransmissions, making it difficult to distinguish it from benign nodes despite its malicious actions.
\item To avoid detection, the Pac-Man node terminates incoming RWs independently with a probability $\zeta \in (0, 1]$, referred to as the termination probability. This randomized behavior allows the Pac-Man to conceal itself among benign nodes. If the Pac-Man node terminates all incoming RWs, it would never propagate RWs to its neighbors, making it easily identifiable as malicious since no RWs would be observed from that node over a long time horizon. In contrast, by terminating incoming RWs only with a positive probability, the Pac-Man can intermittently forward RWs, thereby blending in with normal nodes. Under decentralized operation, where each node relies solely on local observations, this stochastic behavior significantly delays detection by neighboring nodes.
\end{compactenum}
In decentralized systems, traditional fault-tolerance mechanisms are challenging when facing such an adversary, because 
unlike conventional failures---such as link disconnections, node crashes, or consistently faulty behavior \cite{coulouris2011distributed}---a Pac-Man remains active and selectively disrupts system operation while appearing compliant. Our goal is to design a mechanism that prevents a Pac-Man node from causing a system-wide failure by eliminating all RWs.

Another commonly used fault-tolerance technique is the introduction of redundancy.  However, static redundancy alone is {\it ineffective} in this adversarial setup. Simply starting with multiple RWs does not guarantee their survival: for any small termination probability $\zeta>0$, all RWs will eventually be terminated with probability $1$, causing the task to fail. Thus, redundancy cannot serve as a long-term solution. 

%%%%%%%%%%%%%%%%%%%%%%%%%%%%%%%%%%%%%%%%%%%%%%%%%%
\subsection{Random Walk Stochastic Gradient Descent}\label{subsec:RWSGD}
%%%%%%%%%%%%%%%%%%%%%%%%%%%%%%%%%%%%%%%%%%%%%%%%%%
In this paper, the global task is defined as the standard distributed optimization problem:
\begin{align}\label{eq:goal0}
\min_{{\bf x}\in\R^m}f({\bf x}) = \min_{{\bf x}\in\R^m} \,\E_{u\sim\pi}\left[f_u({\bf x})\right],\,\,
\end{align}
where each node $u$  possesses a local function $f_u({\bf x})$, $\pi$ is a target sampling distribution, and $m$ is a positive integer. 

Under local communication constraints, solving problem \eqref{eq:goal0} using RW-based stochastic gradient descent (RW-SGD) algorithms has proven to be highly effective \cite{10.5555/3618408.3618786, doi:10.1137/08073038X, 10.5555/3327546.3327656}. 
This decentralized approach exploits local interactions within a network to collaboratively achieve global optimization objectives.

Given a target sampling distribution $\pi$ \big(assuming $\pi_u>0$ for all $u\in\cV$\big) and a connected graph $\cG$, the Metropolis-Hastings algorithm \cite{doi:10.1137/08073038X} provides a principled way to construct a RW  on $\cG$ with a transition matrix $P$ such that $\pi$ is the stationary distribution of the Markov chain defined by $P$.\footnote{The transition matrix $P$ depends on both the target sampling distribution $\pi$ and the graph $\cG$. We re-parameterize it as $P_{\pi,\cG}\in\R_+^{N\times N}$, and simply write $P$ when there is no risk of confusion.} The RW-SGD is outlined as
\begin{align}\label{eq: updatingRW-SGD}
{\bf x}_{t+1} ={\bf x}_t - \gamma_t\hat{g}_{v_t}({\bf x}_t),
\end{align}
where $v_t$ is the node visited by the random walk at time $t$, and $\hat{g}_{v_t}({\bf x}_t)$ denotes the gradient or sub-gradient at that node. Here, the sequence $\set{v_t}_{t\in\N}$ is a RW that evolves according to $P$. 

In the absence of the Pac-Man node, RW-SGD is known to converge to the global optimum under standard assumptions \cite{10.5555/3618408.3618786, doi:10.1137/08073038X, 10.5555/3327546.3327656}. 
However, the introduction of a Pac-Man node fundamentally alters the dynamics: by probabilistically terminating incoming RWs, it becomes unclear whether RW-SGD remains convergent. Even if convergence occurs, it is not known a priori whether RW-SGD converges to the true optimum or to a biased solution.

%%%%%%%%%%%%%%%%%%%%%%%%%%%%%%%%%%%%%%%%%%%%%%%%%%
\subsection{Designable Properties of the Decentralized Mechanism}\label{subsec:Objectives}
%%%%%%%%%%%%%%%%%%%%%%%%%%%%%%%%%%%%%%%%%%%%%%%%%%
As stated in Sections~\ref{subsec:Pack-ManAttack} and~\ref{subsec:RWSGD}, the entire goal is to design a \textit{decentralized mechanism} such that (i) the global computational task can be carried out via RWs even in the presence of a Pac-Man node, and
(ii) when the proposed decentralized mechanism is integrated with RW-SGD, it remains effective and convergent.

At the beginning of time slot $t$, let $\cZ_t$ denote the set of indices of active RWs, and define the random variable $Z_t \triangleq \abs{\cZ_t}$ as the total number of active RWs at that moment. 
For each RW $j \in \cZ_t$ at time $t$, we denote its location at time $t$ by $X_{j}(t) \in \cV$. 
At the initial time $t = 0$, we denote the number of initial RWs as $Z_0 = z_0$, where $z_0$ is a predetermined scalar. 
To achieve our entire goal, the proposed algorithm must have the following desirable properties. 
\paragraph{\bf No Permanent Extinction} 
To ensure that the global task can be accomplished, we must avoid permanent extinction of the RW population almost surely: 
\begin{align}\label{eq:NoPermanentExtinction}
\Prb\left(\exists t_0,\,\forall t\ge t_0,\, Z_t=0 \,\middle|\, Z_0=z_0\right)=0.
\end{align}

\paragraph{\bf No Blowup}
To ensure system stability, we must avoid uncontrolled growth in the number of RWs.  The algorithm should keep the RW population bounded almost surely:
\begin{align}\label{eq:NoBlowup}
\Prb\left(\sup_t Z_t<\infty\,\middle|\, Z_0=z_0\right) = 1.
\end{align} 

\paragraph{\bf Convergence}
All active RWs under RW-SGD must converge to the same minimizer, in the following sense: 
\begin{align}\label{eq:Effectiveness}
\lim_{t \to \infty} \E\norm{\bx^{(j_t)}_t - \tilde{\bx}^\star}=0,\; \forall j_t\in\cZ_t,
\end{align}
where $\tilde{\bx}^\star$ denotes the limiting point attained by the RW-SGD iterates across active RWs. If the limit in \eqref{eq:Effectiveness} exists, we further aim to characterize the approximation error with respect to the true optimizer $\bx^\star$, by bounding the deviation
$
\norm{\tilde{\bx}^\star - \bx^\star}.
$

\begin{remark}\label{remark:nogossip}
It is worth noting that, \eqref{eq:Effectiveness} characterizes the convergence behavior along any active RW. This is a substantially stronger conclusion than results based on averaging over all active RWs or over all benign nodes, as is commonly done in gossip-based algorithms \cite{tsitsiklis2003distributed}.
\end{remark}

%%%%%%%%%%%%%%%%%%%%%%%%%%%%%%%%%%%%%%%%%%%%%%%%%%
\section{The \CWL Algorithm }\label{sec:AC}
%%%%%%%%%%%%%%%%%%%%%%%%%%%%%%%%%%%%%%%%%%%%%%%%%%
We present the \CWL (\cwl) algorithm, a decentralized mechanism for adaptively creating RWs with the desirable properties of (a) no permanent extinction, (b) no blowup, and (c) convergence, as formalized in \eqref{eq:NoPermanentExtinction}, \eqref{eq:NoBlowup} and \eqref{eq:Effectiveness}.  We first describe the \cwl algorithm and then summarize our main theoretical results.

%%%%%%%%%%%%%%%%%%%%%%%%%%%%%%%%%%%%%%%%%%%%%%%%%%
\subsection{\CWL Algorithm}
%%%%%%%%%%%%%%%%%%%%%%%%%%%%%%%%%%%%%%%%%%%%%%%%%%
We outline the \cwl algorithm in Algorithm~\ref{alg:LC}. The core idea is that each benign node creates a new RW if it has not been visited for too long. 
In particular, each benign node $u \in \cB$ maintains a variable $L^{(u)}_t$, which records the last time \emph{before $t$} that node $u$ was visited by any RW. 
Using this timestamp, node $u$ computes the elapsed time since its most recent visit as $t - L^{(u)}_t$.  
If this interval exceeds a predetermined threshold $A_u$, node $u$ infers that some RW may have been lost, and with probability $q$, creates an \textit{identical} copy of the most recently visited RW stored in its local cache (lines~9, 10, 11 in Algorithm~\ref{alg:LC})\footnote{Similar timeout-based mechanisms have been studied for handling stragglers in decentralized learning, where the timeout threshold affects the trade-off between latency and learning performance~\cite{10534807}.}. The newly created RW is then propagated independently according to the transition matrix $P$. If multiple RWs arrive at node $u$ simultaneously, the node retains any one of them at random.

\begin{algorithm}
\caption{\CWL (\cwl) Algorithm}\label{alg:LC}
\begin{algorithmic}[1]
\State {\bf Input}: The graph $\mathcal{G}$, the thresholds $\set{A_u}_{u\in\cB}$, the initial recording $L_0^{(u)}=0$ for $u\in\cB$, the creation probability $q$, and the initial location of RWs $\cZ_0$.
\For{$t\ge 0$}
\For{$u\in\cB$}
\If{$u\in\cup_{j\in\cZ_t}X_j(t)$} 
\State $L_t^{(u)}\gets t$.
\Else
\If{$t - L_t^{(u)} \le A_u$}
\State $L_t^{(u)}\gets L_{t}^{(u)}$.
\Else
\State With probability $q$, node $u$ generates a new RW by replicating the RW that last visited it at time $L_t^{(u)}$.
\State If a new RW is created, $L_t^{(u)}\gets t$; otherwise $L_t^{(u)}\gets L_{t}^{(u)}$.
%\Else 
%\State $L_t^{(u)}\leftarrow L_{t}^{(u)}+1$.
\EndIf
%\Else
%\State $L_t^{(u)}\leftarrow t$.
\EndIf
\EndFor
\EndFor
\end{algorithmic}
\end{algorithm}
\begin{remark}\label{remark:noextinction}
We observe that, under the construction of the \cwl algorithm, whenever there are no RWs on the graph, at least one new RW will be generated within at most $\min_{u\in\cB}{A_u}$ time slots from the last time all RWs got terminated.  
Consequently, the RW population can never become permanently extinct, and the objective of achieving no permanent extinction defined in~\eqref{eq:NoPermanentExtinction} is inherently satisfied.
\end{remark}

Remark~\ref{remark:noextinction} is illustrated in Fig.~\ref{fig:DeCacomplete}(a), where we compare the resilience of the proposed \cwl algorithm with the baseline \DeCa algorithm in \cite{selfdup, egger2024self}.  
The figure plots the evolution of the number of RWs over time.  
Under the \cwl algorithm, even when the creation threshold is inappropriately chosen, the RW population may temporarily go extinct (hit $0$) but eventually recovers after a finite period (green curve).  
In contrast, when the algorithm parameter in the \DeCa algorithm is inappropriately chosen, the RW population goes extinct permanently (brown curve). This behavior highlights a key robustness advantage of the proposed algorithm over the baseline. 
This phenomenon is not unique to complete graphs; similar behavior is observed on other graph topologies as shown in Fig.~\ref{fig:DeCa} in Section~\ref{sec: simulations}.

Next, we numerically demonstrate that the proposed algorithm remains effective in the presence of multiple Pac-Man nodes.
In Fig.~\ref{fig:DeCacomplete}(b), we consider a complete graph with $100$ nodes, where all benign nodes share the same creation threshold $A = 350$. We examine three scenarios with $1$, $3$, and $10$ Pac-Man nodes, respectively.
The results show that even in the presence of a large number of Pac-Man nodes (up to $10\%$ of the total nodes), the RW population is able to recover after a finite time, following the elimination of all RWs by the Pac-Man nodes. 
This numerical result also highlights the robustness of the proposed algorithm against multiple Pac-Man nodes. 
As discussed in Sections~\ref{sec:Boundedness} and~\ref{sec:Convergence}, our theoretical framework can be straightforwardly extended to accommodate multiple Pac-Man nodes. 
For the remainder of this article, we focus on the case of a single Pac-Man node when presenting our results to improve the clarity of exposition.

\begin{figure}[htbp]
  \centering
  \begin{minipage}[b]{0.45\linewidth}
    \centering
    \includegraphics[width=\linewidth]{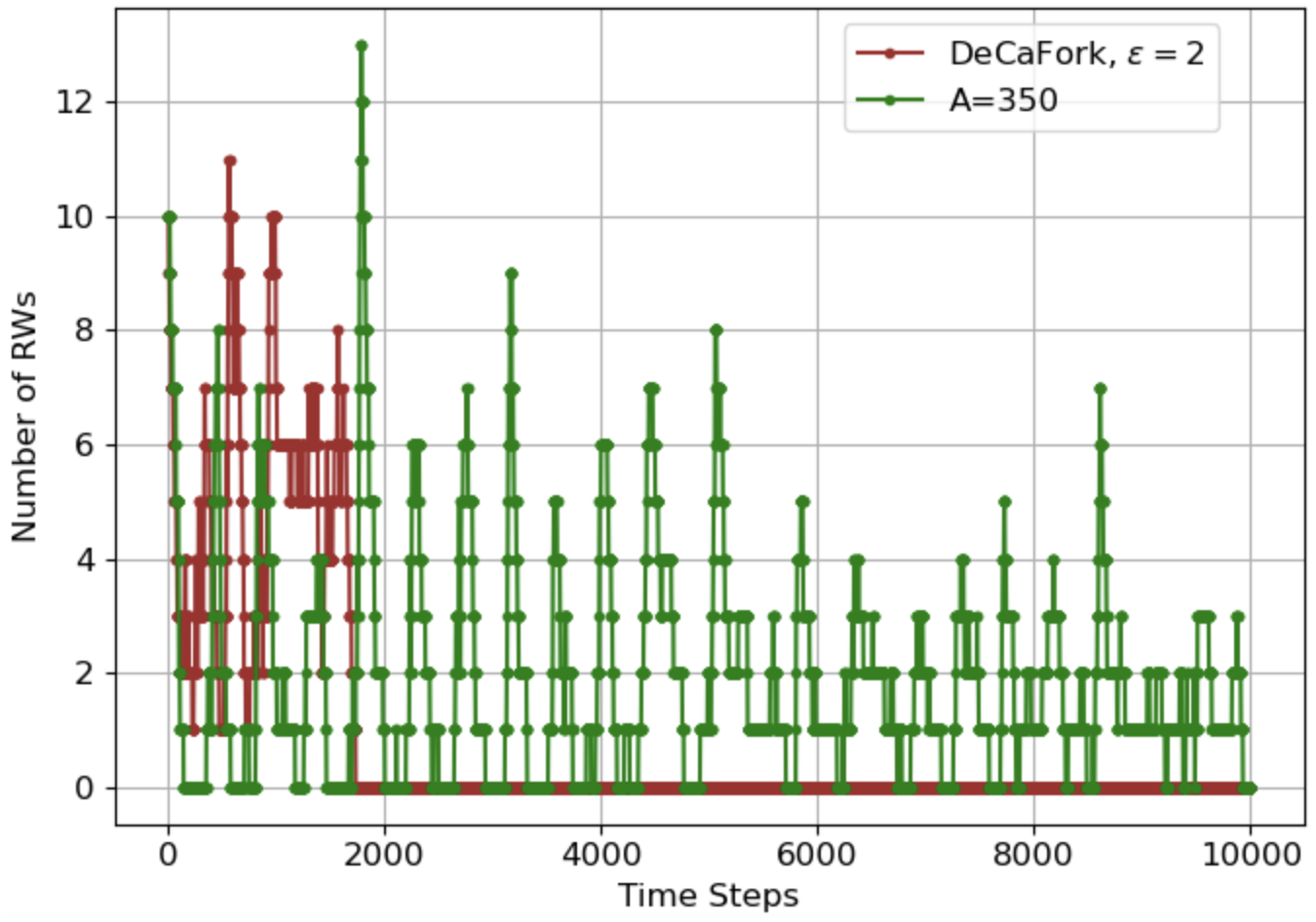}
    \caption*{(a) \cwl and \DeCa}
  \end{minipage}
  \hfill
  \begin{minipage}[b]{0.45\linewidth}
    \centering
    \includegraphics[width=\linewidth]{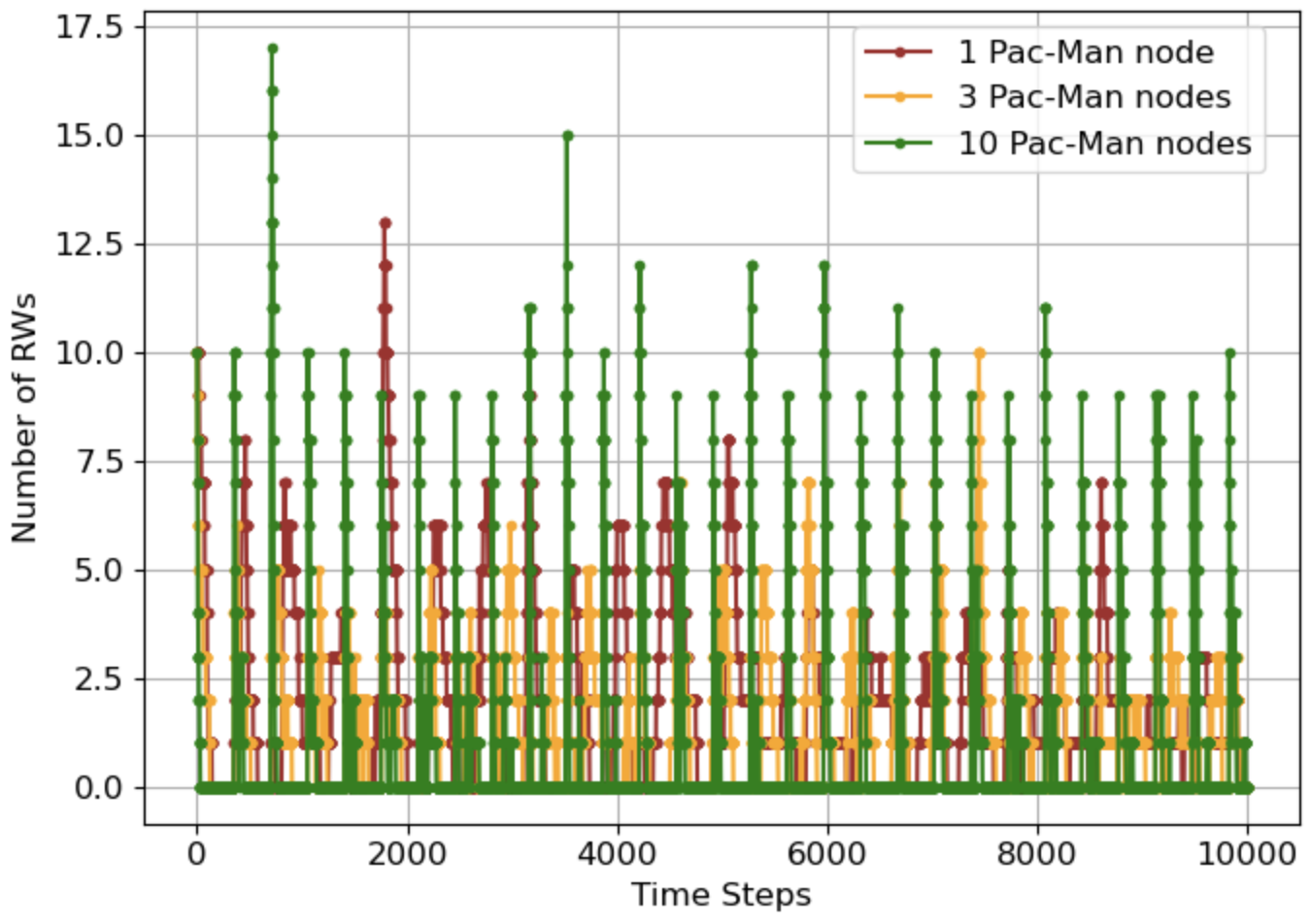}
    \caption*{(b) Multiple Pac-Man nodes}
  \end{minipage}
  \caption{(a): Comparison of the number of RWs under the \cwl and \DeCa algorithms in a complete graph. (b): Number of RWs under the \cwl algorithm in a complete graph with multiple Pac-Man nodes.}
  \label{fig:DeCacomplete}
\end{figure}

Finally, the algorithm design directly affects resource consumption, such as communication cost. When the creation threshold $A$ is large or the creation probability $q$ is small, new RWs are created less frequently.  
This reduces resource consumption but requires a longer time to complete computational tasks. 
Conversely, when the creation threshold $A$ is small or the creation probability $q$ is large, new RWs are created more frequently, resulting in faster task completion at the expense of higher communication cost.

%%%%%%%%%%%%%%%%%%%%%%%%%%%%%%%%%%%%%%%%%%%%%%%%%%
\subsection{Main Results}
%%%%%%%%%%%%%%%%%%%%%%%%%%%%%%%%%%%%%%%%%%%%%%%%%%
We summarize here our main results, which provide theoretical guarantees of no permanent extinction, no blowup, and convergence. 
The details of analysis are provided in Sections~\ref{sec:Boundedness} and \ref{sec:Convergence}.

\paragraph{No Permanent Extinction}
As stated in Remark~\ref{remark:noextinction}, by construction, the \cwl algorithm prevents the RW population from becoming permanently extinct with probability one; that is, \eqref{eq:NoPermanentExtinction} holds. %Under the \cwl algorithm, by time $t$, the RWs perform a number of update steps that is at least linear in $t$ (see Proposition~\ref{pro:NumIter}).

\paragraph{No Blowup}
With a Pac-Man node present in the graphs, the \cwl algorithm ensures that the RW population size remains bounded with probability one; that is, \eqref{eq:NoBlowup} holds, (see Theorem~\ref{thm:FiniteRWs}). Moreover, we provide a theoretical regime for selecting the creation probability $q$ such that the expected peak population,  $\sup_{t\ge0}\E[Z_t]$, becomes independent of the network size $N$ (see Theorem~\ref{thm:Peak}). 

%Before analyzing the objective of achieving no blowup, it is instructive to understand the limiting behavior of the population size $\set{Z_t}_{t\ge0}$ in the absence of a Pac-Man node.  Without the Pac-Man,  the \cwl algorithm leads to an explosion of the RW population with probability one (see Theorem~\ref{thm:InfinitNum}). Although the population diverges in this setting, we further show that the growth is slow: the expected time for the population to reach size~$m$ grows at least exponentially in~$m$ (see Proposition~\ref{pro:ExpectedTimeM}). Based on Theorems~\ref{thm:InfinitNum} and~\ref{thm:FiniteRWs}, in the absence of a Pac-Man node, one can designate an arbitrary node to act as a Pac-Man by terminating any incoming RW. This artificial Pac-Man prevents the RW population from exploding. 

\paragraph{Convergence} 
Under the \cwl mechanism, RW-SGD is shown to converge; that is, \eqref{eq:Effectiveness} holds 
(see Theorem~\ref{thm:Effectiveness}).
In addition, we show that it converges to a biased optimum, and we characterize the error between this skewed optimizer and the true optimizer, $\|\tilde{\bf x}^\star - {\bf x}^\star\|$ (see Proposition~\ref{pro:Bounds}).

\section{Preliminaries: Notations, Definitions, and Assumptions}\label{sec:assumptions}
In this section, we introduce the notation, definitions, and assumptions used in the system model and its transition dynamics.

\begin{definition}[Communication topologies and RWs]\label{defn:GenGph}
A communication topology is defined as a finite directed graph $\cG \triangleq (\cV, E)$ with the set of nodes $\cV = [N]$ and the set of edges $E \subseteq \binom{\cV}{2}$. 
Each RW $X_j:\Omega\to\cV^{\Z_+}$ on this graph is assumed to be \iid and can be defined by the common transition probability matrix $P:\cV\to \cM(\cV)$, 
where the probability of transition from node $u$ to node $v$ in one time step at time $t\in\Z_+$, is 
\begin{align*}
P_{uv} \triangleq \Prb\left(X_j(t+1)= v \,\middle|\, X_j(t)=u\right).
\end{align*}
We will call the RW on graph $\cG$ aperiodic if transition probability matrix $P$ is aperiodic. 
\end{definition}
\begin{remark}
Without loss of generality, we assume that $P_{uv} > 0$ for all nodes $v$ connected to node $u$ in graph $\cG$.  
Therefore, transition matrix $P$ is irreducible iff graph $\cG$ is connected. 
\end{remark}

\begin{definition}[Timing conventions]\label{defn:Timing} 
We define the timing conventions of the system as follows:
\begin{compactenum}[(1)]
\item At the beginning of time slot $t$, let $\cZ_t$ denote the set of indices of active RWs. 
Each active RW $j\in\cZ_t$ is associated with a birth time $\theta_j\ge0$ and an initial location $u_j \triangleq X_j(\theta_j)\in\cV$. 
Specifically, let $Z_t=\abs{\cZ_t}$.    
\item At the end of time slot $t$, each active RW $j\in\cZ_t$ moves from its current location to a randomly selected neighbor in the next time slot. 
We denote by $X_j(t+1)$ the location of RW $j$ after this movement at time $t+1$. 
\item Upon arrival at location $X_j(t+1)$ at time $t+1$, termination operation at node $1$ and creation operations at other nodes is performed (if applicable).  
\end{compactenum}
\end{definition}

Under the model assumption, the Pac-Man node is fixed at location $1$. We next present a formal definition of the system including the Pac-Man node. 

\begin{definition}\label{defn:PacMan}
Consider a communication topology $\cG = (\cV, E)$, as defined in Definition~\ref{defn:GenGph}, where node $1$ acts as the Pac-Man, such that if a RW visits node $1$, it is sent to a death node $0$\footnote{The death node is a virtual node used only for a clearer presentation.} with a termination probability $\zeta\in(0, 1]$. 
That is, we augment graph $\cG$ to $\cG^\prime = (\cV^\prime, E^\prime)$ where $\cV^\prime \triangleq \cV\cup\set{0}$ and $E^\prime \triangleq E \cup \set{(1,0)}$. 
In the presence of this Pac-Man, the original RW transition probability matrix $P$ now changes to $P^\prime$, where for each state $u,v \in \cV^\prime$
\begin{equation}
P^\prime_{uv} \triangleq 
\begin{cases}
P_{uv}, & u \in\cB, v \in \cV,\\
(1-\zeta)P_{uv}, & u =1, v \in \cV,\\
\zeta,& u = 1, v=0,\\
1,& u = 0, v = 0\\
0,&u=0, v\in\cV.
\end{cases}
\end{equation}
\end{definition}
Based on Definition~\ref{defn:PacMan}, the transition matrix $P'$ can be characterized by two cases:
$\zeta=1$ and $0<\zeta<1$. When $\zeta=1$, the Pac-Man node becomes an absorbing state, whereas when $0<\zeta<1$, the Pac-Man node is non-absorbing. In particular,
\begin{compactenum}[(1)]
\item When $\zeta=1$, the Pac-Man node eats every incoming RW, and the local data at the Pac-Man node cannot be utilized, making it an absorbing state. 
From Definition \ref{defn:PacMan}, the death node $0$ is also absorbing. 
Therefore, we merge the Pac-Man node and the death node into a single absorbing one, which we continue to denote as node $1$. 
The corresponding transition matrix $P^\prime$ becomes 
\begin{align}\label{eq:NewTransMatCase1}
P^\prime=\scalebox{0.8}{
$\begin{bmatrix}
1  & 0&\cdots& 0 \\
P_{21}&P_{22}&\cdots&P_{2N}\\
\vdots&\vdots&\vdots&\vdots\\
P_{N1}&P_{N2}&\cdots&P_{NN}
\end{bmatrix}$}
\triangleq\begin{bmatrix}
1& {\bf 0}_{1\times (N-1)}\\
R^{(1)}&Q^{(1)}
\end{bmatrix}.
\end{align}
We observe that $Q^{(1)} \in [0,1]^{\cB\times\cB}$ is sub-stochastic matrix where $Q^{(1)}_{uv} = P_{uv}$ for each $u,v \in \cB$, and $R^{(1)} \in [0,1]^{\cB\times 1}$ column vector where $R^{(1)}_u = P_{u1}$ for each benign node $u \in \cB$. 
\item When $0<\zeta<1$,  the local data at the Pac-Man node cannot be reliably utilized, because any RW visiting it is terminated with probability strictly less than $1$.
In this case, only the death node $0$ is absorbing. 
For clarity, we denote the location of the death node as $0$. 
According to Definition \ref{defn:PacMan} and denoting $\bar\zeta \triangleq 1-\zeta$, the corresponding transition matrix $P^\prime$ becomes
\begin{align}\label{eq:NewTransMatCase2}
P^\prime=&\scalebox{0.8}{
$\begin{bmatrix}
1  &0 &\cdots& 0 \\
\zeta &\bar\zeta P_{11}&\cdots&\bar\zeta P_{1N}\\
0&P_{21}&\cdots&P_{2N}\\
\vdots&\vdots&\vdots&\vdots\\
0&P_{N1}&\cdots&P_{NN}
\end{bmatrix}$}\triangleq\begin{bmatrix}
1& {\bf 0}_{1\times N}\\
R^{(\zeta)} & Q^{(\zeta)}
\end{bmatrix}.
\end{align}
We observe that 
\begin{xalignat*}{2}
&R^{(\zeta)} \triangleq \begin{bmatrix}\zeta\\  {\bf 0}_{(N-1) \times 1} \end{bmatrix},&&Q^{(\zeta)} \triangleq \begin{bmatrix}\bar\zeta P_{11}& \bar\zeta S\\ R^{(1)} & Q^{(1)}\end{bmatrix}.
\end{xalignat*} 
The row vector $S \in [0,1]^{1\times \cB}$ such that $S_{1v} = P_{1v}$ for each benign node $v \in \cB$. 
\end{compactenum}
\begin{remark}\label{remark:MultiplePacMan}
If $\zeta=1$, the analysis extends straightforwardly to the setting with multiple Pac-Man nodes. In this case, all Pac-Man nodes can be treated as a single ``super node'', which acts as an absorbing state of the system. The resulting analysis is essentially identical to that of the single Pac-Man node case. 
If $\zeta<1$, the analysis becomes more involved, since the transition probability to the absorbing state (i.e., the death node $0$) depends on the identity of the Pac-Man node through (a) their corresponding termination probabilities and (b) their connectivity to the benign nodes.
Nevertheless, the analysis can still be carried out within the theoretical framework proposed for the single Pac-Man node setting, although the analysis would be more tedious for a larger number of Pac-Man nodes. 
\end{remark}
\begin{definition}[Robustly Connected Graph]\label{defn:StrongConnect}
A graph $\mathcal{G}$ is robustly connected if
\begin{compactenum}[(i)]
\item every pair of benign nodes in $\cB$ is connected by a path that avoids the Pac-Man node, and
\item the Markov chain corresponding to each active RW is aperiodic.
\end{compactenum}
\end{definition}
\begin{remark}
According to Definition~\ref{defn:StrongConnect}, a robustly connected graph cannot be partitioned into two disjoint components by Pac-Man. In addition, in a robustly connected graph, the Markov chain corresponding to each active RW is irreducible and aperiodic.
\end{remark}

\begin{definition}\label{def:matrixQ}
For $\zeta \in (0,1]$, let $Q^{(\zeta)}$ be the matrix defined in \eqref{eq:NewTransMatCase1} and \eqref{eq:NewTransMatCase2}.  
Let $\alpha^{(\zeta)}$ denote the (unique) maximum eigenvalue of $Q^{(\zeta)}$\footnote{
We assume a unique maximum eigenvalue for simplicity.  
If $Q^{(\zeta)}$ has multiple dominant eigenvalues, the arguments can be extended using a standard Jordan decomposition.
},  
and let %$\phi^{(\zeta)}$ and 
$\nu^{(\zeta)}$ denote the associated %right and 
left normalized positive eigenvectors with unit sum.  %respectively.
\end{definition}

\begin{assumption}\label{assu:IndependentRW}
Each RW has an \iid evolution on this graph with transition probability matrix $P$, conditioned on their initial locations. 
\end{assumption}

\begin{assumption}\label{assu:robustlyconnected}
The graph $\cG$ defined in Definition~\ref{defn:GenGph} is a robustly connected graph.
\end{assumption}
Next, we adopt the common assumptions used in standard distributed optimization problems, as follows.
\begin{assumption}\label{assu:StandCons}
Each local function $f_u({\bf x})$ in \eqref{eq:goal0} with $u\in\cV$ is $\mu$-strongly convex and $L$-smooth.
\end{assumption}

\begin{assumption}\label{assu:Boundedgrad}
Bounded norm of the local gradient at the global optimum ${\bx}^\star$, i.e. 
$
\sup_{u \in \cV}\norm{\nabla f_u({\bx}^\star)}^2\le \sigma^2,
$
where ${\bx}^\star$ is the minimizer of \eqref{eq:goal0}.
\end{assumption}
In the remainder of this article, we assume that Assumptions \ref{assu:IndependentRW} -- \ref{assu:Boundedgrad} hold.

For clarity of presentation, the theoretical results in Section~\ref{sec:Boundedness} and Section~\ref{sec:Convergence} are presented for a single Pac-Man node, and Theorem~\ref{thm:Peak} is further specialized to a complete-graph setting. These assumptions are made only to simplify the notation and analysis. The results extend to multiple Pac-Man nodes on a general connected graph with minor modifications, provided that the robustly connected graph condition in Definition~4 and Assumption~2 holds; that is, removing the Pac-Man nodes does not disconnect the network. Further details are provided in Appendix~G.

\section{Population Boundedness with Pac-Man}\label{sec:Boundedness}
In this section, we introduce one Pac-Man node, and we examine the long-term behavior of the RW population process $\set{Z_t}_{t\in\N}$.

In the following theorem, we show that, regardless of the creation thresholds $\set{A_u: u\in\cB}$ chosen by the benign nodes, the RW population remains almost surely bounded over time. This boundedness is crucial for ensuring algorithmic stability (in terms of population) and preventing the network from being flooded with RWs.

\begin{theorem}\label{thm:FiniteRWs}
On any finite graph $\cG^\prime = (\cV^\prime, E^\prime)$ of Definition \ref{defn:PacMan}, with $z_0\ge 1$, $A_u \ge 1$, $q\le 1$, and $\zeta \in (0,1]$, the \cwl algorithm ensures that 
$
\limsup_{t\to\infty}Z_t < \infty
$
almost surely\footnote{We abstract away network delays and other lower-layer communication issues to focus on the proposed random-walk-based algorithm and its resilience to termination. These issues are assumed to be handled by the underlying network protocols.}.
\end{theorem}
\begin{proof}
(Roadmap) Let $\cF_t$ denote the history of the evolution of all RWs until time $t$.
First, we prove that the $d$-step drift of $Z_t$ is bounded as $\E[Z_{t+d}-Z_t\mid\cF_t] \le  - c\zeta Z_t + (N-1)d$, where $d$ is the graph diameter and $c$ is the smallest probability of reaching the Pac-Man within $d$ steps from any node $u \in \cB$. 
The formal definitions of $d$ and $c$ are provided in Definition~9 in Appendix~A.
This inequality implies that the expected drift in the number of RWs after $d$ steps is negative when $Z_t$ is large.   
%In particular, the negative drift term $-c\zeta Z_t$ dominates, driving $\E[Z_{t+d}]$ to a small value---even below zero---thus 
This prevents $Z_t$ from diverging.  
The full proof is in Appendix A. A discussion of the extension to multiple Pac-Man nodes is provided in Appendix~G-A.
\end{proof}

A natural follow-up question is whether we can characterize the maximum number of random walks at across all times $t$.   
As such, we next analyze the peak of expected number of RWs, defined as
\begin{align}\label{eq:PeakPopulation}
\bar{Z}^\star = \sup_{t\ge0}\E[Z_t].
\end{align}
We focus on the special case in which the creation thresholds satisfy $A_u=1$ for all $u\in\cB$. Under this choice, a new RW is created immediately whenever a benign node is not visited by a RW, thereby making RW creation as frequent as possible. Consequently, this setting corresponds to the worst-case scenario for RW proliferation. The expected peak number derived under this setting therefore provides a universal upper bound for all threshold choices with $A_u>1$. Since larger thresholds lead to less frequent RW creation, the resulting bound may be conservative for general values of $A_u$.

\begin{theorem}\label{thm:Peak}
 Consider a complete graph with $N$ nodes and a \textit{uniform} target sampling distribution $\pi=(\frac{1}{N}, \frac{1}{N},\cdots,\frac{1}{N})$. Suppose there are $z_0$ initial RWs, each node has a creation threshold $A_u = 1$, and the creation probability is $q \in (0, 1]$.  The expected peak population under the \cwl algorithm satisfies
\begin{align*}
\bar{Z}^\star=O\Big(\frac{q N^2}{\zeta}\Big),
\end{align*}
where $\bar{Z}^\star$ is defined in \eqref{eq:PeakPopulation}.
\end{theorem}
\begin{proof}
(Roadmap) By unrolling the recursion derived from the transition probabilities $P_{uv}$ and the termination probability $\zeta$,  we obtain the following upper bound for every $t\ge0$:
\begin{align*}
\E[Z_t]\le \frac{q}{\zeta}(N-1)N + (1-\frac{\zeta}{N})^t\left(z_0-\frac{q}{\zeta}(N-1)N\right).
\end{align*}
Consequently,
\begin{align}\label{eq:barZUB}
\bar{Z}^\star\le\max\set{z_0, \frac{q}{\zeta}N^2}.
\end{align}
The full proof is given in Appendix~B. A discussion of the extension to multiple Pac-Man nodes and general connected graphs is provided in Appendix~G-B.
\end{proof}

\begin{cor}\label{cor:Peak}
Under the same setting in Theorem~\ref{thm:Peak}, under the \cwl algorithm, for any $0\le\alpha\le2$, if $z_0=\Theta(1)$ and $\zeta=\Theta(1)$, $q=N^{-\alpha}$, then $\bar{Z}^\star = O(N^{2-\alpha})$.
\end{cor}
\begin{proof}
By substituting the appropriate values of $(z_0, \zeta, q)$ into \eqref{eq:barZUB}, we obtain the desired bounds.
\end{proof}

From a practical perspective, Corollary~\ref{cor:Peak} provides an explicit rule for selecting the creation probability $q$ according to the desired RW population. In many practical settings, the initial number of RWs is typically $z_0=\Theta(1)$, while the termination probability $\zeta$, determined by the Pac-Man, satisfies $\zeta\in(0,1]$ and is usually not too small, implying $\zeta=\Theta(1)$. Under these conditions, if the system designer wishes to keep the expected peak RW population at the order of $O(N^{2-\alpha})$, where $0\le\alpha\le2$, then $q$ can be selected as $q=N^{-\alpha}$. For example, $q=N^{-2}$ keeps the expected peak RW population bounded independently of $N$, $q=N^{-1}$ allows it to grow linearly with $N$, and $q=\Theta(1)$ may result in quadratic growth. In implementation, one may use
$q=\min\left\{1,rN^{-\alpha}\right\}$,
where $\alpha$ specifies the desired population scaling and the constant $r>0$ can be tuned according to the available resource budget.

A similar population-control principle also applies to other connected graphs. Although the required scaling of the creation probability may differ from $N^{-\alpha}$, the expected peak number of RWs, $\bar{Z}^\star$, can still be controlled through an appropriate choice of $q$. Therefore, rather than selecting the creation probability heuristically, one may first specify a desired growth rate for the RW population and then choose $q$ accordingly.

%%%%%%%%%%%%%%%%%%%%%%%%%%%%%%%%%%%%%%%%%%%%%%%%%%
\section{Convergence of the Learning Algorithm}\label{sec:Convergence}
%%%%%%%%%%%%%%%%%%%%%%%%%%%%%%%%%%%%%%%%%%%%%%%%%%
In this section, we study how the \cwl algorithm affects the convergence of RW-SGD in the presence of a Pac-Man node. 
Since the Pac-Man node terminates incoming RWs with probability $\zeta$, the local data at the Pac-Man node is only partially exploited in the learning process. This raises two questions: (i) whether RW-SGD still converges, and (ii) if it does, whether the resulting limit point $\tilde{\bf x}^\star$ (defined in \eqref{eq:Effectiveness}) coincides with the true optimizer ${\bf x}^\star$.

It is important to note that each single RW will eventually be terminated by the Pac-Man node with probability one. To address the two questions above, rather than relying on a single RW, we need to analyze the behavior of a {\it chain} of RWs. Specifically, when a benign node has not been visited for too long, it creates a new RW by copying the model of the most recently visiting RW. We refer to this most recent RW as the \textit{parent} and newly created RW as its \textit{child}.

\begin{definition}\label{defn:ChainRWs}
(Chain of RWs)  
Start with an initial RW $j_0\in[z_0]$. At each $s\ge 0$, if $j_s$ has at least one child, let $j_{s+1}$ be a uniformly chosen child of $j_s$; otherwise, set $j_{s'}=-1$ for all $s'\ge s$ to mark termination. We refer to $\set{j_s: s\in\N}$ as a chain of RWs.
\end{definition}

\begin{figure}[htbp]
\centering
\includegraphics[width=0.7\linewidth]{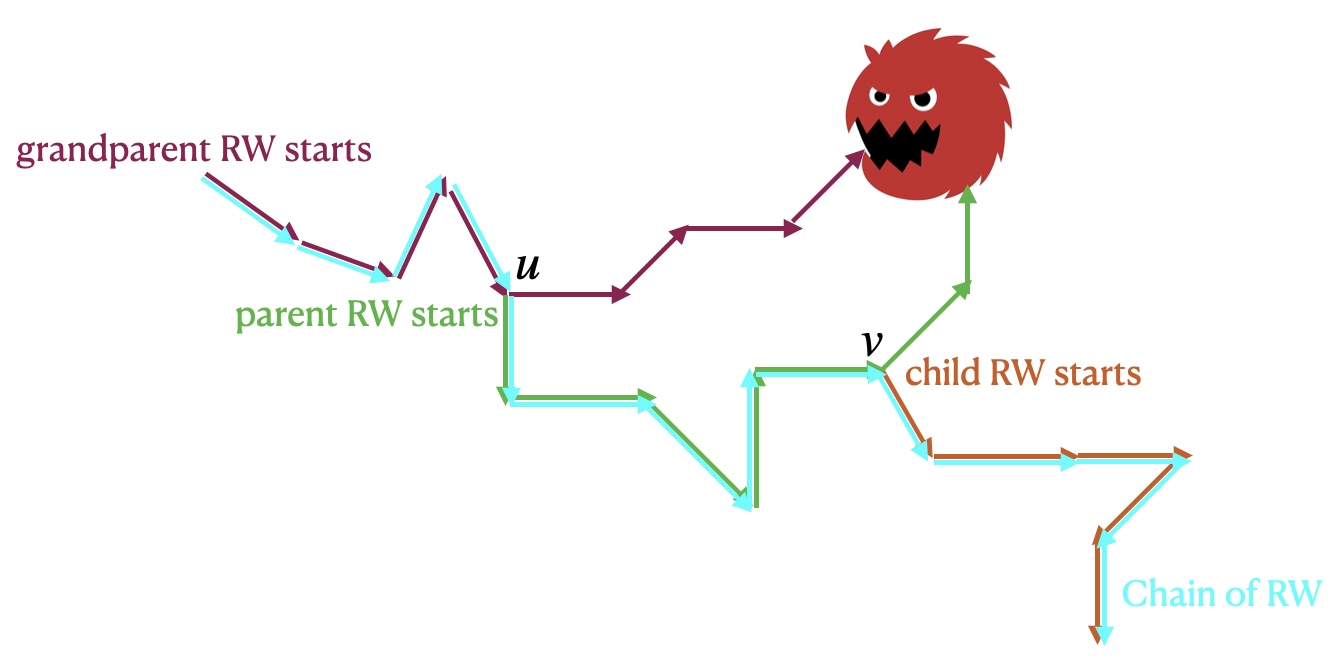}
\caption{An example of a chain of RWs. When a benign node $u$ has not been visited for $A_u$ consecutive time slots, it creates a new RW (green). Similarly, when node $v$ has not been visited for $A_v$ time slots, it creates another new RW (orange). The blue trajectory connects these RWs and forms a chain.}
\label{fig:chain}
\end{figure}

An illustrative example of a chain of RWs is shown in Fig.~\ref{fig:chain}. Each parent-child pair is separated by a random waiting period that is at least as large as the corresponding creation threshold, due to the creation probability $q$. However, these waiting periods do not affect learning performance, as they simply correspond to \textit{idle} times during which no updates are applied at the node. Thus, we can \textbf{conceptually condense} all waiting periods and treat the entire chain as an \textit{effective single} RW that inherits all model updates along the parent–child lineage. We therefore study the limiting behavior of such chains. Since every benign node can create new RWs, at least one infinite parent-child chain must almost surely exist.

In the presence of the Pac-Man node, a chain of RWs does not follow the original transition matrices $P$ (or $P'$). The effective transition probability matrix now depends explicitly on the termination probability $\zeta$. We denote this effective transition matrix by $P^{(\zeta)}_{\text{chain}}$, and its stationary distribution by $\pi^{(\zeta)}_{\text{chain}}$, where the explicit forms of $P^{(\zeta)}_{\text{chain}}$ and $\pi^{(\zeta)}_{\text{chain}}$ are given in \eqref{eq:pichain} and \eqref{eq:UniqueP}, respectively.

\begin{theorem}\label{thm:Effectiveness}
Under the \cwl mechanism, RW-SGD converges to the minimizer of the following problem:
\begin{align}\label{eq:SurOpt}
\min_{{\bf x}\in\R^m}\tilde{f}({\bf x}) = \min_{{\bf x}\in\R^m} \,\E_{u\sim\pi^{(\zeta)}_{\text{chain}}}[f_u({\bf x})],
\end{align}	
where
\begin{align}\label{eq:pichain}
\pi^{(\zeta)}_{\text{chain}} = \left\{
\begin{aligned}
&[0, \nu^{(1)}]&&\zeta=1,\\
&\nu^{(\zeta)}&&\zeta\in(0,1),
\end{aligned}
\right.
\end{align}
and $\nu^{(\zeta)}$ is defined in Definition~\ref{def:matrixQ}.
\end{theorem}
\begin{proof}
(Roadmap)  We first show that, as long as the initial RW $j_0$ remains active, the distribution of the chain coincides with that of $j_0$. 
Next, we demonstrate that this distribution converges to the quasi-stationary distribution \cite{collet2012quasi} of RW $j_0$, which can be computed following the approach in \cite{darroch1965quasi}. 
Consequently, under the \cwl algorithm, RW-SGD converges to the surrogate optimization problem \eqref{eq:SurOpt}. The full proof is given in Appendix~C. A discussion of the extension to multiple Pac-Man nodes is provided in Appendix~G-C.
\end{proof}

Now we derive the explicit form of the effective transition probability matrix $P^{(\zeta)}_{\text{chain}}$, which will be used to quantify the limiting model from the true global optimum $\bx^\star$. As stated before, without loss of generality, we \textit{condense the waiting time} between the each pair of parent and child (i.e., the flat segments in Fig.~\ref{fig:Convergence}). Then, $P_\text{chain}^{(\zeta)}$ takes the following form (the proof is provided in Appendix~D):
\begin{align}\label{eq:UniqueP}
\left[P^{(\zeta)}_{\text{chain}}\right]_{uv}=
\frac{Q^{(\zeta)}_{uv}}{\sum_{v}Q^{(\zeta)}_{uv}}
\end{align}
and $Q^{(\zeta)}$ is defined in \eqref{eq:NewTransMatCase1} and \eqref{eq:NewTransMatCase2}.

In the standard RW-SGD algorithm, it is well known that if the stepsize $\eta_t$ decreases with the number of local-update iterations $t$ and tends to $0$, the algorithm converges to a deterministic point; if $\eta_t$ remains constant, it converges to a random variable \cite{10.5555/3618408.3618786, doi:10.1137/08073038X, 10.5555/3327546.3327656}.
\begin{proposition}[Shift of Optima]\label{pro:Bounds}
Define $P_{\mathrm{chain}}^{(\zeta)}$ as in \eqref{eq:UniqueP}, and let $\gamma_\text{chain}^{(\zeta)}$ be its spectral gap. Denote by $\|\cdot\|_{\text{TV}}$ be the total variation distance, and let ${\bx}_0$ be the starting point. Let $\pi_{\text{chian}}^{(\zeta)}$ be the stationary distribution defined in \eqref{eq:pichain}. Under the \cwl mechanism:
\begin{compactenum}[(1)]
\item If the step size $\eta_t$ decreases with the number of iterations and tends to $0$, then RW-SGD converges to the minimizer $\tilde{\bx}^\star$ of \eqref{eq:SurOpt}. Moreover, 
\begin{align*}
\frac{1}{L}\|\nabla f(\tilde{\bx}^\star\big)\|\le	\|\tilde{\bx}^\star-{\bx}^\star\|\le \frac{1}{\mu}\|\nabla f(\tilde{\bx}^\star)\|.
\end{align*}
\item If the step size keeps constant, i.e., $\eta_t=\eta < \tfrac{1}{L}$, then, the expected error satisfies
\begin{align*}
\lim_{T\to\infty}\E\left[\|\tilde{\bx}_T-{\bx}^\star\|\right]\le \frac{\eta L\sigma^2}{\gamma_\text{chain}^{(\zeta)}\mu^2} + \frac{\left\|\tilde{\nu}^{(\zeta)} - \pi\right\|_{TV}^2\sigma^2L}{\mu^3}.
\end{align*}	
\end{compactenum}
\end{proposition}
\begin{proof}
These inequalities are derived from existing results, and the full proof is given in Appendix~E. A discussion of the extension to multiple Pac-Man nodes is provided in Appendix~G-D.
\end{proof}

\begin{remark}
An interesting problem is whether the sampling bias caused by Pac-Man attacks can be reduced using importance sampling \cite{KahnMarshall1953} and update reweighting \cite{HorvitzThompson1952}. However, in a fully decentralized network, the location of the Pac-Man node and the global transition probability matrix are unknown. Standard reweighting factors are therefore unavailable, while estimating them from local observations may introduce substantial error and amplify adversarial updates. Developing reliable bias-correction methods for this setting is left for future work.
\end{remark}

Since temporary extinctions may cause the RW population to become inactive, not every physical time slot contributes to learning. To characterize its real-time convergence behavior, we focus on the number of local-update iterations completed by time $t$. For an infinite chain of RWs, let $\itr_t$ denote this number. The ratio $\frac{\E[\itr_t]}{t}$ measures the expected fraction of time steps that effectively contribute to learning, and therefore partially characterizes the convergence speed in real time. When $\itr_t=t$, every time step produces a local update, so the algorithm progresses at the same rate as classical RW-SGD. When $\itr_t<t$, some time steps are lost due to RW termination and regeneration, and the real-time convergence is slowed down is slowed down according to the effective update ratio.

To characterize the communication overhead, let $C_t$ denote the total number of transmissions on the graph up to time $t$. Since each local update is associated with one RW transmission, $\itr_t$ also counts the number of transmissions generated by a chain. Note that $C_t$ counts all transmissions generated by all active RWs on the graph. Therefore, we have $\itr_t \le C_t$.

\begin{proposition}[Convergence speed \& communication overhead]\label{pro:NumIter}
Under the same setting in Theorem~\ref{thm:Peak}, with $A_u=A\ge 1$, the \cwl algorithm guarantees that the expected number of local-update iterations up to time $t$ satisfies
\begin{align}
& 1\ge \frac{1}{t}\E[\itr_t],\,\, \lim_{t\to\infty}\frac{1}{t}\E[\itr_t] \ge \frac{\frac{N}{\zeta}}{\frac{N}{\zeta} + A - 1 + \frac{1}{q}},\label{eq:TradeoffRecovery}\\
& \lim_{t\to\infty}\frac{1}{t}\E[C_t] \ge \frac{\frac{N}{\zeta}}{\frac{N}{\zeta} + A - 1 + \frac{1}{q}}.\label{eq:CommunicationOverhead}
\end{align}
\end{proposition}
\begin{proof}
(Roadmap) The upper bound $\E[\itr_t]\le t$ is immediate, since at most one learning iteration can occur per time slot. To obtain the lower bound, we analyze a virtual worst-case system in which at most one RW can be active RW at any time. Whenever a new RW would be created while another is active, the creation is suppressed; when no RW is active and a creation trial succeeds, a single RW is started. The expected number of local-update iterations in this restricted system therefore provides a lower bound on $\E[\itr_t]$. Since $\itr_t \le C_t$, a lower bound on $\mathbb{E}[C_t]$ follows immediately. The full proof is provided in Appendix~F.
\end{proof}

\begin{remark}[Communication--convergence tradeoff]\label{remark:Tradeoff}
The lower bounds in \eqref{eq:TradeoffRecovery} and \eqref{eq:CommunicationOverhead} reflect the role of the regeneration parameters and partially reveal a communication--convergence tradeoff. After a temporary extinction, RW regeneration is governed by the threshold $A$ and the creation probability $q$. A larger $A$ or a smaller $q$ makes regeneration more conservative, which lengthens the expected inactive period and decreases the lower bounds on both the effective update ratio $\E[\itr_t]/t$ and the normalized communication cost $\E[C_t]/t$. Conversely, a smaller $A$ or a larger $q$ leads to more aggressive regeneration, which increases the fraction of time slots used for learning but also increases the communication cost. Thus, aggressive regeneration improves real-time learning progress at the expense of higher communication overhead, whereas conservative regeneration reduces communication but slows down real-time convergence.
\end{remark}

\begin{remark}[Selection of the creation threshold $A$]
A complete design rule for choosing $A$ remains an open problem. The optimal choice of $A$ may depend on the interaction among RW termination, recovery time, communication cost, and real-time convergence progress. Thus, it is intractable to analyze the optimal $A$ directly. In practice, $A$ may be tuned adaptively by monitoring empirical observations. Such an online procedure does not give a closed-form optimal threshold, but it may provide an acceptable data-driven choice of $A$. We leave a rigorous adaptive selection rule for $A$ as future work.
\end{remark}

\begin{remark}
Equation~\eqref{eq:TradeoffRecovery} implies that although the extinction events slow down convergence, the delay (the length of horizontal segments) grows at most linearly with $t$. This means that RW-SGD performs at least a linear number of updates with respect to $t$, and therefore will eventually converge given a long time horizon.
\end{remark}

\section{Simulations}\label{sec: simulations}
In this section, we present simulation results to validate the theoretical findings in Sections~\ref{sec:Boundedness} and \ref{sec:Convergence}. 

\subsection{Simulation Setup}\label{subsec: Setup}
\paragraph{Graph Settings}
We consider $4$ connected graphs of $100$ nodes, with $1$ Pac-Man node and $99$ benign nodes, including: (i) a regular graph with degree $d=99$ (complete graph), (ii) a random regular graph with degree $d=8$ (expander graph \cite{expanderdegree}), (iii) a regular graph with degree $d=2$ (ring topology), and (iv) an \ER graph with the edge probability $p=0.1$. The target sampling distribution is set to be uniform, i.e., $\pi = (\frac{1}{100},\frac{1}{100},\cdots,\frac{1}{100})$, and the transition matrix $P$ is obtained by the Metropolis-Hastings algorithm \cite{doi:10.1137/08073038X}. We assume $A_u=A$ for $u\in\cB$. We set both the forking and termination probabilities to $1$, i.e., $q=\zeta=1$. The initial number of RWs $z_0=10$.

\paragraph{ Learning Settings} For the distributed learning problem defined in \eqref{eq:goal0}, we evaluate our approach on both synthetic and public benchmark datasets. 

\textbf{Synthetic dataset}. We consider a decentralized linear regression task, where each node $u$ minimizes a local mean squared error (MSE) loss of the form:
\begin{align*}
f_u({\bf x}) = ({\bf w}^T{\bf x}+b - y_u)^2,
\end{align*}
where ${\bf x}$ is the input feature, $y_u$ is the target label at node $u$, ${\bf w}$ is the weight vector, and $b$ is the bias. We assume each node holds only a single data point. This local objective $f_u({\bf x})$ is strongly convex and $L$-smooth.

\textbf{Public benchmark dataset}. We use the standard MNIST handwritten digit dataset \cite{deng2012mnist}. The dataset is evenly partitioned into $100$ disjoint subsets, with each node is assigned a unique subset. We consider both \iid and non-\iid data partitioning schemes for distributing data across network nodes. In the \iid case, the dataset is uniformly and independently divided into $100$ disjoint subsets, one for each node. In the non-\iid case, data heterogeneity is introduced by sampling from a Dirichlet distribution \cite{hsu2019measuring}: The concentration parameter $\alpha$ controls the degree of heterogeneity: as $\alpha \to 10$, the partitioning approaches the \iid case; as $\alpha \to 0$, the data becomes highly non-\iid Throughout this paper, we set $\alpha = 0.5$, which corresponds to a moderate level of non-\iid partitioning. Under both partitioning schemes, whenever a RW visits a node, it uniformly samples a mini-batch of size $B=256$ from the node's local data to perform a training update. Each node $u$ minimizes the empirical cross-entropy loss over its local dataset $\cD_u$:
\begin{align*}
f_u(w) = \frac{1}{|\cD_u|}\sum_{({\bf x}, y)\in\cD_u}\ell(w; {\bf x}, y),
\end{align*}
where $w$ denotes the model parameters, $\ell(w; {\bf x}, y)$ is the standard cross-entropy loss function, and $\cD_u$ is the local subset assigned to node $u$. We adopt the Adam optimizer from the PyTorch library ($\mathtt{torch.optim.Adam}$) for model training.

\paragraph{Baseline}
To the best of our knowledge, the algorithm most closely related to our work is the \DeCa algorithm \cite{selfdup, egger2024self}. \DeCa is a duplication-based algorithm, whereas our proposed algorithm is creation-based.\footnote{Another duplication-based algorithm introduced in \cite{selfdup, egger2024self} is the \MP algorithm, which has been less explored in the literature. We do not include the \MP algorithm as a baseline in this paper for two reasons: (i) the primary difference between \DeCa and \MP lies in their duplication mechanisms, and \DeCa has been shown to be more robust than \MP in \cite{selfdup, egger2024self}; (ii) when applied to decentralized learning settings, both algorithms exhibit similar convergence performance, as reported in \cite{AC}.}

Another baseline we consider is the classical \textsc{Gossip-based SGD} \cite{tsitsiklis2003distributed, pmlr-v119-koloskova20a}. In this scheme, each node transmits its locally updated model to all of its neighbors at every iteration, and model parameters are updated via neighborhood averaging. We incorporate the same adversarial setting as before: the Pac-Man node independently terminates each incoming model update with the termination probability $\zeta$. As a result, the Pac-Man node is unable to reliably incorporate all information from its neighbors, leading to biased and incomplete aggregation.

Fig.~\ref{fig:zetaregular} compares \cwl and \textsc{Gossip-based SGD} in terms of communication overhead, where the communication overhead is measured by the number of transmissions. The loss curve of \cwl (blue curve) decays much faster than that of \textsc{Gossip-based SGD} (orange curve), showing that \cwl achieves faster convergence under the same level of communication overhead. In this experiment, we set $\zeta=0.8<1$ and choose the creation threshold as $A=55$, and we observe that the number of RWs fluctuates around the initial value $z_0=10$. The cumulative communication overhead up to time $t$ is calculated by $\sum_{\tau\leq t} Z_\tau$\footnote{In addition, as discussed in Fig.~\ref{fig:peakno}, the creation probability $q$ can be appropriately chosen to keep the RW population small, thereby reducing the communication overhead.}. In contrast, for \textsc{Gossip-based SGD}, the number of communications (on a complete graph) at each time step is fixed and equals $N(N-1)$. Thus, the horizontal axis in Fig.~\ref{fig:zetaregular} represents the cumulative communication overhead measured in transmissions. The results suggest that \textsc{Gossip-based SGD} requires substantially more transmissions to achieve the same loss level, whereas \cwl can preserve and propagate informative updates more efficiently through active RWs.

Due to its slow convergence, we find \textsc{Gossip-based SGD} to be an unsuitable baseline in the Pac-Man setting, and we therefore exclude it from the remaining simulation results.

\begin{figure}[htbp]
\centering
\includegraphics[width=0.7\linewidth]{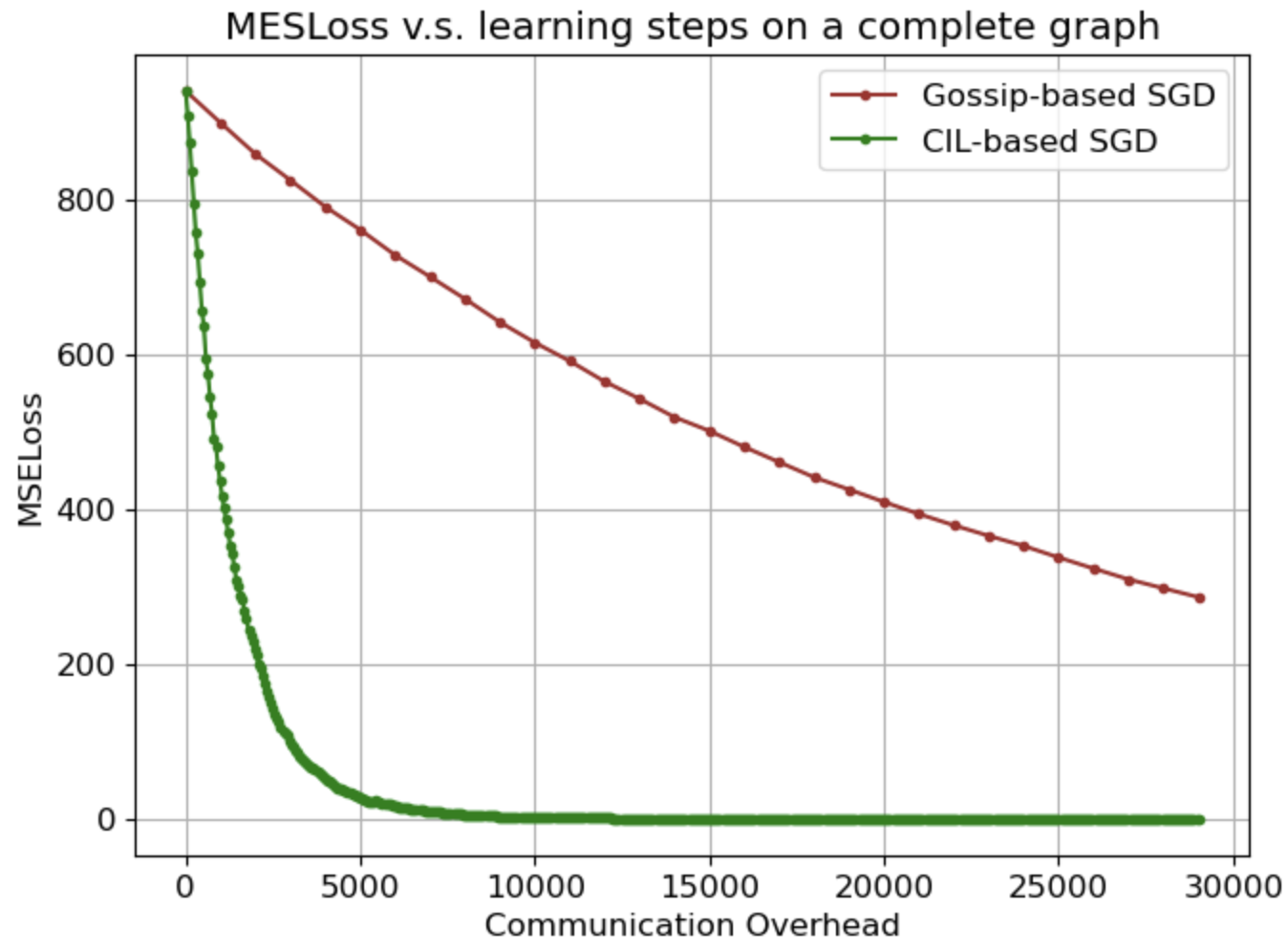}
\caption{Loss v.s. learning steps on a complete graph: comparison between \cwl-based and gossip-based SGD.}
\label{fig:zetaregular}
\end{figure}

\subsection{Boundedness}\label{subsec: Boundedness}

In the beginning of the subsection, we validate Theorem~\ref{thm:FiniteRWs}, which asserts that the number of RWs remains bounded under the \cwl algorithm. As shown in Fig.~\ref{fig:Boundedness}, the RW population process $\set{Z_t}_{t\in\Z}$ remains bounded on all tested graphs. The $y$-axis represents the number of RWs, and the $x$-axis denotes the time steps. We present a single realization of $\set{Z_t}_{t\in\Z}$. On each graph, when the threshold $A$ is small (noting that the definitions of ``small'' and ``large'' vary by graph), the RW population fluctuates at a relatively high level region. In contrast, for large $A$, the RW population may temporarily go extinct but subsequently recovers. In all cases, the population remains bounded, confirming the theoretical analysis.

\begin{figure}[htbp]
  \centering
  \begin{minipage}[b]{0.45\linewidth}
    \centering
    \includegraphics[width=\linewidth]{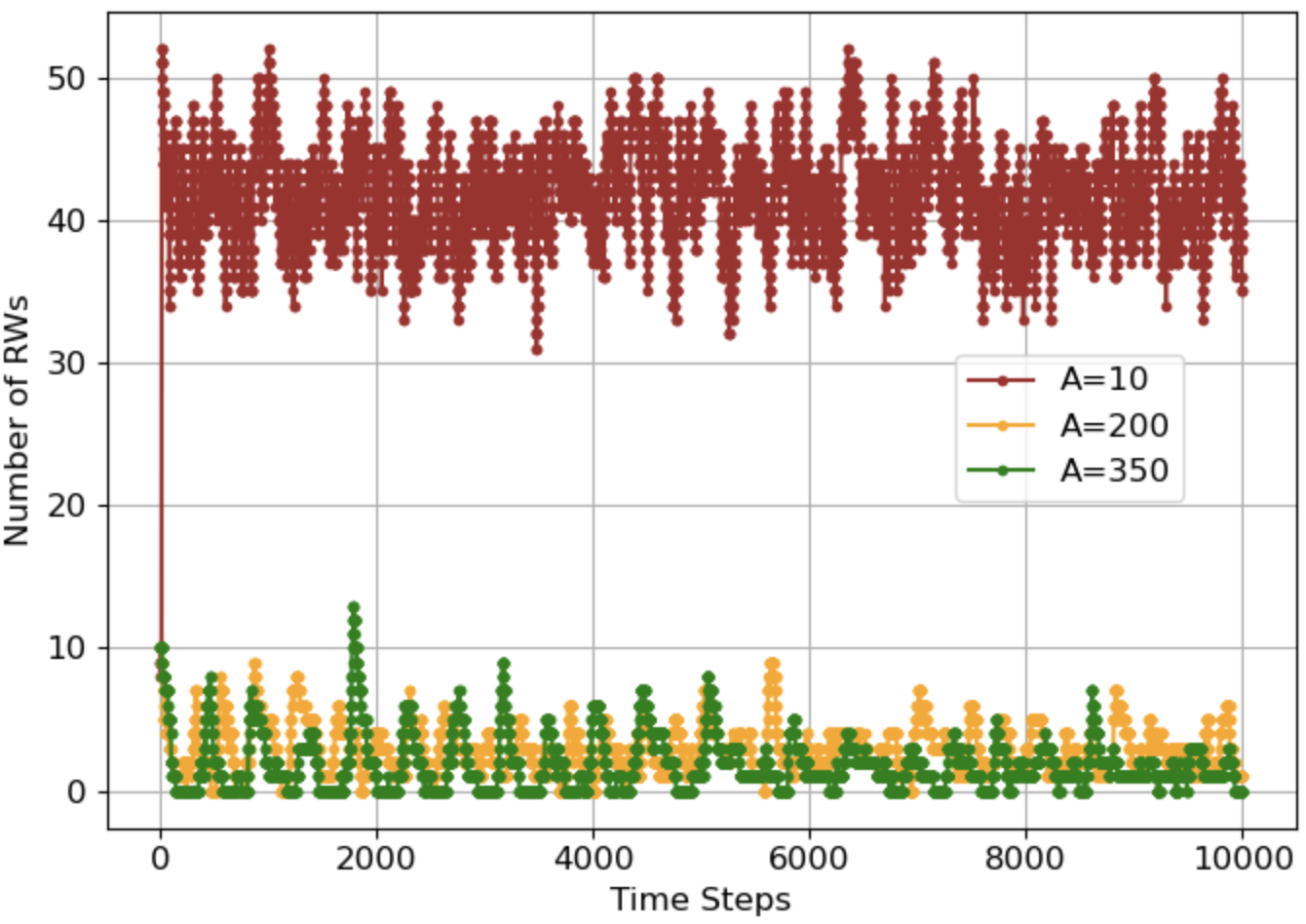}
    \caption*{(a) complete graph}
  \end{minipage}
  \hfill
  \begin{minipage}[b]{0.45\linewidth}
    \centering
    \includegraphics[width=\linewidth]{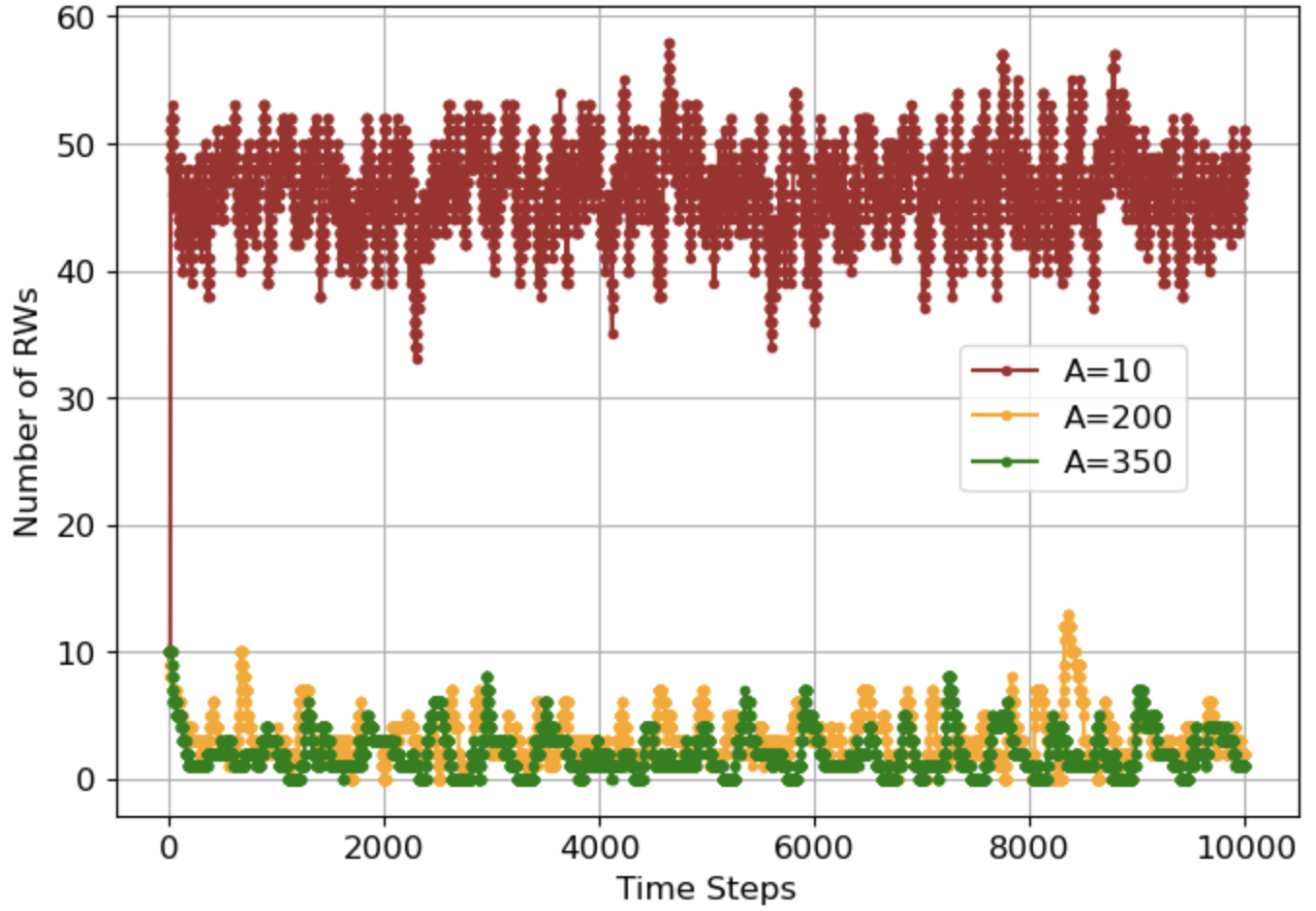}
    \caption*{(b) random regular graph}
  \end{minipage}
  \hfill
  \begin{minipage}[b]{0.45\linewidth}
    \centering
    \includegraphics[width=\linewidth]{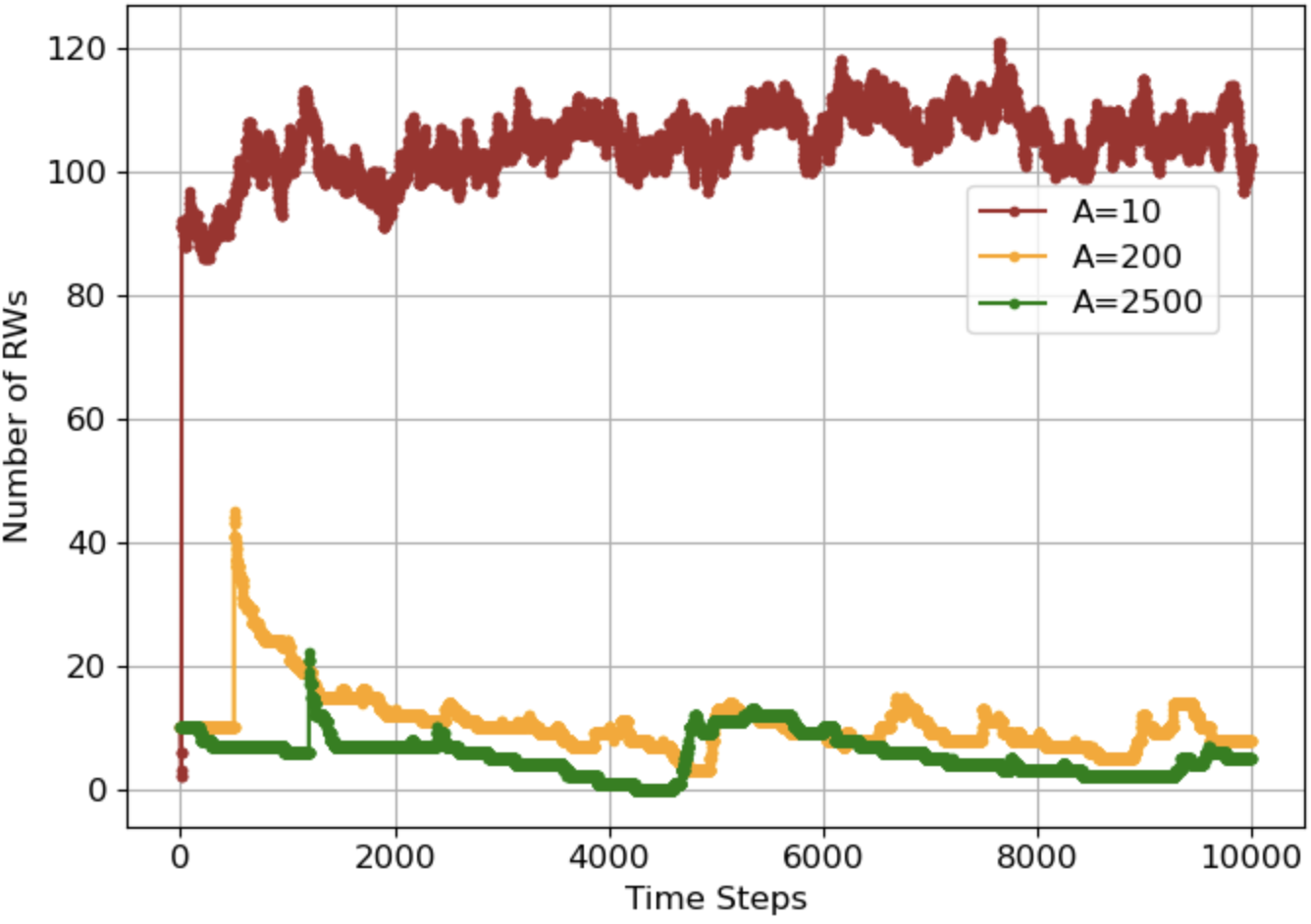}
    \caption*{(c) ring topology}
  \end{minipage}
    \hfill
  \begin{minipage}[b]{0.45\linewidth}
    \centering
    \includegraphics[width=\linewidth]{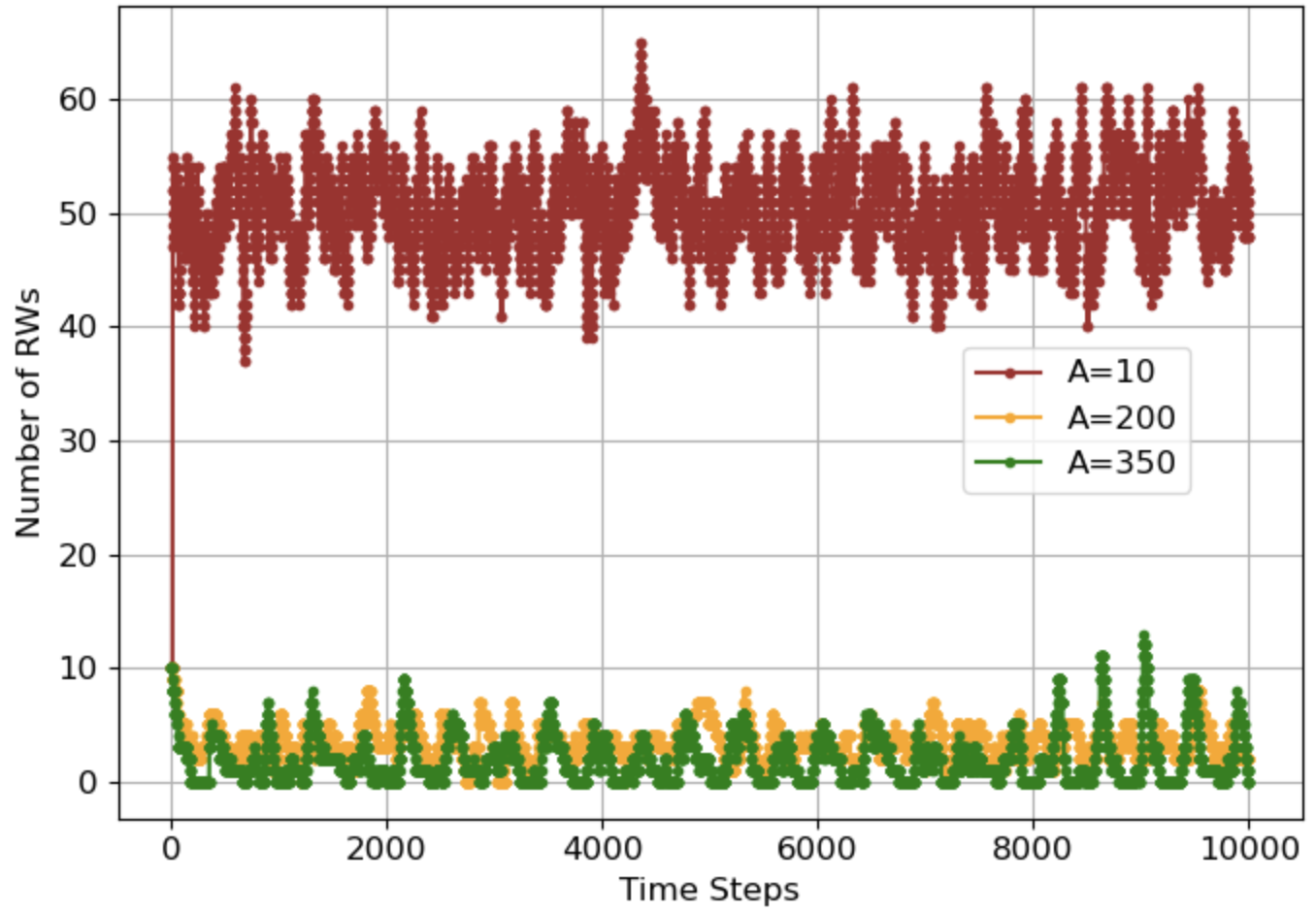}
    \caption*{(d) Erd\H{o}s--R\'enyi graph}
  \end{minipage}
  \caption{Number of RWs over time on different graphs.}
  \label{fig:Boundedness}
\end{figure}

Next, we fix the creation threshold $A=10$ and vary the termination probability $\zeta$  instead of fixing it at $1$. In Fig.~\ref{fig:zeta}, we set $\zeta \in \set{0.1, 0.01, 0.001}$ and present a \textit{single} realization of the RW population process $\set{Z_t}_{t\in\Z}$. From the figure, we observe that the population size remains bounded across different graph types, even when the termination probability $\zeta$ is very small. This observation is consistent with Theorem~\ref{thm:FiniteRWs}.

\begin{figure}[htbp]
  \centering
  \begin{minipage}[b]{0.45\linewidth}
    \centering
    \includegraphics[width=\linewidth]{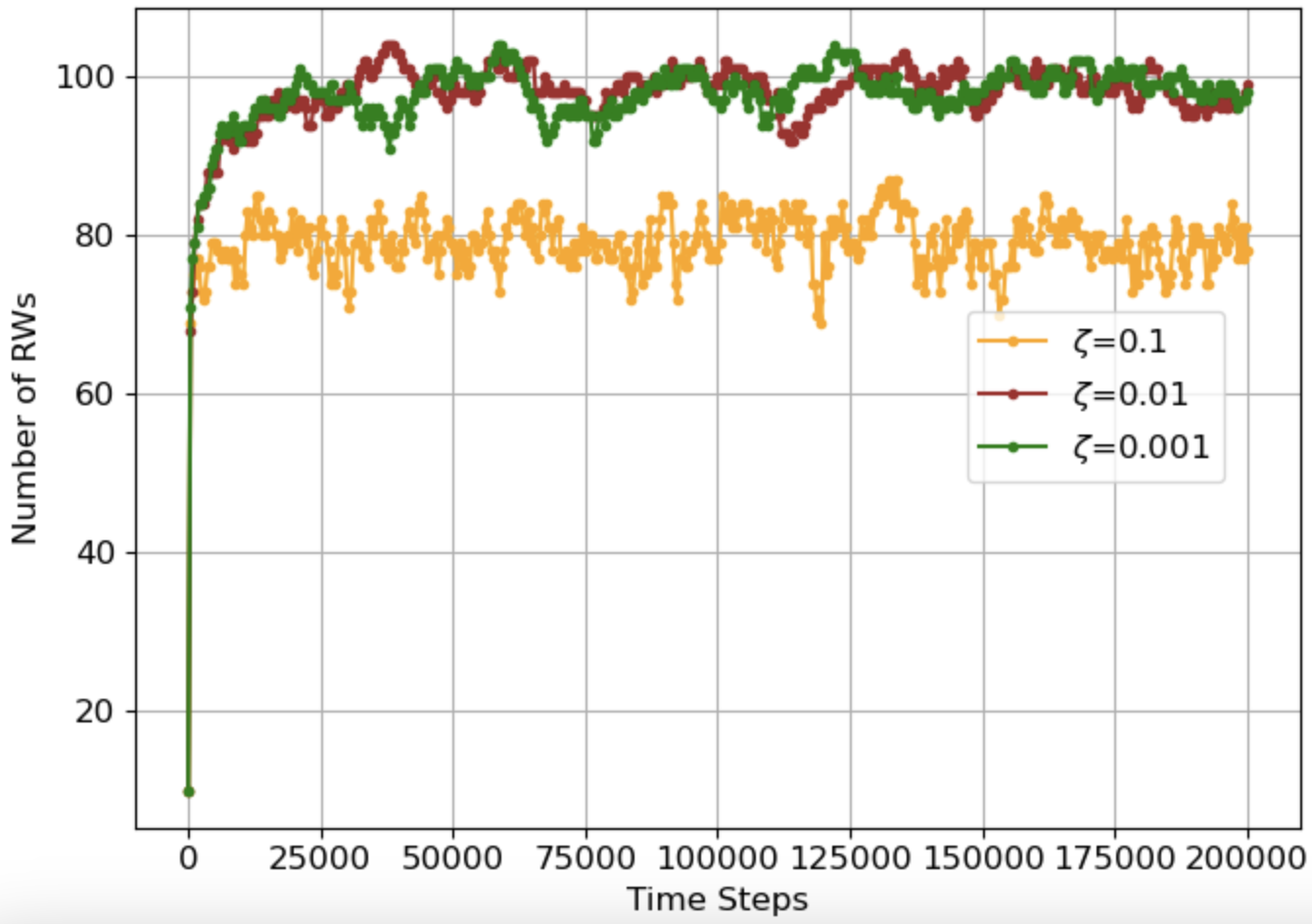}
    \caption*{(a) complete graph}
  \end{minipage}
  \hfill
  \begin{minipage}[b]{0.45\linewidth}
    \centering
    \includegraphics[width=\linewidth]{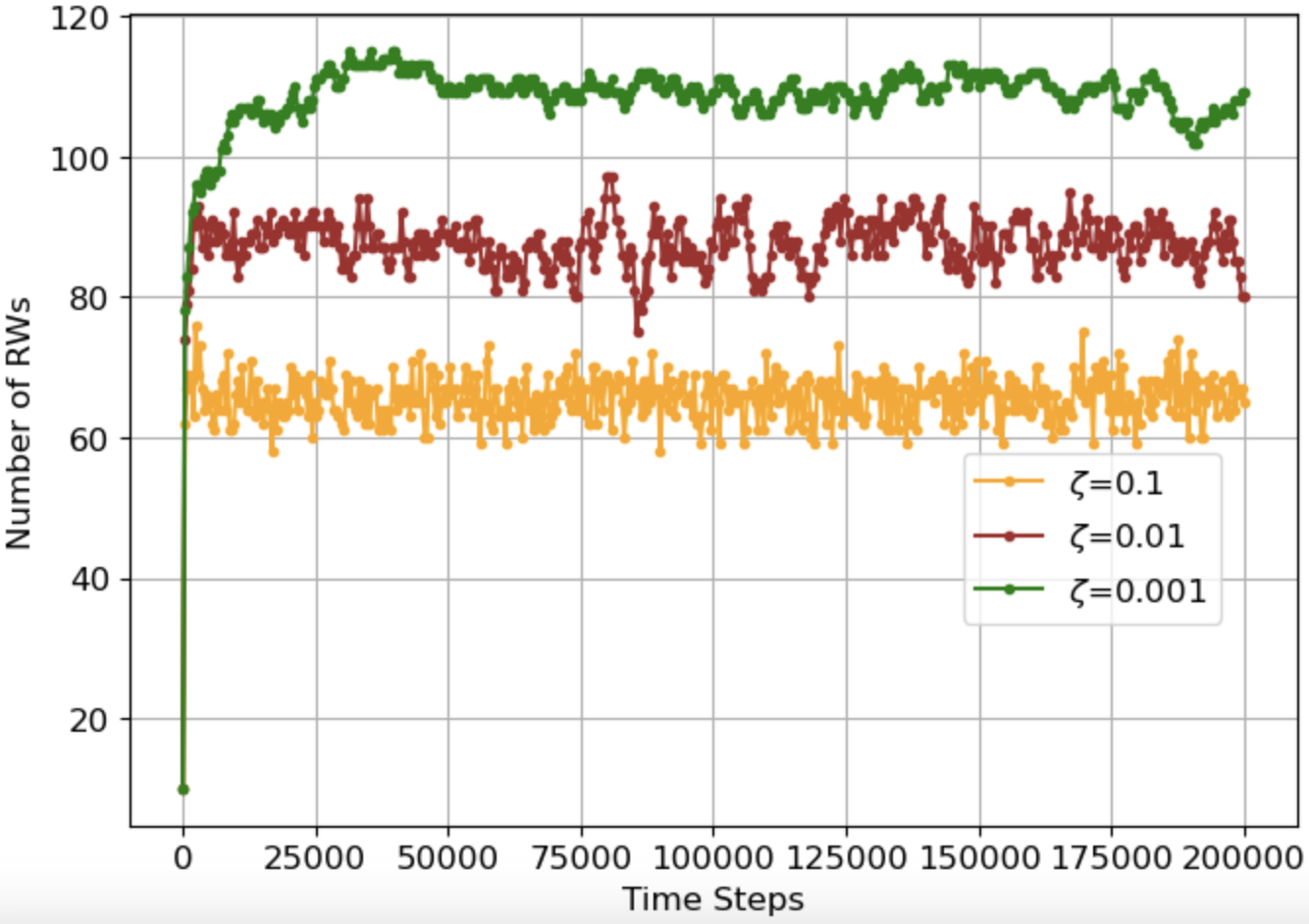}
    \caption*{(b) random regular graph}
  \end{minipage}
  \hfill
  \begin{minipage}[b]{0.45\linewidth}
    \centering
    \includegraphics[width=\linewidth]{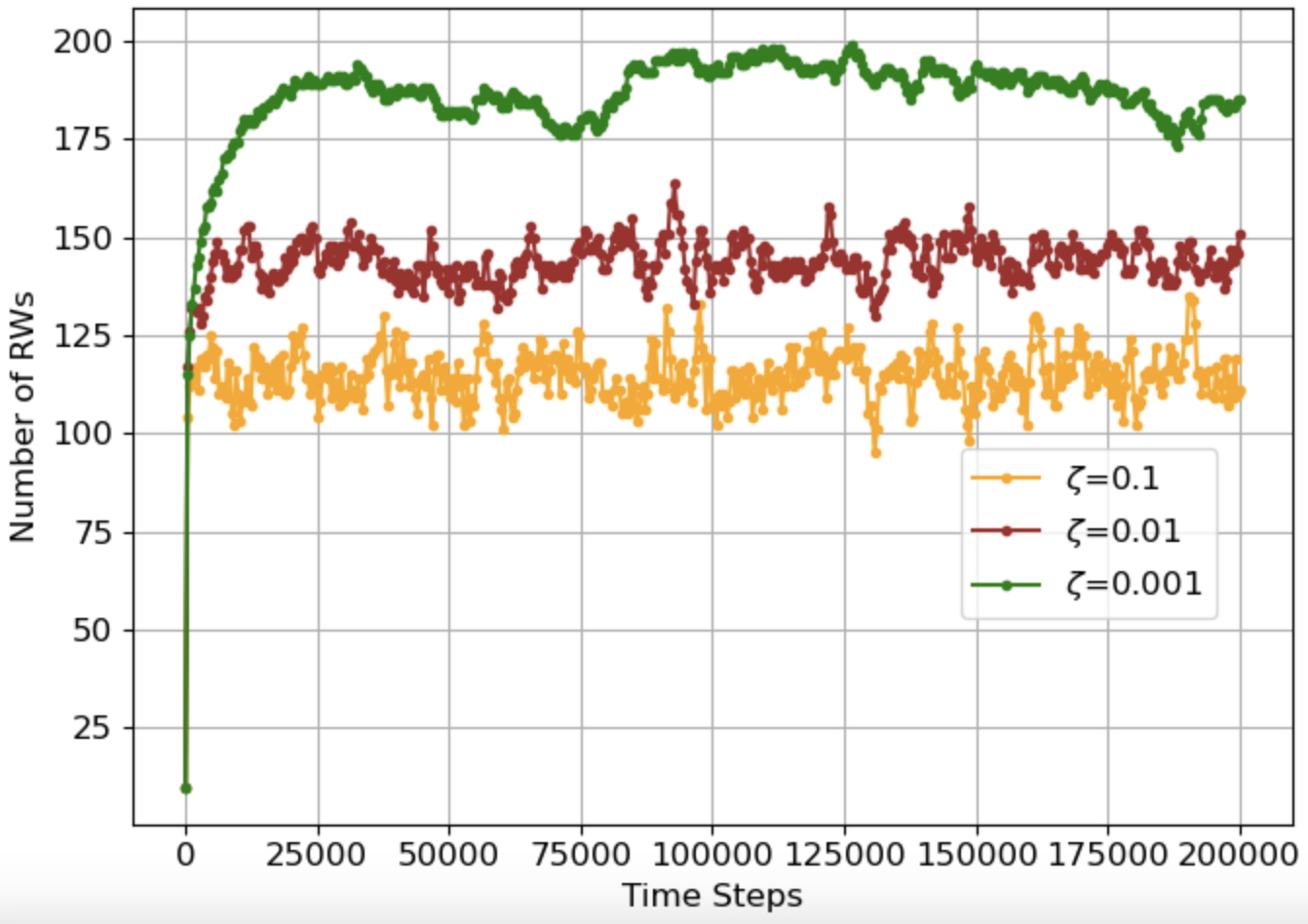}
    \caption*{(c) ring topology}
  \end{minipage}
    \hfill
  \begin{minipage}[b]{0.45\linewidth}
    \centering
    \includegraphics[width=\linewidth]{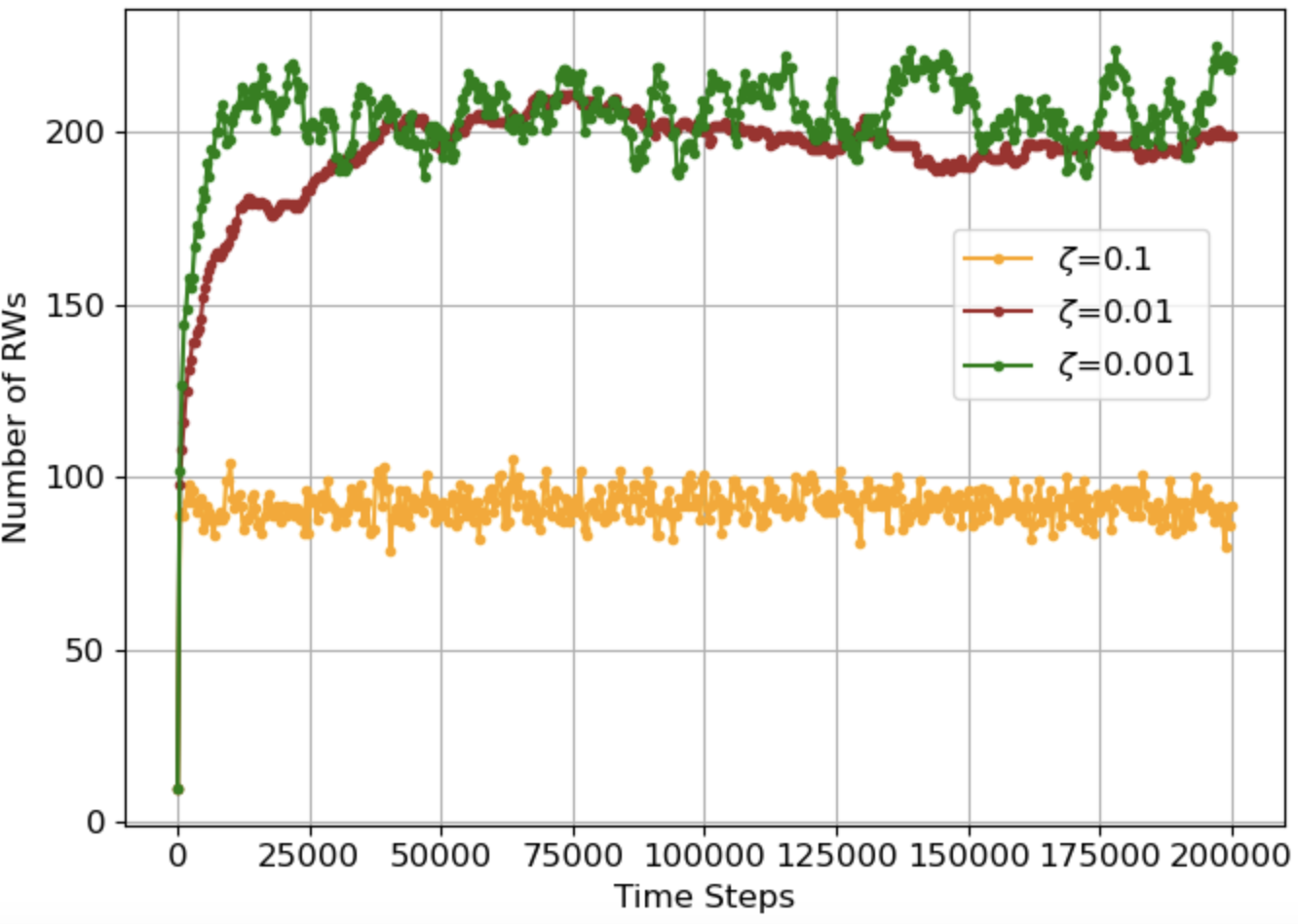}
    \caption*{(d) Erd\H{o}s--R\'enyi graph}
  \end{minipage}
  \caption{Number of RWs over time on different graphs when $\zeta\in\set{0.1, 0.01, 0.001}$.}
  \label{fig:zeta}
\end{figure}

The estimated $\bar{Z}^\star$ and the upper bound from Theorem~\ref{thm:Peak} are illustrated in Table~\ref{tab:Peak} and Fig.~\ref{fig:peakno}. Although Theorem~\ref{thm:Peak} formally applies only to the complete graph, we additionally evaluate $\bar{Z}^\star$ numerically on other graphs with small diameter, including regular and \ER graphs. Ring graphs are excluded from this experiment because their creation probability $q$ requires a different scaling regime and is therefore not directly comparable.

In Table~\ref{tab:Peak}, we fix the initial population $z_0=1$, termination probability $\zeta=1$, and creation threshold $A=1$. For different choices of the creation probability $q$,  the empirical $\bar{Z}^\star$ remains strictly below the theoretical upper bound, thereby validating Theorem~\ref{thm:Peak}.

\begin{table}[h]
\centering
\begin{tabular}{|c|c|c|c|c|c|}
\hline
 & $q=1$& $q=\frac{1}{N}$ & $q=\frac{1}{N^2}$ \\
\hline
$\bar{Z}^\star\rule{0pt}{2.5ex}$ in Complete& $205.3$ & $45.2$ & $1.13$  \\
\hline
$\bar{Z}^\star\rule{0pt}{2.5ex}$ in Regular & $211.9$ & $44.6$ & $1.25$  \\
\hline
$\bar{Z}^\star\rule{0pt}{2.5ex}$ in \ER & $208.1$& $47.5$ & $1.37$ \\
\hline
Upper Bound
Theorem~\ref{thm:Peak}& $O(10^4)$ & $O(10^2)$ & $O(1)$  \\
\hline
\end{tabular}
\caption{The expected peak number of RWs, $\bar{Z}^\star$, evaluated on complete, regular, and \ER graphs.}
\label{tab:Peak}
\end{table}

Fig.~\ref{fig:peakno} shows that under the scaling $q=\frac{1}{N^2}$ and with $z_0=10$, the expected number of RWs $\E[Z_t]$ (and consequently $\bar{Z}^\star$) on complete, regular, and \ER graphs is always below the bound in \eqref{eq:barZUB} for all tested network size $N\in\set{20, 50, 100, 200}$. This provides strong empirical evidence that the peak RW population can be made independent of the network size $N$ even beyond the complete-graph setting. It is worth noting that the bound in \eqref{eq:barZUB} is inherently conservative, since it is derived under the worst-case choice $A_u=1$ for all $u\in\cB$, which maximizes the RW creation rate. Therefore, \eqref{eq:barZUB} serves as a universal upper bound for all threshold choices with $A_u>1$, but may be loose for larger thresholds, where RW creation becomes less frequent.

\begin{figure}[htbp]
  \centering
  \begin{minipage}[b]{0.45\linewidth}
    \centering
    \includegraphics[width=\linewidth]{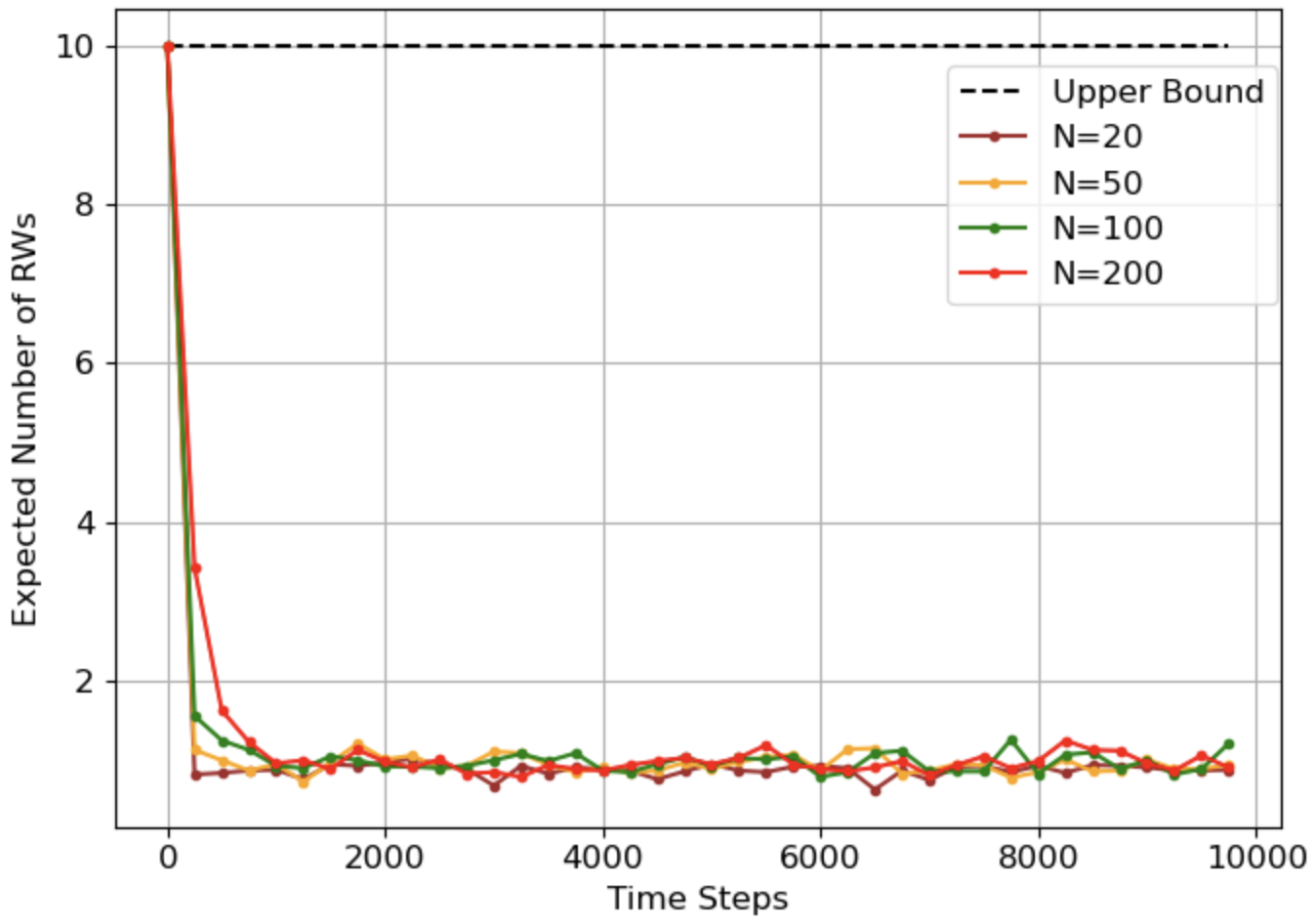}
    \caption*{(a) complete graph}
  \end{minipage}
  \hfill
  \begin{minipage}[b]{0.45\linewidth}
    \centering
    \includegraphics[width=\linewidth]{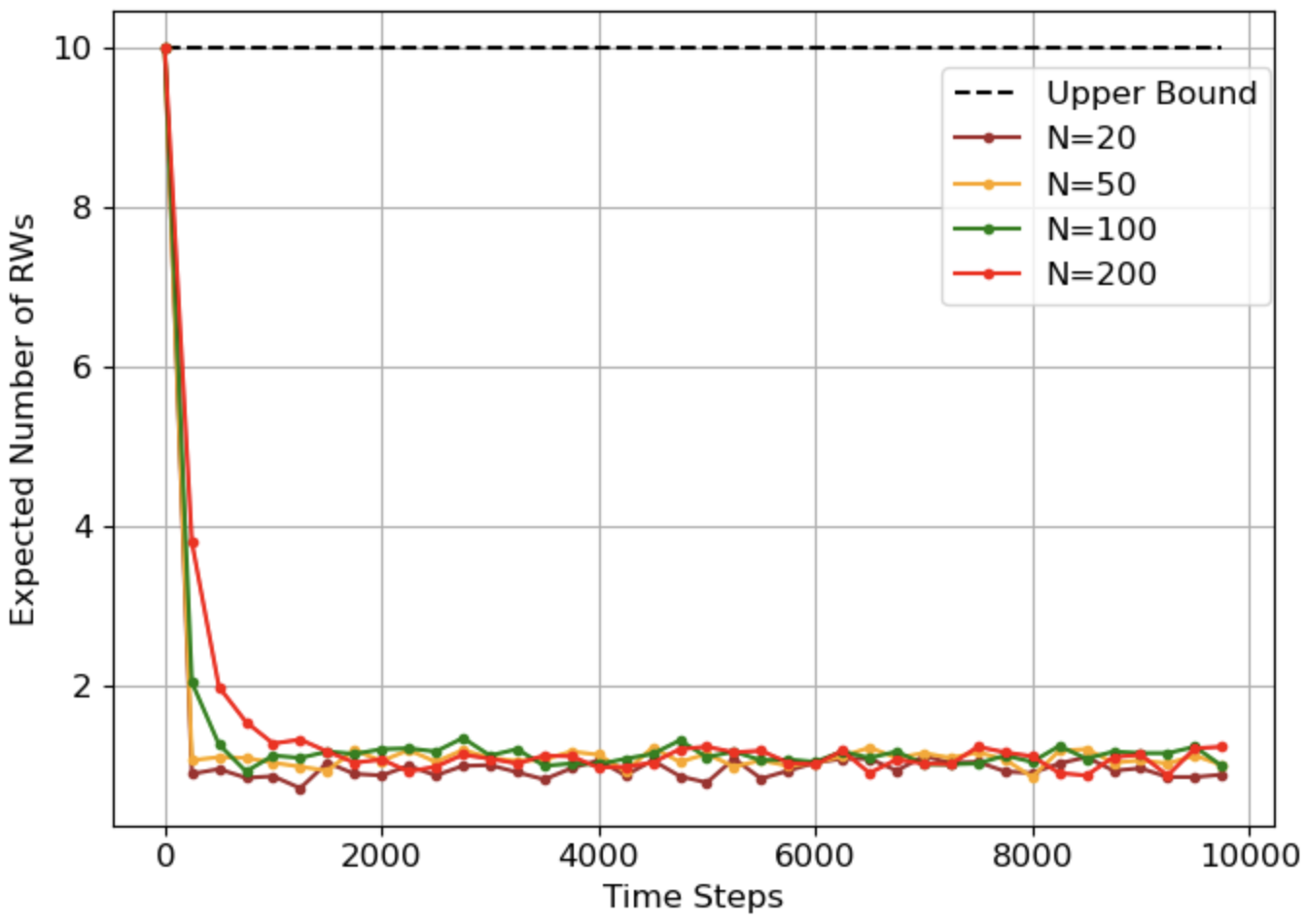}
    \caption*{(b) random regular graph}
  \end{minipage}
    \hfill
  \begin{minipage}[b]{0.45\linewidth}
    \centering
    \includegraphics[width=\linewidth]{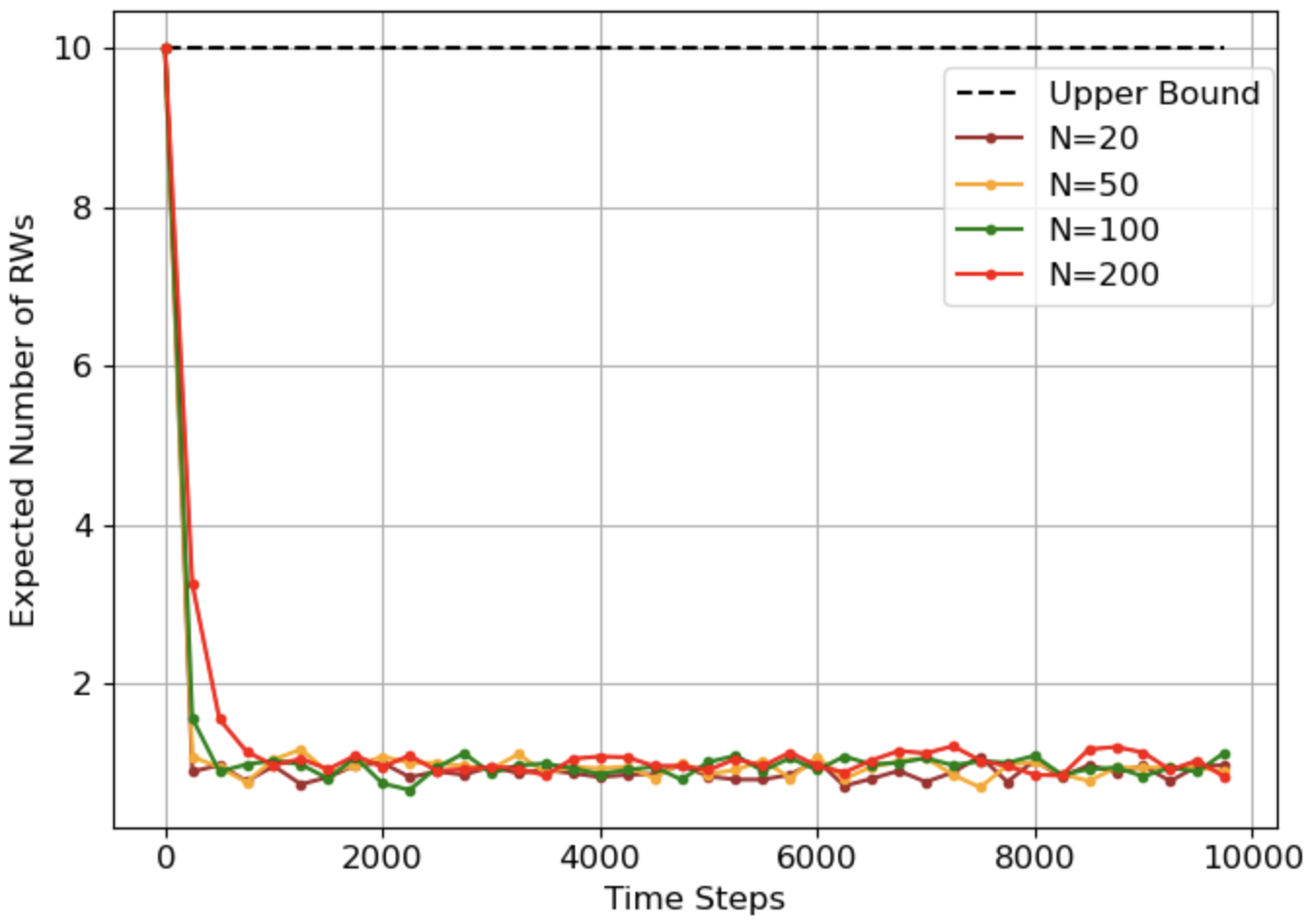}
    \caption*{(c) \ER graph}
  \end{minipage}
  \caption{Expected number of RWs over time on different graphs under the scaling $q=\frac{1}{N^2}$ for $N\in\set{20, 50, 100, 200}$, compared with the upper bound in \eqref{eq:barZUB}.}
  \label{fig:peakno}
\end{figure}

In the end of this subsection, we present a comparison between the number of RWs generated by the \cwl algorithm and by the baseline \DeCa algorithm. To ensure a fair comparison, the two  mechanisms adopt the same RW transition matrix $P$, so their movement behavior remains identical. The results are shown in Fig.~\ref{fig:DeCa}.  We observe that the \cwl algorithm outperforms the \DeCa algorithm in terms of robustness. Under the \DeCa algorithm, the permanent extinction of the RW population is highly sensitive to the algorithm parameter $\epsilon$. If $\epsilon$ is not chosen properly, the RW population eventually goes extinct (brown curves). In contrast, under the \cwl algorithm, the permanent extinction never occurs. Regardless of how large the creation threshold $A$ is, if the population goes extinct, after a certain period of time, the RW population will always revive.

\begin{figure}[htbp]
  \centering
  \begin{minipage}[b]{0.45\linewidth}
    \centering
    \includegraphics[width=\linewidth]{Fig/LCDcomplete.png}
    \caption*{(a) complete graph}
  \end{minipage}
  \hfill
  \begin{minipage}[b]{0.45\linewidth}
    \centering
    \includegraphics[width=\linewidth]{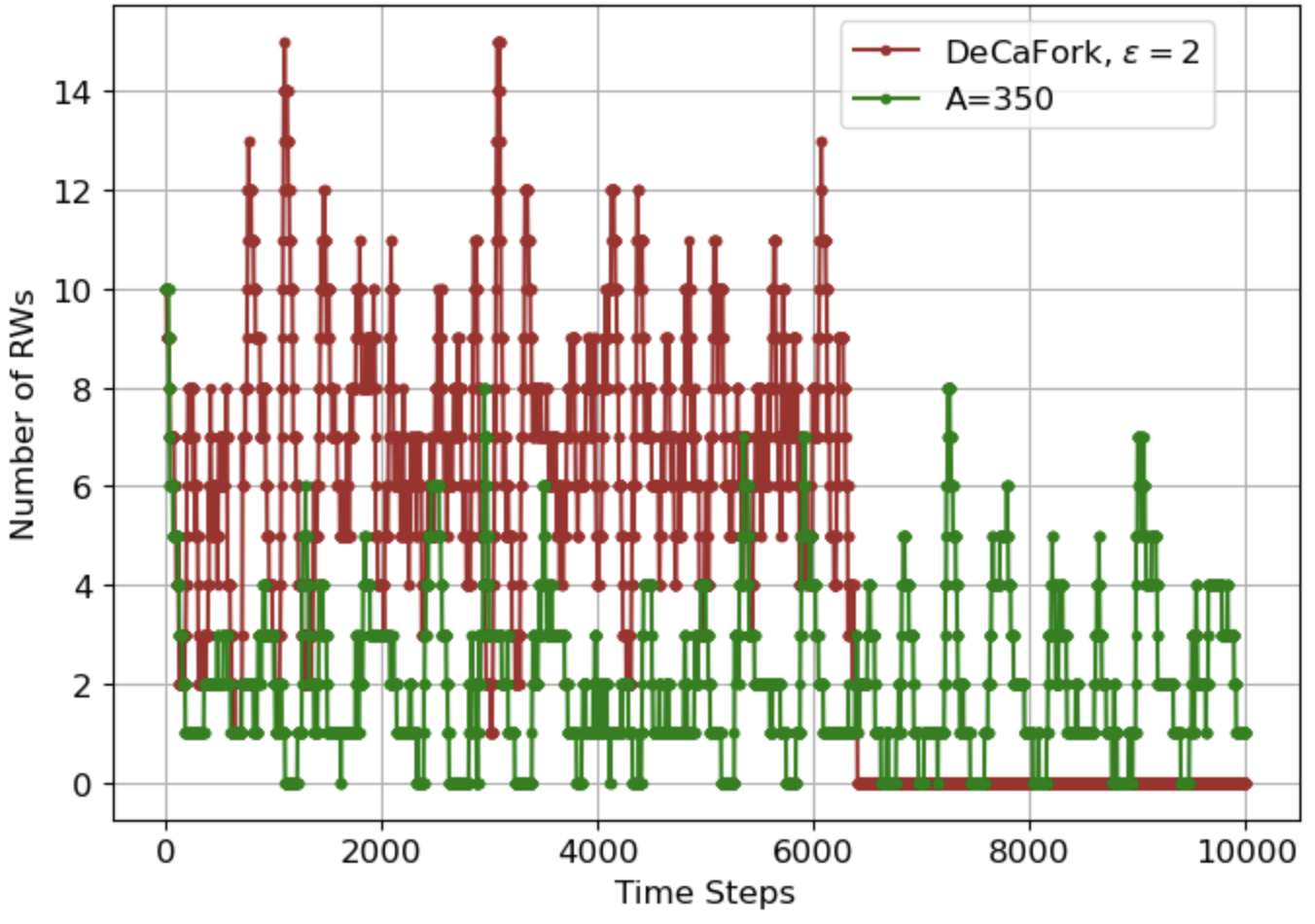}
    \caption*{(b) random regular graph}
  \end{minipage}
  \hfill
  \begin{minipage}[b]{0.45\linewidth}
    \centering
    \includegraphics[width=\linewidth]{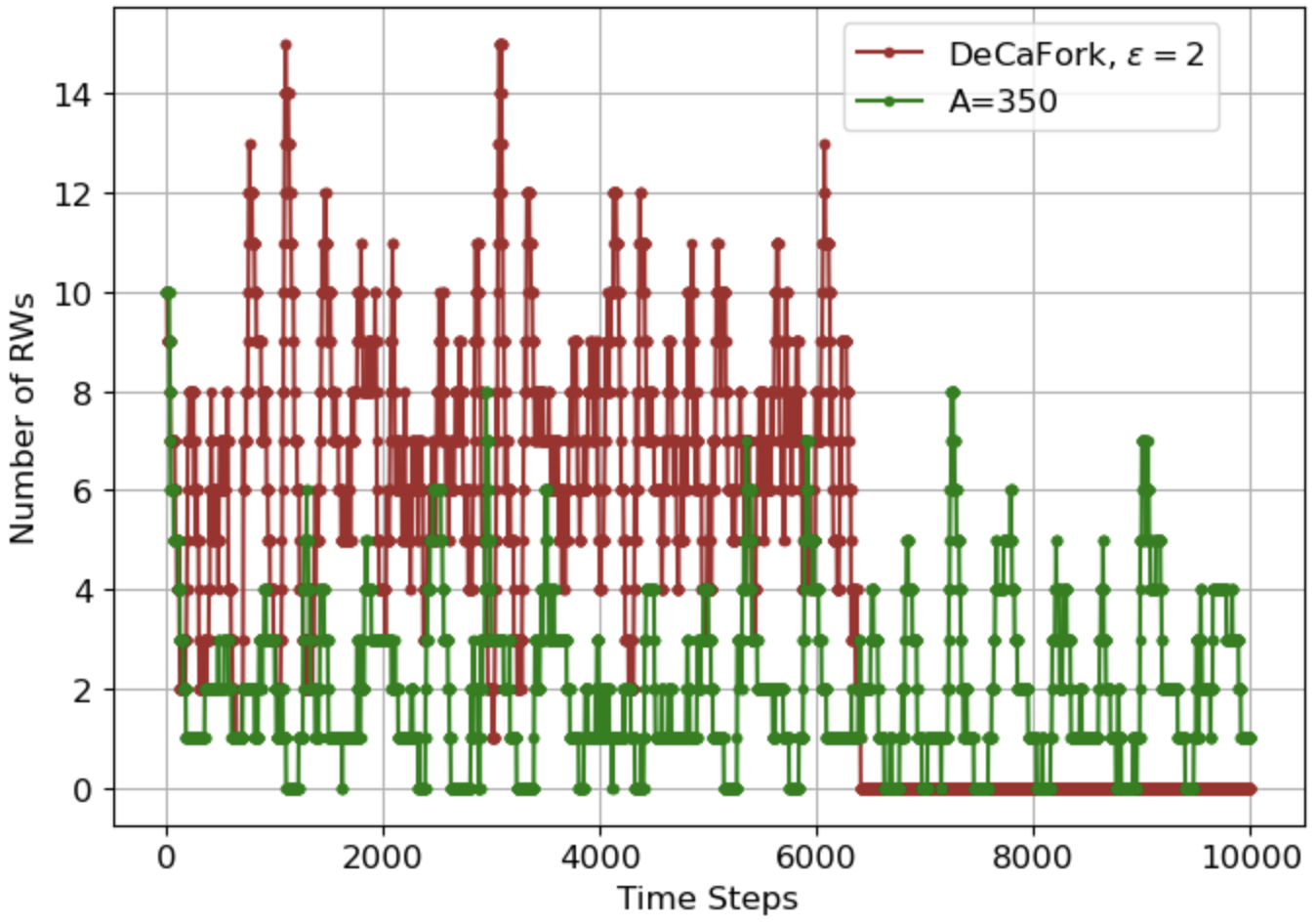}
    \caption*{(c) ring topology}
  \end{minipage}
    \hfill
  \begin{minipage}[b]{0.45\linewidth}
    \centering
    \includegraphics[width=\linewidth]{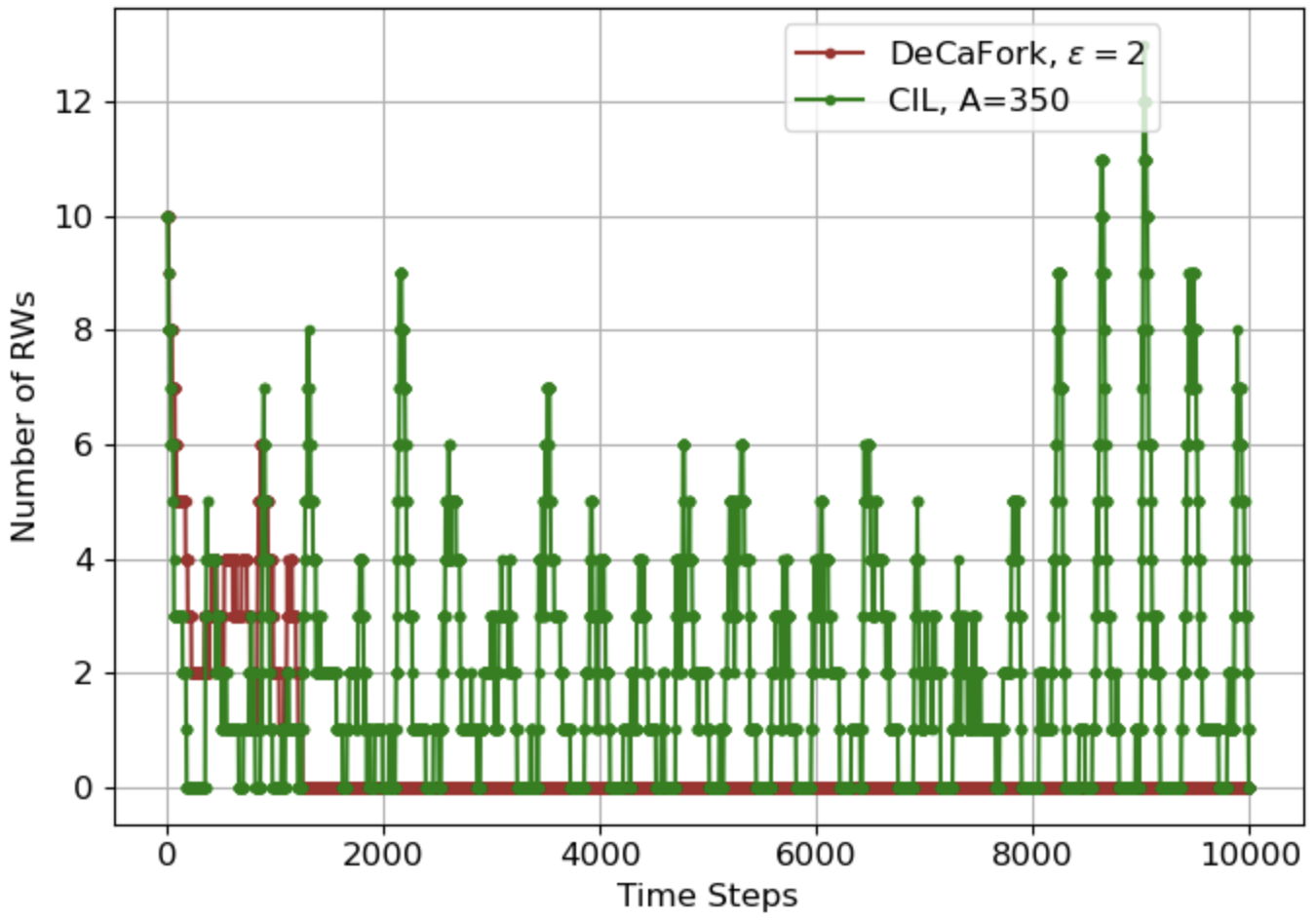}
    \caption*{(d) \ER graph}
  \end{minipage}
  \caption{Comparison of the number of RWs under the \cwl and \DeCa algorithms.}
  \label{fig:DeCa}
\end{figure}

\subsection{Convergence}\label{subsec: convergence}
We now evaluate the performance of the proposed \cwl algorithm against the baseline \DeCa. We begin with experiments on the synthetic dataset. As noted in Remark~\ref{remark:nogossip}, we do not average the models across all benign nodes as gossip-based algorithms do \cite{tsitsiklis2003distributed}; we also do not average across all active RWs. Instead, we select \textit{one} chain of RWs as defined in Definition~\ref{defn:ChainRWs} and take the expectation only over the intrinsic randomness of the algorithm, such as the stochastic-gradient noise and the RW dynamics. Accordingly, when plotting the loss curve, we randomly choose one chain of RWs and report its loss. We clarify that the loss curves in Figs.~\ref{fig:Convergence},~\ref{fig:Convergencerealiid}, and~\ref{fig:ConvergencerealNoniid} are used to compare the convergence behavior of a single chain of RWs. The communication overhead associated with each plotted loss curve is the overhead incurred along that particular chain of RWs. When the threshold $A$ is small, more RWs may be generated on the graph, which can indeed increase the total communication overhead of the overall system. However, this additional graph-level overhead is not reflected in the per-chain loss curve.

Fig.~\ref{fig:Convergence} presents the convergence performance of RW-SGD using the proposed \cwl algorithm, compared with the \DeCa baseline. In each sub-figure, the $y$-axis represents the value of the global loss (log scale), while the $x$-axis denotes the number of time steps. Across different graph topologies\footnote{In the ring topology, we set the initial number of RWs to $z_0=1$ instead of $z_0=10$. This is because an RW is much less likely to hit the Pac-Man node in a ring topology than in complete, regular, and \ER graphs; using a smaller initial RW population makes the impact of the Pac-Man attack more visible.}:
\begin{compactenum}[(i)]
\item For the \cwl with a small creation threshold (e.g, $A=10$), the loss curve consistently decreases over time and eventually approach zero, indicating effective convergence. When $A=10$, extinction events under the \cwl algorithm occur only rarely, so with high probability the loss curve contains no flat segments. In this experiment, no extinction events happen. 
\item For \cwl with a large creation threshold (e.g, $A=350$), the loss curve exhibits multiple horizontal segments (green curve). During each flat segment, the RW population becomes extinct, and the learning process remains inactive. After a certain time period, once at least one node has waited for $A$ time slots since its last visit, new RWs are created and the learning process resumes. 
\item As shown in Fig.~\ref{fig:DeCa}, the robustness of the \DeCa algorithm depends strongly on the choice of the parameter $\epsilon$. The \DeCa algorithm can delay RW extinction but does not completely prevent it. In contrast, \textsc{CIL} avoids permanent extinction by creating new RWs when needed, at the cost of additional communication overhead. We observe that, when $\epsilon$ is not chosen properly, i.e., $\epsilon=3$, the RW population may eventually become extinct. In this case, no RW remains to carry out local updates, and the learning task fails to complete. It is worth noting that, based on the design of \DeCa \cite{selfdup, egger2024self}, it has a positive probability of RW extinction. Therefore, regardless of how carefully the hyperparameter $\epsilon$ is chosen, there always remains a nonzero risk that all RWs become extinct, thereby terminating the training process.
\end{compactenum}

From Fig.~\ref{fig:Convergence}, one might conclude that the \textsc{CIL} algorithm performs better with a smaller threshold $A$, since it converges faster in terms of time steps. This observation may suggest choosing a very small threshold, such as $A=1$, to accelerate the learning process. However, a smaller threshold can result in substantially higher communication overhead. As shown in Fig.~\ref{fig:LossCost}, to achieve the same loss level, the \textsc{CIL} algorithm with $A=10$ incurs greater communication overhead than that with $A=350$. This observation illustrates the tradeoff between convergence speed and communication overhead discussed in Proposition~\ref{pro:NumIter} and Remark~\ref{remark:Tradeoff}. Although Fig.~\ref{fig:LossCost} presents results only for the complete graph, similar phenomena are also observed for regular, ring, and \ER graphs.

The distance between the new convergence point and the original optimal solution, along with the corresponding theoretical bounds, is reported in Table~\ref{tab:Bounds}. This validates the theoretical guarantees established in Proposition~\ref{pro:Bounds}(1).

\begin{figure}[h]
\centering
  \centering
  \begin{minipage}[b]{0.45\linewidth}
    \centering
    \includegraphics[width=\linewidth, height=0.9\linewidth]{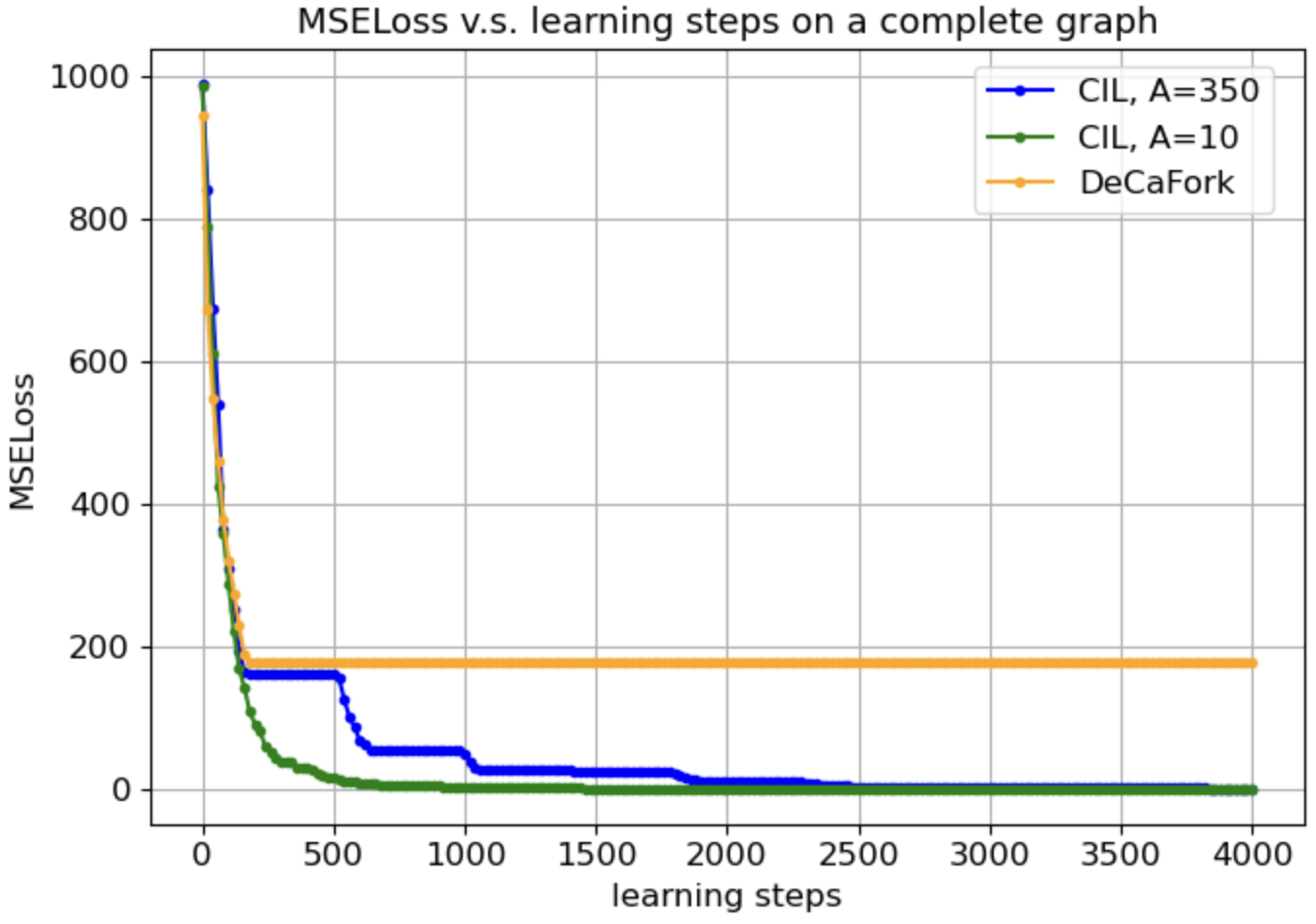}
    \caption*{(a) complete graph}
  \end{minipage}
  \hfill
  \begin{minipage}[b]{0.45\linewidth}
    \centering
    \includegraphics[width=\linewidth, height=0.9\linewidth]{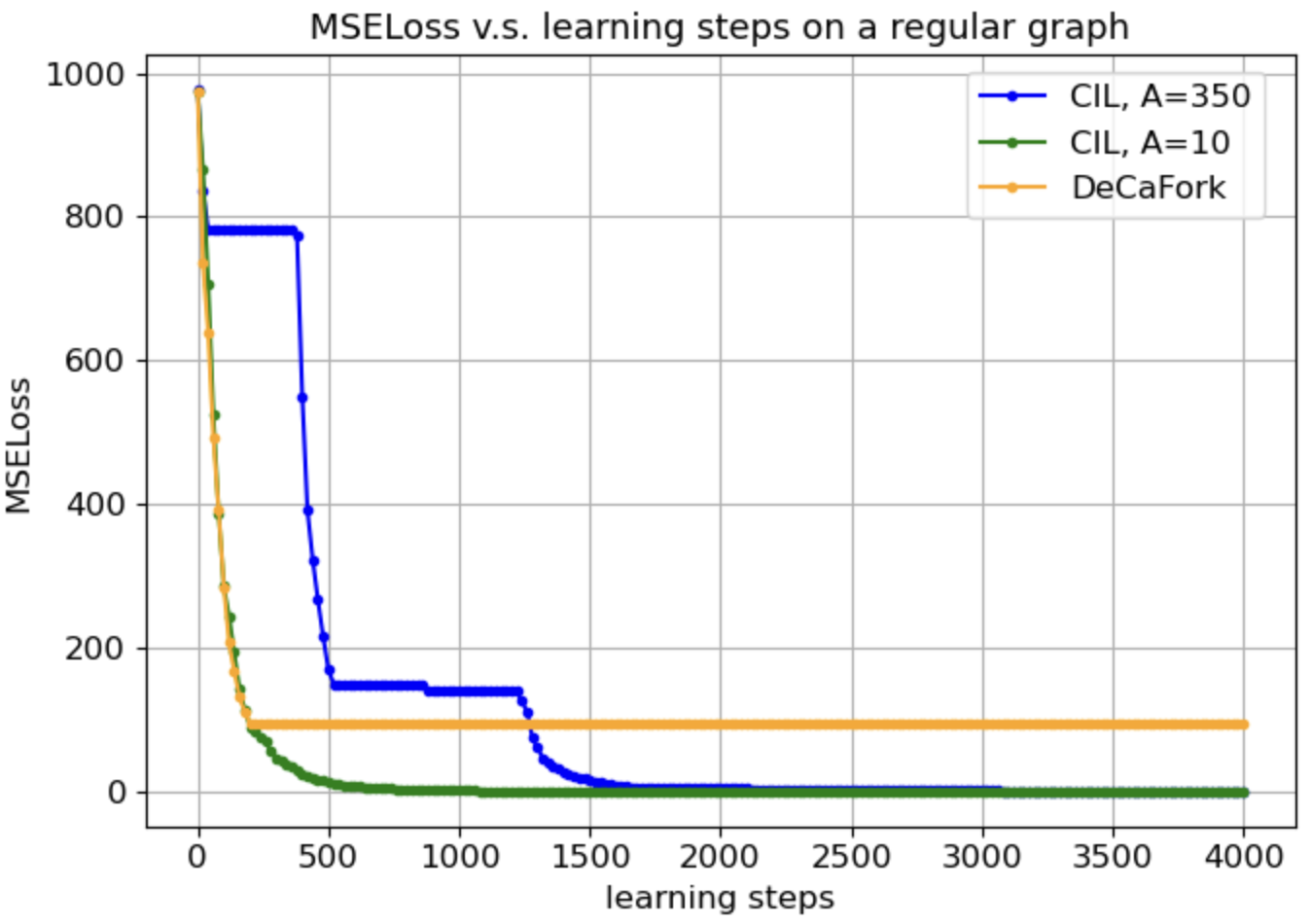}
    \caption*{(b) random regular graph}
  \end{minipage}
  \hfill
  \begin{minipage}[b]{0.45\linewidth}
    \centering
    \includegraphics[width=\linewidth, height=0.9\linewidth]{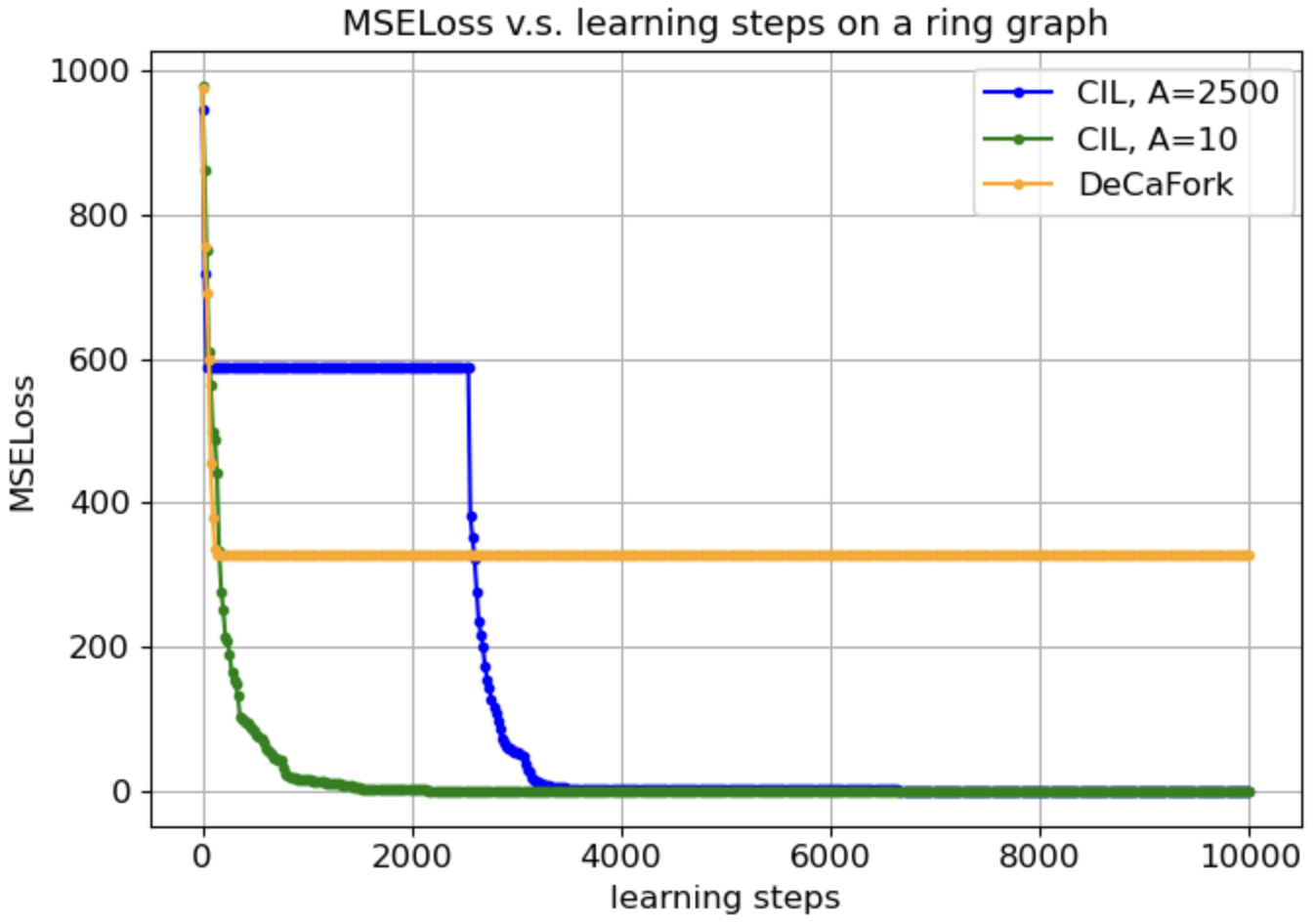}
    \caption*{(c) ring topology}
  \end{minipage}
  \hfill
  \begin{minipage}[b]{0.45\linewidth}
    \centering
    \includegraphics[width=\linewidth, height=0.9\linewidth]{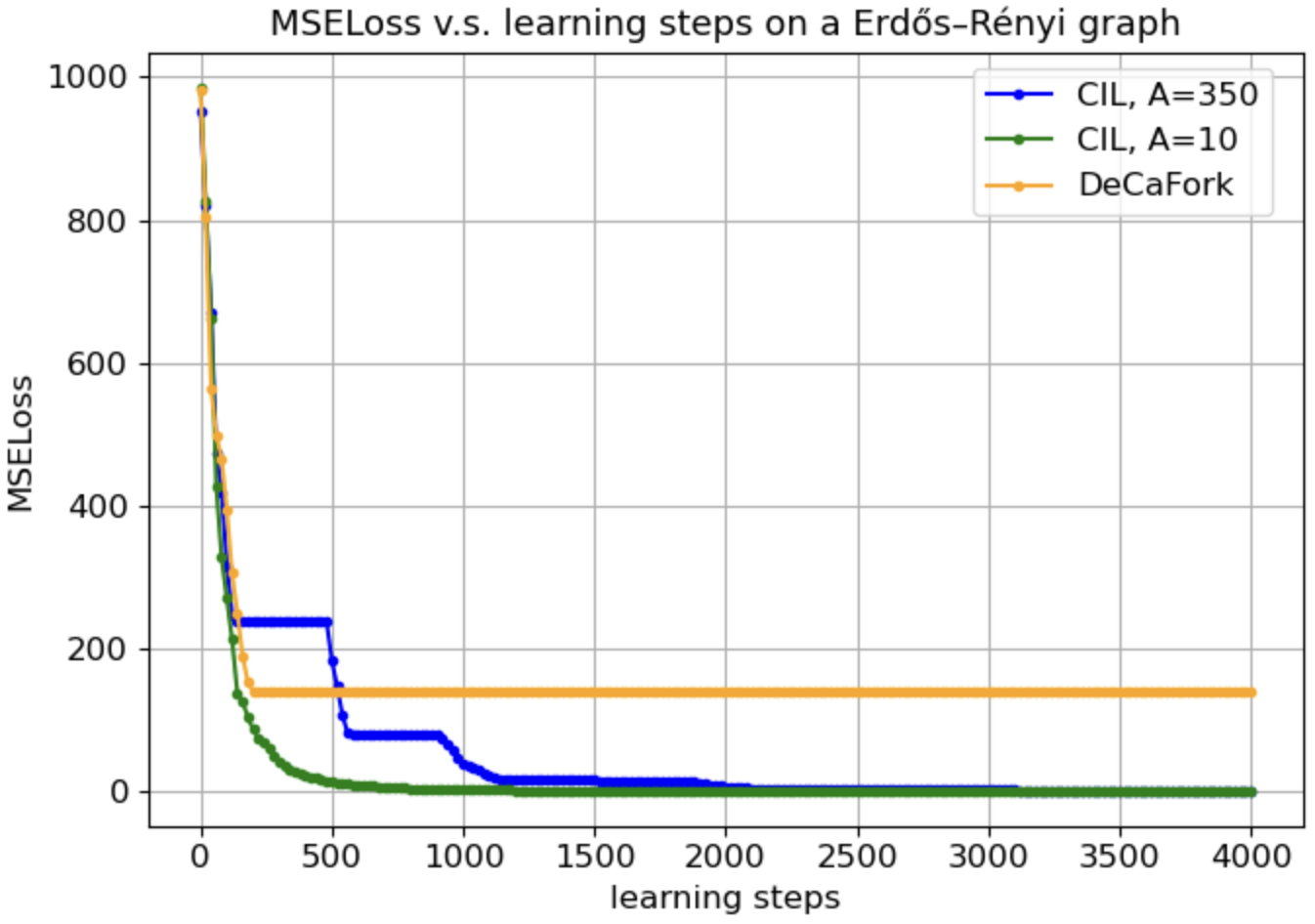}
    \caption*{(c) Erd\H{o}s--R\'enyi graph}
  \end{minipage}
\caption{Loss function v.s. learning steps on different graphs.}
\label{fig:Convergence}
\end{figure}

\begin{figure}[htbp]
  \centering
\includegraphics[width=0.675\linewidth]{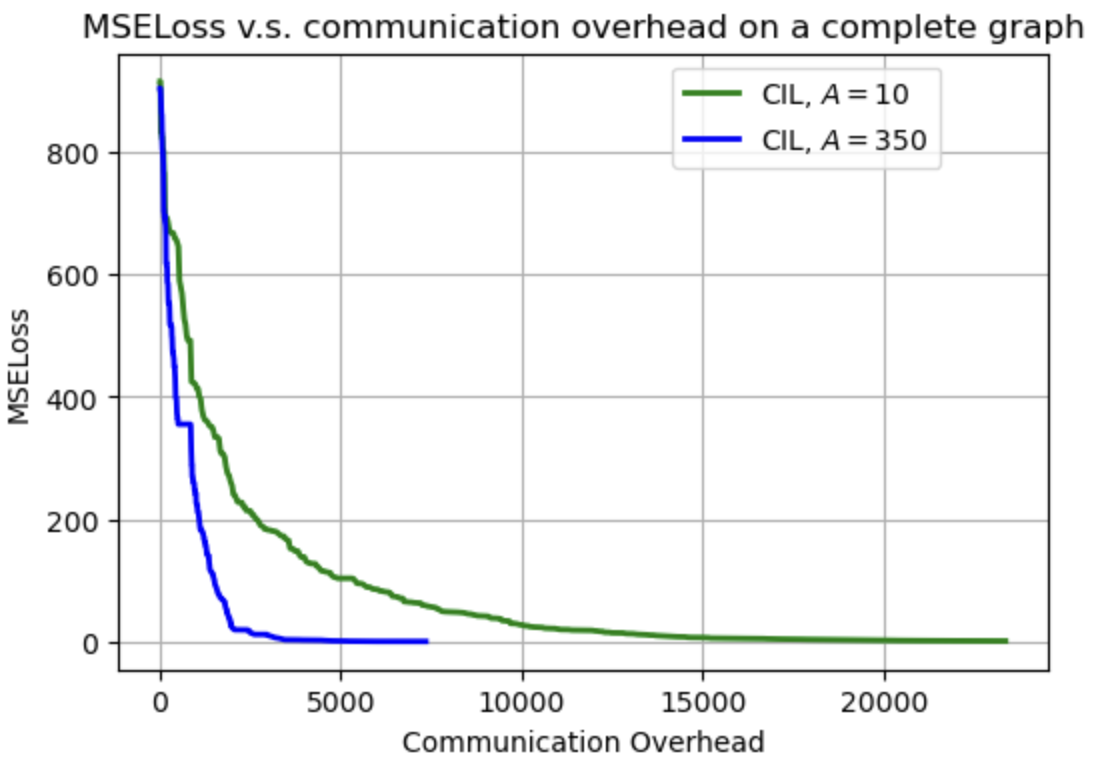}
  \caption{Loss v.s. communication overhead for the \textsc{CIL} algorithm with different thresholds $A$ on a complete graph.}
  \label{fig:LossCost}
\end{figure}

\begin{table}[h]
\centering
\begin{tabular}{|c|c|c|c|c|}
\hline
graph type& Complete & Regular & Ring & Erd\H{o}s--R\'enyi \\
\hline
$\frac{1}{L}\|\nabla f(\tilde{\bf x}^\star)\|$ & $0.018$ & $0.013$  & $0.042$ & $0.015$\\
\hline
$\|\tilde{\bf x}^* - {\bf x}^*\|$ & $0.024$ & $0.018$ & $0.051$ & $0.022$\\
\hline
$\frac{1}{\mu}\|\nabla f(\tilde{\bf x}^\star)\|$ & $0.041$ & $0.033$ & $0.139$ & $0.037$\\
\hline
\end{tabular}
\caption{$\|\tilde{\bf x}^* - {\bf x}^*\|$ and its bounds on different graphs.}
\label{tab:Bounds}
\end{table}

We then conducted experiments on the real-world dataset, considering both \iid partitioning and non-\iid partitioning scenarios.

\begin{figure}[h]
\centering
  \centering
  \begin{minipage}[b]{0.45\linewidth}
    \centering
    \includegraphics[width=\linewidth, height=0.9\linewidth]{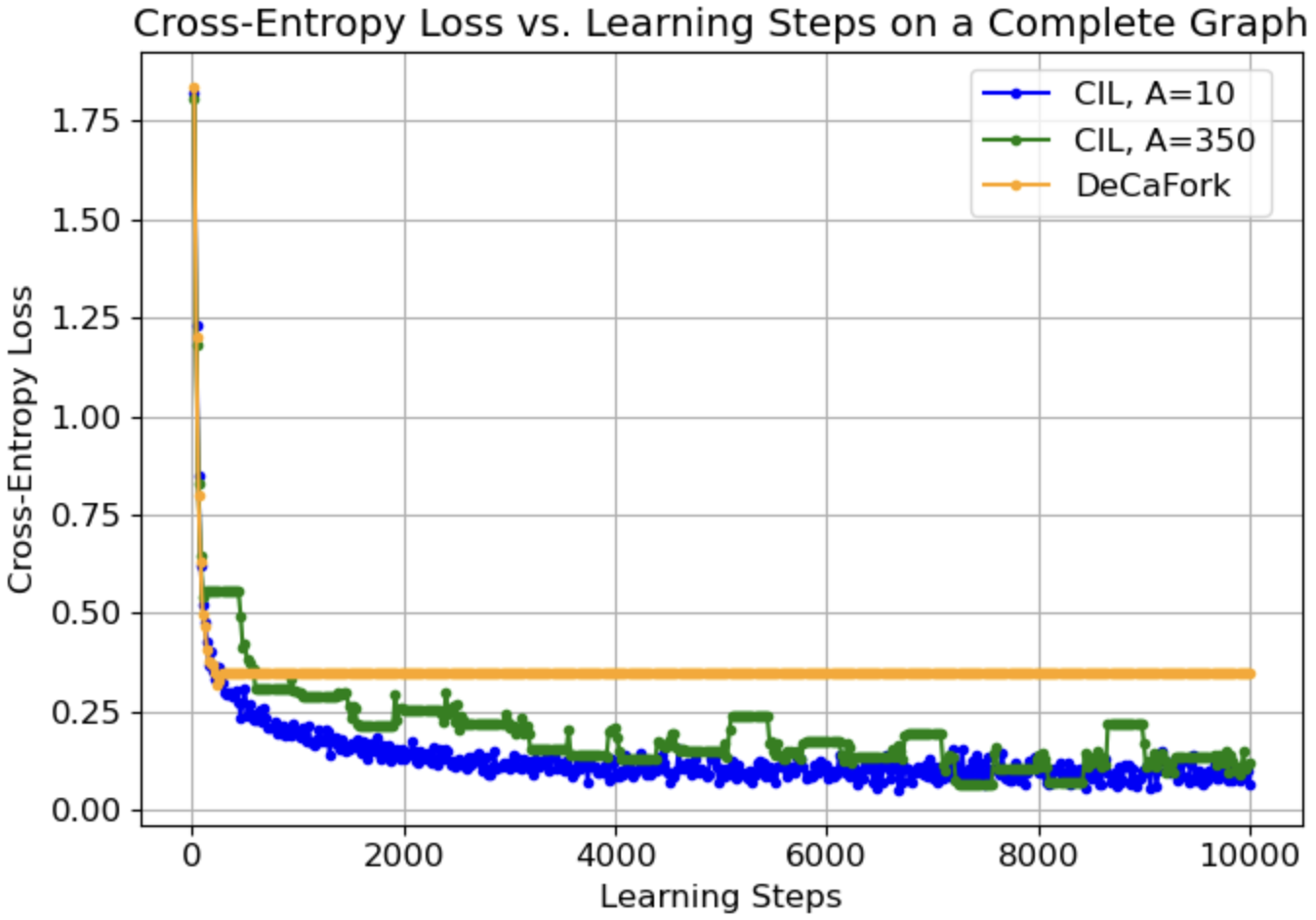}
    \caption*{(a) complete graph}
  \end{minipage}
  \hfill
  \begin{minipage}[b]{0.45\linewidth}
    \centering
    \includegraphics[width=\linewidth, height=0.9\linewidth]{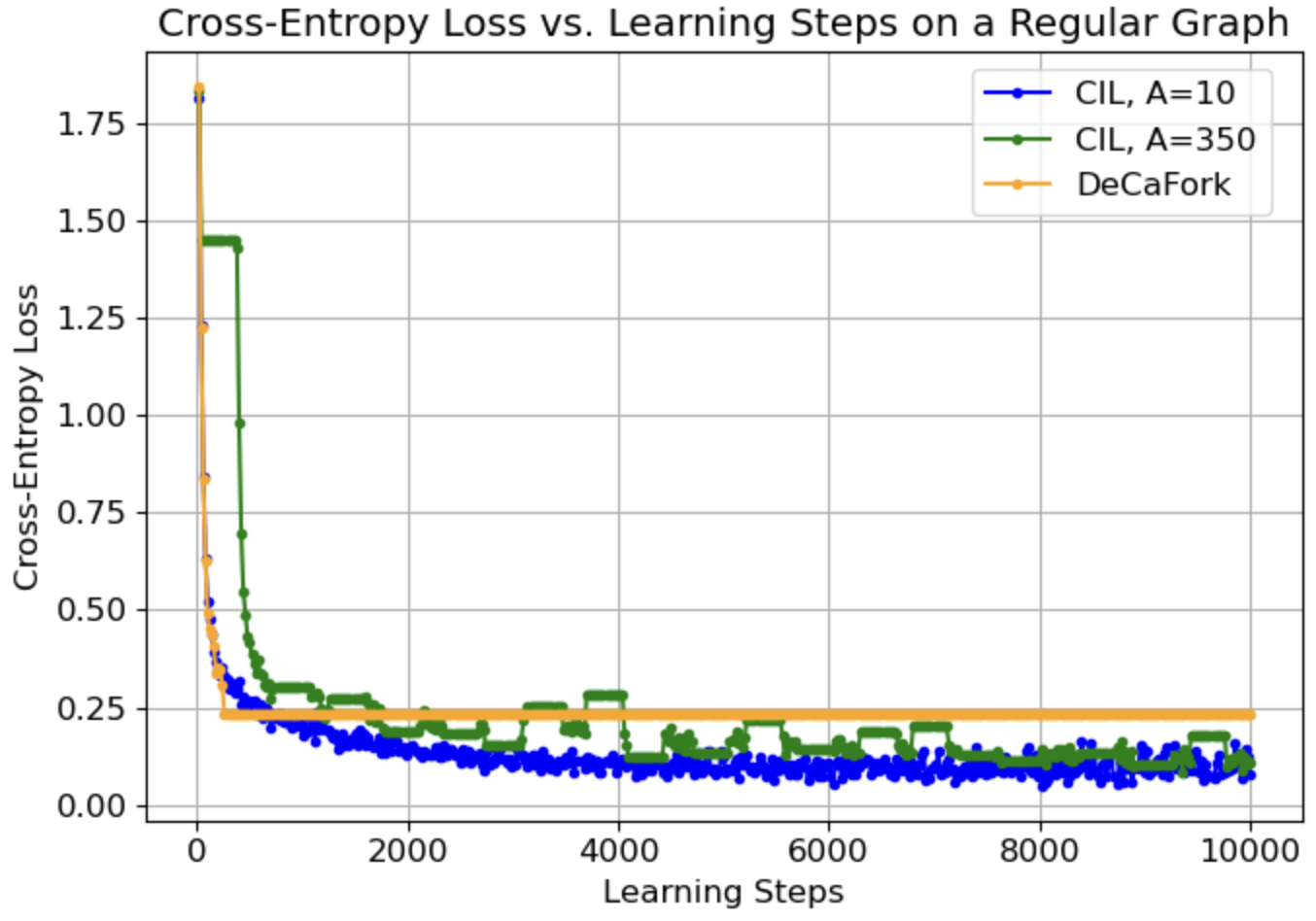}
    \caption*{(b) random regular graph}
  \end{minipage}
  \hfill
  \begin{minipage}[b]{0.45\linewidth}
    \centering
    \includegraphics[width=\linewidth, height=0.9\linewidth]{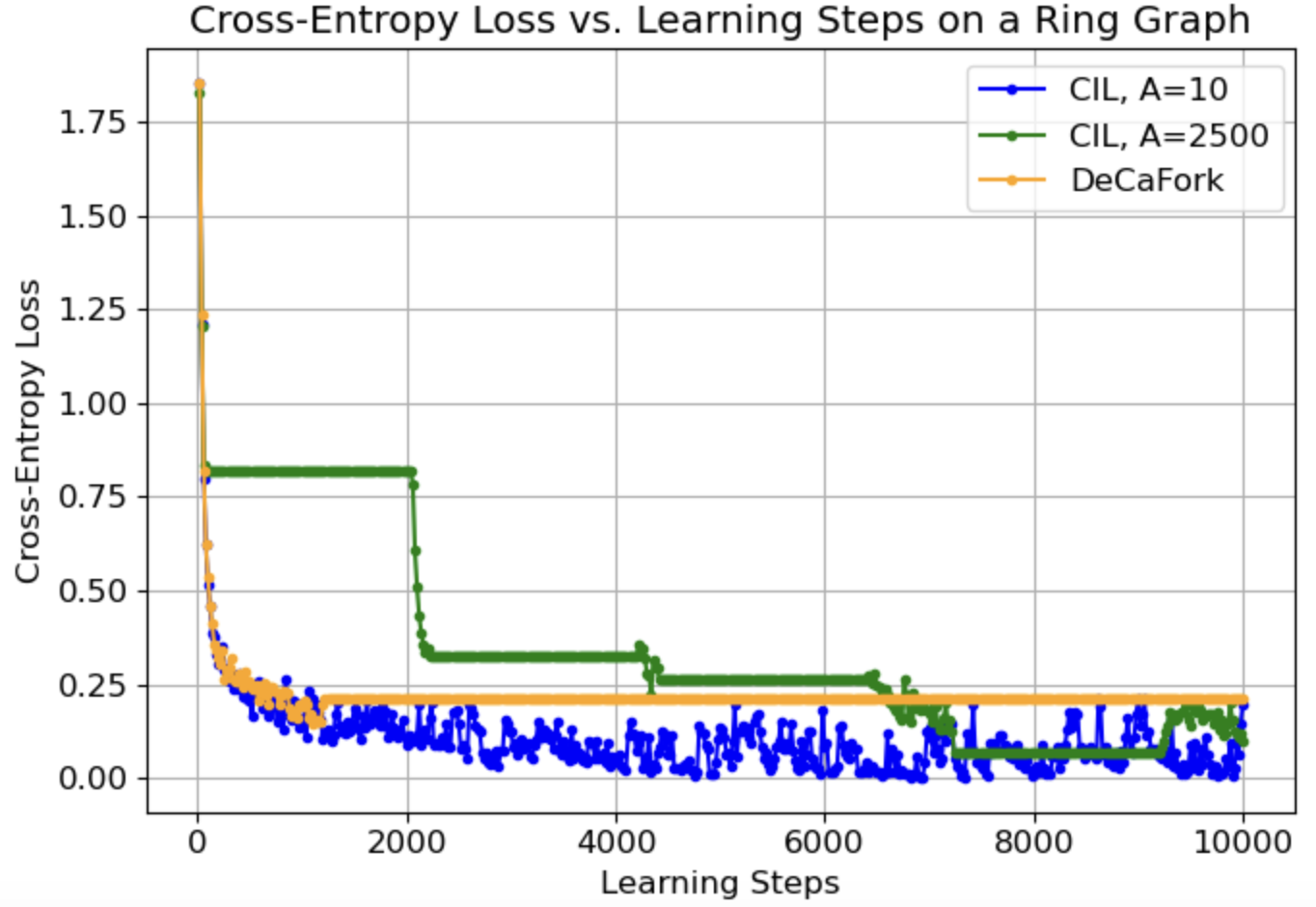}
    \caption*{(c) ring topology}
  \end{minipage}
  \hfill
  \begin{minipage}[b]{0.45\linewidth}
    \centering
    \includegraphics[width=\linewidth, height=0.9\linewidth]{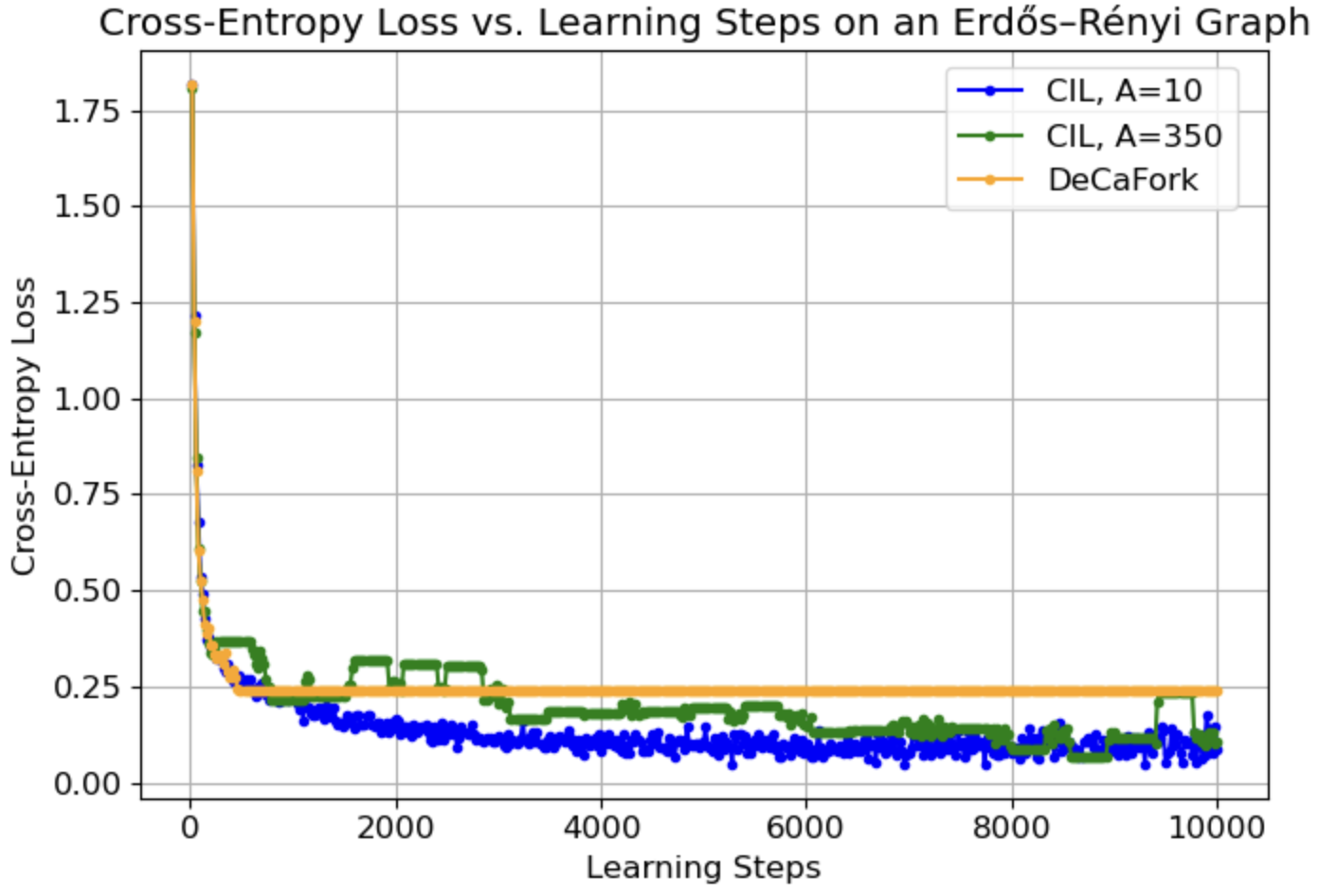}
    \caption*{(c) Erd\H{o}s--R\'enyi graph}
  \end{minipage}
\caption{Loss function v.s. learning steps on different graphs under \iid partitioning of the public benchmark dataset.}
\label{fig:Convergencerealiid}
\end{figure}

\begin{figure}[h]
\centering
  \centering
  \begin{minipage}[b]{0.45\linewidth}
    \centering
    \includegraphics[width=\linewidth, height=0.9\linewidth]{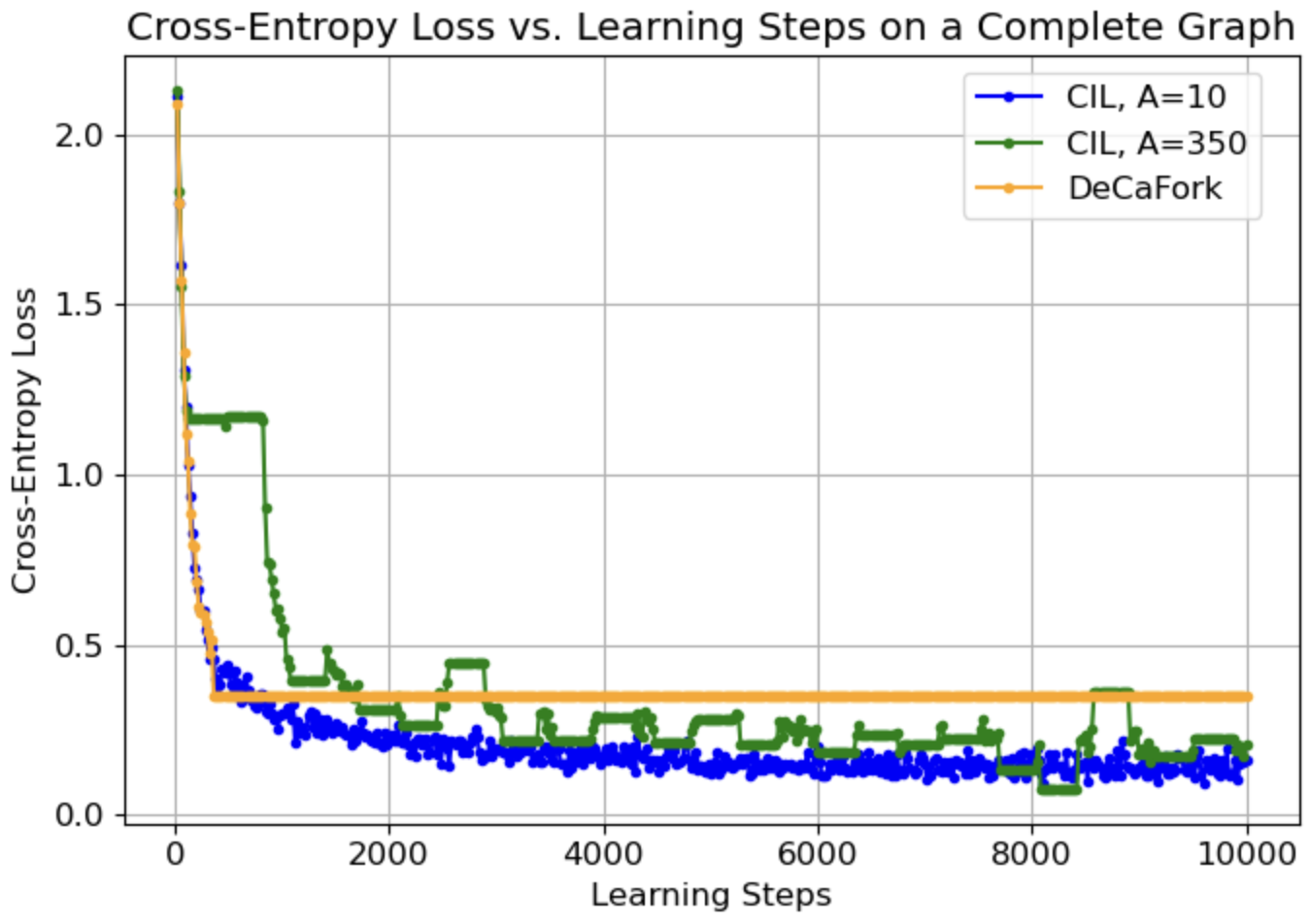}
    \caption*{(a) complete graph}
  \end{minipage}
  \hfill
  \begin{minipage}[b]{0.45\linewidth}
    \centering
    \includegraphics[width=\linewidth, height=0.9\linewidth]{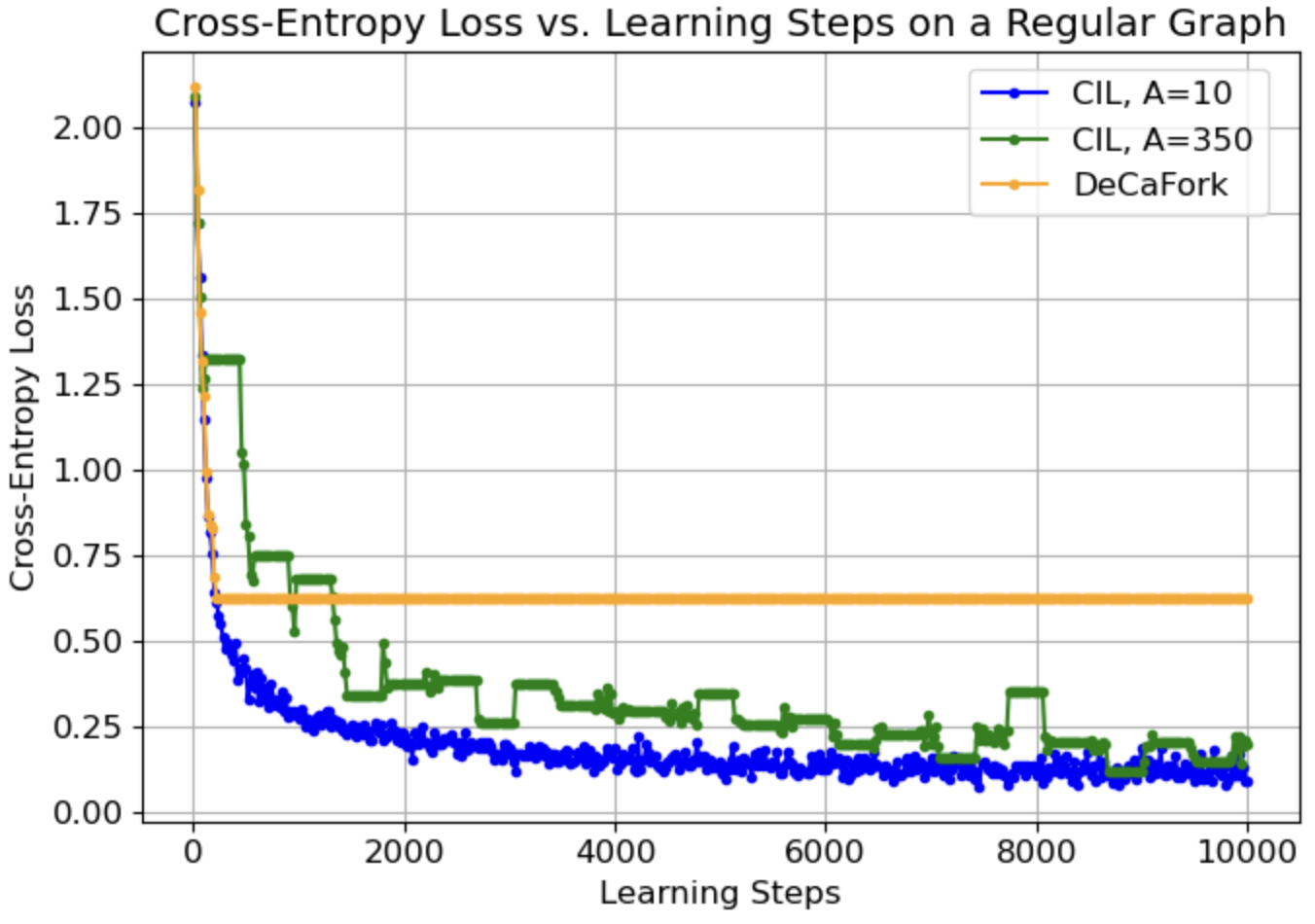}
    \caption*{(b) random regular graph}
  \end{minipage}
  \hfill
  \begin{minipage}[b]{0.45\linewidth}
    \centering
    \includegraphics[width=\linewidth, height=0.9\linewidth]{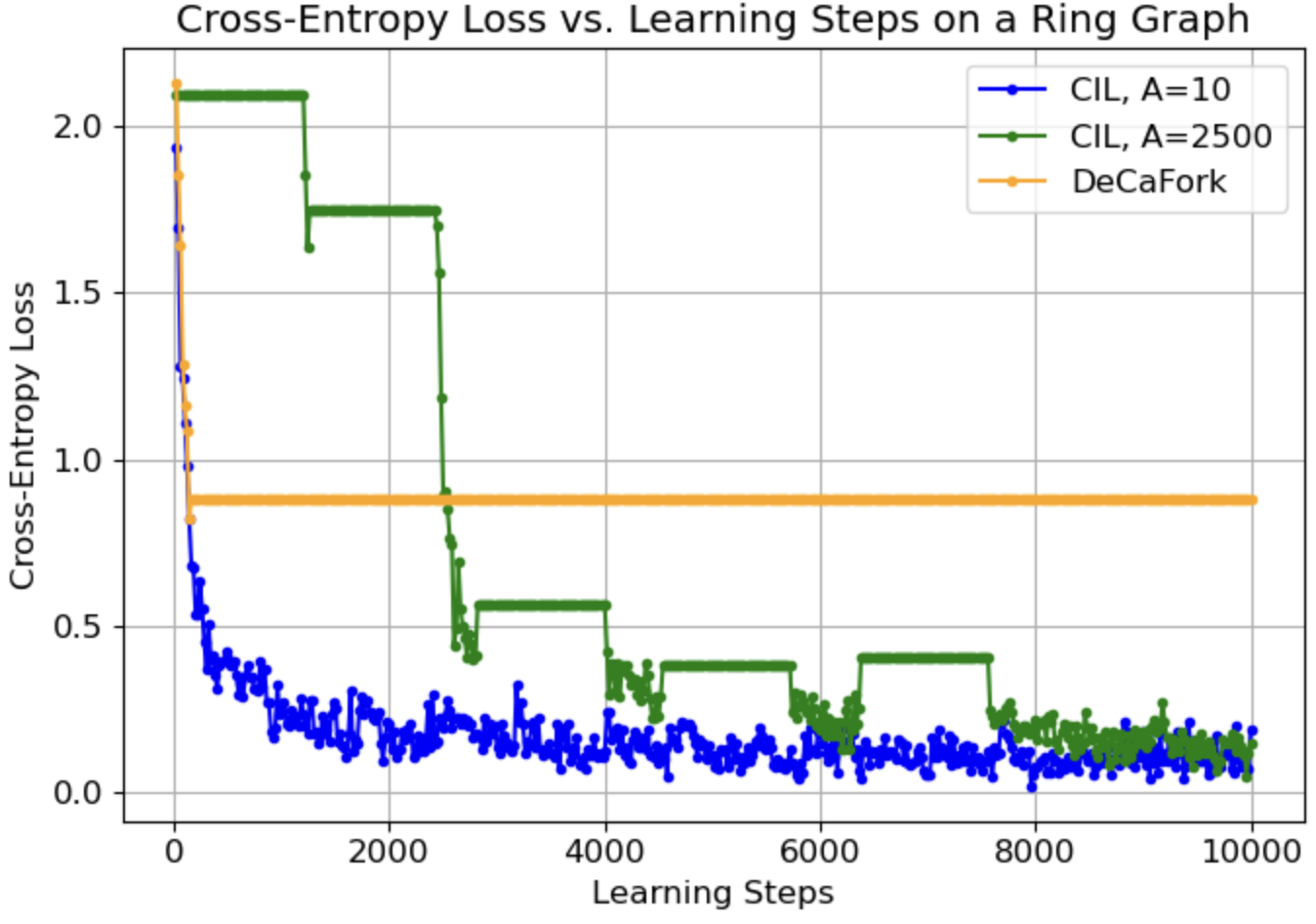}
    \caption*{(c) ring topology}
  \end{minipage}
  \hfill
  \begin{minipage}[b]{0.45\linewidth}
    \centering
    \includegraphics[width=\linewidth, height=0.9\linewidth]{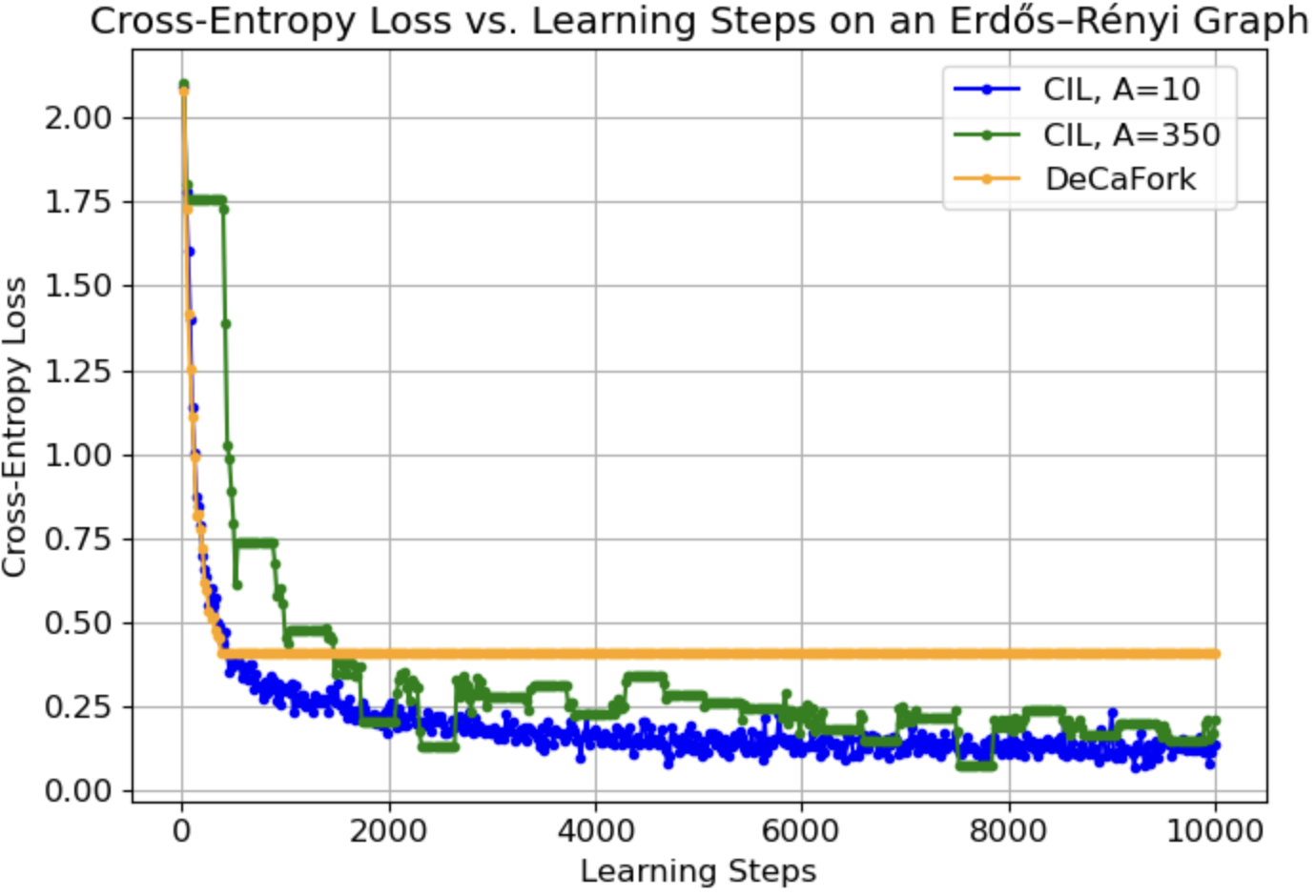}
    \caption*{(c) Erd\H{o}s--R\'enyi graph}
  \end{minipage}
\caption{Loss function v.s. learning steps on different graphs under non-\iid partitioning of the public benchmark dataset.}
\label{fig:ConvergencerealNoniid}
\end{figure}

Fig.~\ref{fig:Convergencerealiid} and Fig.~\ref{fig:ConvergencerealNoniid} show the convergence behaviors of RW-SGD under the \cwl with different creation thresholds ($A=10$ and $A=350$) and the \DeCa, on the public benchmark dataset \cite{deng2012mnist}, using \iid and non-\iid data partitioning across nodes, respectively. 
In each figure, the subplots depict the global loss (on the $y$-axis) over time steps (on the $x$-axis). 
A moving average with a window size of $20$ is applied to smooth the curves and mitigate the variances, to better reveal the underlying trends. 
Across both partitioning settings, the convergence patterns are consistent with those observed on synthetic dataset. The loss curve under \cwl with $A=10$ decreases quickly to near-zero levels and remain stable thereafter, confirming that the algorithm converges. The global loss under \cwl with $A=350$ also converges, but contains several horizontal segments and proceeds at a slower rate. Under the \DeCa algorithm, the RW population may eventually become extinct, in which case no RW remains to perform local updates and the training process terminates prematurely.

\begin{table}[h]
\centering
\begin{tabular}{|c|c|c|c|c|}
\hline
& Complete & Regular & Ring & Erd\H{o}s--R\'enyi \\
\hline
Standard SGD & $0.9727$ & $0.9784$ & $0.9715$ & $0.9763$\\
\hline
CIL, small $A$ & $0.9749$ & $0.9792$ & $0.9736$ & $0.9771$ \\
\hline
CIL, large $A$ & $0.9758$ & $0.9720$ & $0.9767$ & $0.9713$\\
\hline
DeCaFork & $0.9795$ & $0.9742$ & $0.9776$  & $0.9731$ \\
\hline
\end{tabular}
\caption{Testing accuracies on different graphs under \iid partitioning.}
\label{tab:Accuracyiid}
\end{table}

Finally, we evaluate the learning performance by using testing accuracies of the final model.  
Given the target distribution $\pi = \left(\frac{1}{100}, \dots, \frac{1}{100} \right)$, Metropolis-Hastings algorithm \cite{doi:10.1137/08073038X} implies that $\nu^{(1)}$ is uniform over the $99$ benign nodes in complete, random regular, and ring graphs, and approximately uniform in the \ER graph due to its near-regular structure. This indicates that in the presence of a Pac-Man node, active RWs  visit benign nodes uniformly or near-uniformly.

Table~\ref{tab:Accuracyiid} presents the performances of final models under \iid data partitioning. In the ``Standard SGD'' baseline, a single RW performs standard SGD on a graph without a Pac-Man node. From this table, we observe that the performances are nearly identical across all graph types, regardless of whether a Pac-Man node is present. This is because, although a Pac-Man node is present, the benign nodes share the same data distribution (due to the \iid partitioning) and are accessed uniformly (or approximately so) by active RWs. As a result, the active RWs observe an approximately representative sample of the global data distribution, and the final models still perform well.

\begin{table}[h]
\centering
\begin{tabular}{|c|c|c|c|c|}
\hline
& Complete & Regular & Ring & Erd\H{o}s--R\'enyi \\
\hline
Standard SGD & $0.9738$ & $0.9772$ & $0.9729$ & $0.9794$\\
\hline
CIL, small $A$ & $0.9659$ & $0.9661$ & $0.9683$ & $0.9542$ \\
\hline
CIL, large $A$ & $0.9677$ & $0.9634$ & $0.9612$ & $0.9661$\\
\hline
DeCaFork & $0.9655$ & $0.9613$ & $0.9522$  & $0.9638$ \\
\hline
\end{tabular}
\caption{Testing accuracies on different graphs under non-\iid partitioning.}
\label{tab:AccuracyNoniid}
\end{table}

Table~\ref{tab:AccuracyNoniid} shows the results under non-\iid partitioning. Here, the presence of a Pac-Man node impacts performance. Two trends emerge: (i) the performance in the ``Standard SGD'' baseline outperforms those in the Pac-Man cases, as the RW in the baseline has full access to all nodes and thus to the complete data distribution; in contrast, the active RWs in the Pac-Man cases never access the Pac-Man node, and the remaining benign nodes---now with heterogeneous data (due to the non-\iid partitioning)---no longer provide a representative sample of the entire dataset, resulting in degraded performance.
(ii) The performances in the Pac-Man cases (under the \cwl and \DeCa algorithms) are nearly identical across all graph types. This is because their active RWs achieve similar quasi-stationary distributions. As a result, they sample similarly skewed subsets of data, which leads to a similar performances of the final models.

%%%%%%%%%%%%%%%%%%%%%%%%%%%%%%%%%%%%%%%%%%%%%%%%%%%%%%%%%%%%
\section{Conclusion}\label{sec:Conclusion}
%%%%%%%%%%%%%%%%%%%%%%%%%%%%%%%%%%%%%%%%%%%%%%%%%%%%%%%%%%%%
In this work, we focused on RW-based decentralized learning algorithms, which are susceptible to terminations and Pac-Man like attack. 
This stealthy and highly disruptive threat can gradually degrade the decentralized operations without triggering any detectable failure signals. 
We proposed a new creation based algorithm \cwl that maintains RW population, and thereby improving the resilience of RW-based SGD under termination-based Pac-Man attacks. 
The proposed \cwl algorithm is a fully decentralized duplication mechanism based solely on local visitation intervals and hence does not need to depend on system parameter estimation. 
We provided a rigorous analysis of the performance of this algorithm, highlighting its strengths.  
Extensive experiments on synthetic and public benchmark datasets validate our findings.

An interesting future direction is to develop principled methods for selecting the creation threshold $A$ and characterizing the fundamental trade-offs among communication cost, recovery time, and convergence speed. Another interesting future direction is to establish tighter graph-dependent bounds on the peak RW population and better understand the role of graph topology in RW population dynamics. In addition, designing random-walk-based decentralized learning algorithms under combined threat models, where Pac-Man termination attacks coexist with other adversarial behaviors such as Byzantine attacks is a promising avenue for future work.

%\nocite{bookprobability, foley2017yaglom, billingsley1995probability, boyd2004convex, 10619692}

\bibliographystyle{IEEEtran}
\bibliography{IEEEabrv,references}

\vspace{-30pt}

\begin{IEEEbiography}
[{\includegraphics[width=1in,height=1.25in,keepaspectratio,clip]{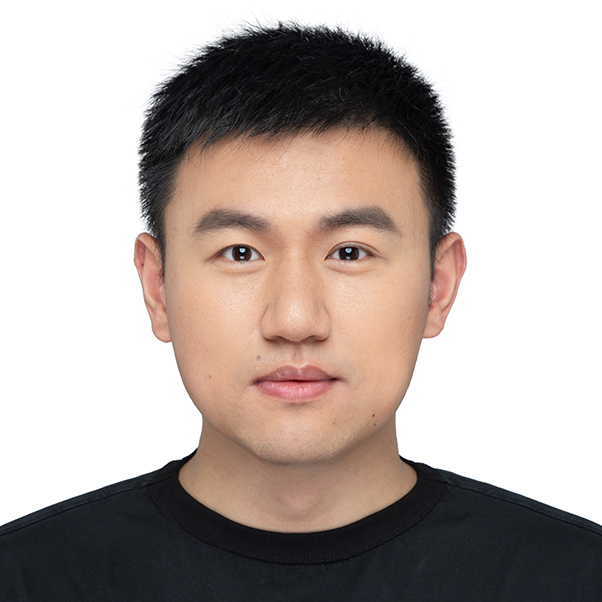}}]
{Xingran Chen (S'18-M'23)}  is a tenure-track Assistant Professor in the Engineering Systems and Design Pillar at Singapore University of Technology and Design (SUTD). Prior to joining SUTD, he was a postdoc at Rutgers University, and an Assistant Professor at the University of Electronic Science and Technology of China. He earned his Ph.D. in Electrical and Systems Engineering from the University of Pennsylvania. Dr. Chen's research focuses on fundamental limits and algorithm design for decentralized learning and decision-making in networked systems. He received the IEEE Communications Society \& Information Theory Society Joint Paper Award in 2023. He has served the research community as an Advisory TPC member for IEEE INFOCOM 2027 and as a TPC member for IEEE ICC 2027, IEEE ITW 2027, WiOpt 2026, IEEE ICC 2024, and IEEE GLOBECOM Workshop 2020.
\end{IEEEbiography}

\vspace{-20pt}

\begin{IEEEbiography}[{\includegraphics[width=1in,height=1.25in,clip,keepaspectratio]{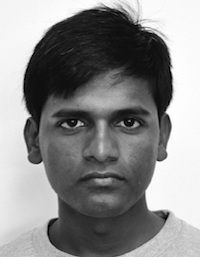}}]{Parimal Parag (S'04-M'11-SM'20)} is an associate professor in the electrical communication engineering department at Indian Institute of Science, where he has been since Dec 2014. 
Prior to that, he was a senior system engineer (R\&D) at Assia Inc. in Redwood City (2011-2014). 
He received Ph.D. degree from the Texas A\&M University in 2011, 
the M.Tech. and B.Tech. degrees in 2004 from IIT Madras, all in electrical engineering. 
He is a coauthor of the 2018 IEEE ISIT student best paper, and a recipient of the 2017 early career award from the Science and Engineering Research Board. 
His research interests include the design and analysis of large scale distributed systems. 
\end{IEEEbiography}

\vspace{-20pt}

\begin{IEEEbiography}
[{\raisebox{20pt}{%
\includegraphics[width=1in,height=1.25in,clip,keepaspectratio]{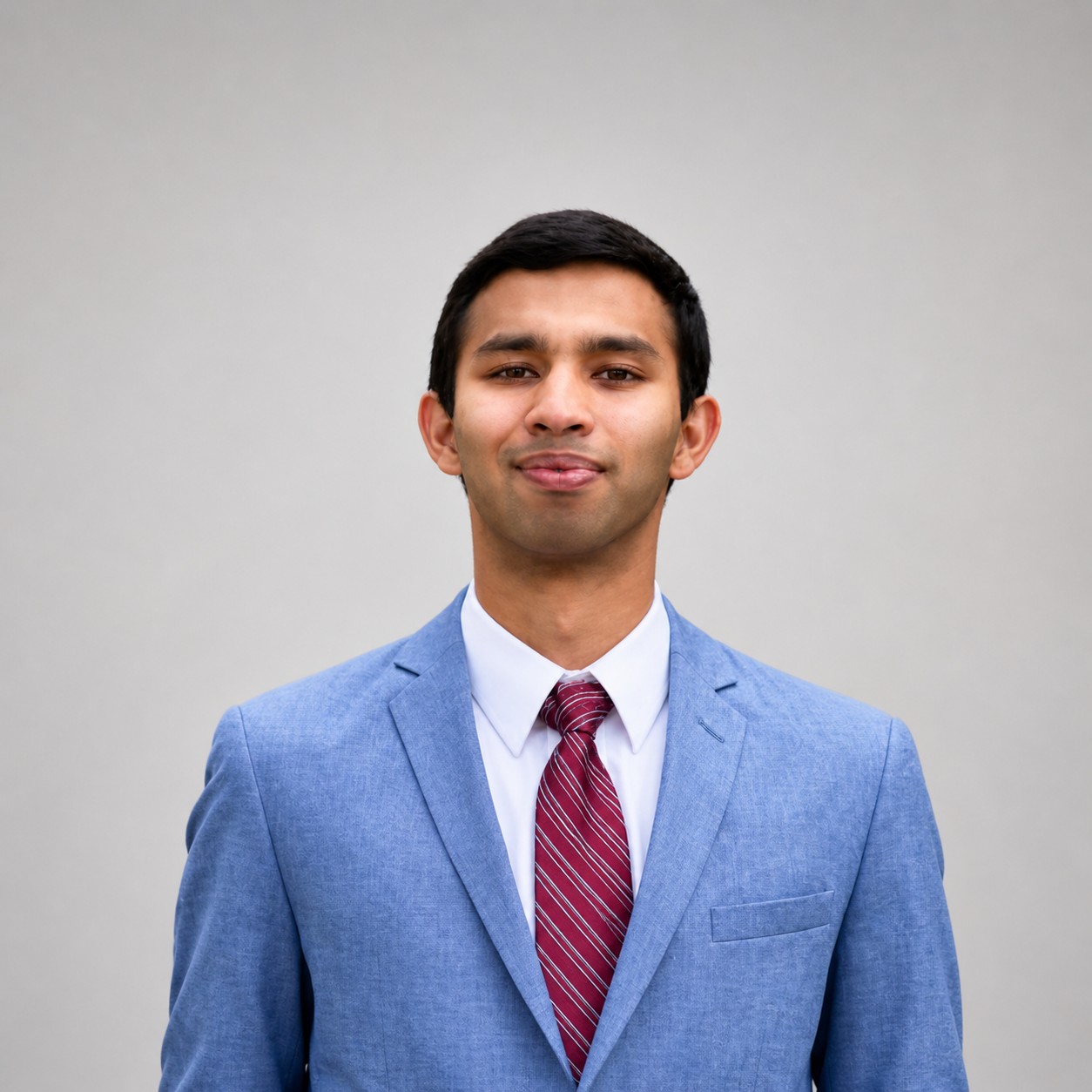}}}]
{Rohit Bhagat} 
is an undergraduate student double majoring in Electrical and Computer Engineering and Computer Science at Rutgers University–New Brunswick. His research interests include federated machine learning, distributed optimization, and adversarial robustness, with a focus on implementing algorithms for robust, scalable decentralized learning systems.
\end{IEEEbiography}

\vspace{-45pt}

\begin{IEEEbiography}
[{\raisebox{5pt}{%
\includegraphics[width=1in,height=1.25in,
clip,keepaspectratio]{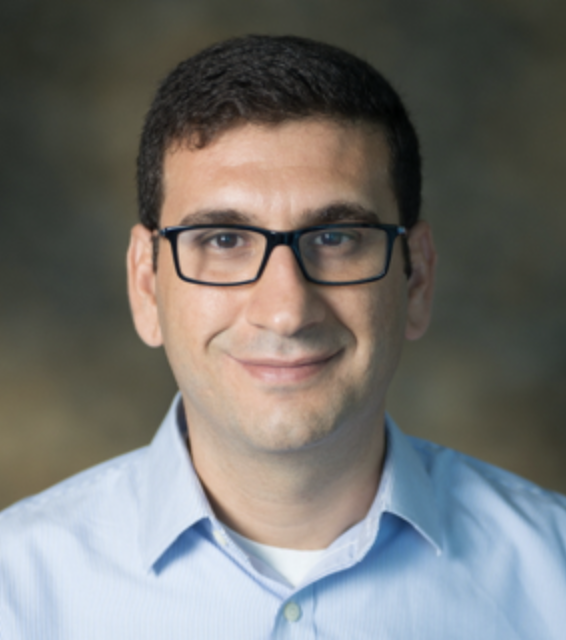}}}]
{Salim El Rouayheb} received the Diploma degree in electrical engineering from the Lebanese University, Lebanon, in 2002, the M.S. degree from the American University of Beirut, Lebanon, in 2004, and the Ph.D. degree in electrical engineering from Texas A\&M University, College Station, TX, USA, in 2009.

He was a Postdoctoral Research Fellow at the University of California, Berkeley, from 2010 to 2011, and a Research Scholar at Princeton University from 2012 to 2013. From 2013 to 2017, he was an Assistant Professor with the Department of Electrical and Computer Engineering at the Illinois Institute of Technology, Chicago, IL, USA. He is currently an Associate Professor with the Department of Electrical and Computer Engineering at Rutgers University, New Brunswick, NJ, USA, where he held the Walter Tyson Junior Faculty Chair in 2019.

His research interests lie at the intersection of information theory, security, privacy, and distributed computing, with particular emphasis on information-theoretic privacy and security, coding for distributed systems, and resilient decentralized and federated learning. He received the NSF CAREER Award in 2016 and the Google Faculty Award in 2018.
\end{IEEEbiography}

%%%%%%%%%%%%%%%%%%%%%%%%%%%%%%%%%%%%%%%%%%%%%%%%%%%%%%%%%%%%
\clearpage
\appendices
%%%%%%%%%%%%%%%%%%%%%%%%%%%%%%%%%%%%%%%%%%%%%%%%%%%%%%%%%%%%

%%%%%%%%%%%%%%%%%%%%%%%%%%%%%%%%%%%%%%%%%%%%%%%%%%%%%%%%%%%%
\section{Proof of Theorem~\ref{thm:FiniteRWs}}\label{Appe:FinitNum}
%%%%%%%%%%%%%%%%%%%%%%%%%%%%%%%%%%%%%%%%%%%%%%%%%%%%%%%%%%%%
\noindent\textbf{Roadmap}. The proof proceeds as follows. We first introduce several useful definitions that provide the necessary framework for the subsequent analysis (Definitions~\ref{defn:Initdistribution}--\ref{defn:BirthDeathRW}). We then establish upper bounds on the drift of $Z_t$ (Lemma~\ref{lem:DriftZ} and Corollary~\ref{cor:Lyapunov}). These drift inequalities are tailored to our model. Using this drift bound, we construct a Lyapunov function $V$ (Definition~\ref{defn:SupMart}) and show that the process $M_k = V(Z_{t_0+dk})$ is a supermartingale (Lemma~\ref{lem:SupMart}). Finally, combining the drift estimate with the supermartingale property, we obtain the boundedness result in Theorem~\ref{thm:FiniteRWs}.

\subsection{Useful Definitions}

\begin{definition}\label{defn:Initdistribution}
Let $\Prb_u$ denote the probability measure under which the RW $j$ starts at node $u$, i.e., $X_{j}(0)=u$. Given a distribution $\nu$ over the node set $\cV$, we define the mixed law %$\Prb_{\nu}$
\begin{align*}
\Prb_{\nu} \triangleq \int_{u\in \cV}\Prb_u d\nu(u),
\end{align*}
which corresponds to initializing RW $j$ according to $\nu$.
\end{definition}

\begin{definition}\label{defn:NaturalFilt}
Consider the number of active RWs $Z_t$ at time $t \in \Z_+$. The natural filtration for the random sequence $Z \triangleq (Z_t: t \in \Z_+)$ is denoted by $\cF_\bullet \triangleq (\cF_t: t \in \Z_+)$ where $\cF_t \triangleq \sigma(Z_s, s \le t)$. 
\end{definition}
Consider the graph $\cG^\prime$ defined in Definition \ref{defn:PacMan}. Since the original graph $\cG$ is connected and finite, then the Pac-Man is {\it reachable} from any other node $u \in \cB$, i.e., there exists $n_u \in \N$ such that $(P^\prime)^{n_u}_{u1} > 0$.
\begin{definition}
\label{defn:MinPathPacMan} 
We define the smallest number of steps to reach the Pac-Man from node $u\in\cB$ as 
\begin{align*}
d_u \triangleq \inf\set{n \in \N: (P^\prime)^n_{u1} > 0}.    
\end{align*}
We define the maximum of the minimum time steps to reach the Pac-Man from node $u \in \cB$ as 
\begin{align}\label{eq:MaxMinPath}
d \triangleq \max_{u \in \cB}d_u.    
\end{align}
Accordingly, we define the smallest probability of reaching the Pac-Man within $d$ steps from node $u \in \cB$, as 
\begin{align*}
c \triangleq \min_{u \in [N]}(P^\prime)^{d_u}_{u1}. 
\end{align*}
\end{definition}
\begin{remark}
Since the Pac-Man is reachable, $n_u$ is finite for each $u\in[N]$, then $d_u \le n_u$ and hence is finite. By Definition~\ref{defn:PacMan}, $\left(P^\prime_{u1}\right)^{d_u} >0$ for each $u \in \cB$, and then $c$ is positive.  Thus, $d$ and $c$ are well-defined.
\end{remark}
\begin{definition}\label{defn:BirthDeathRW}
Consider a finite connected graph $\cG$ with $d$ defined in \eqref{eq:MaxMinPath}. 
During a fixed finite and half-open time interval $T \subseteq \R_+$, we denote the number of random walks that hit PacMan by $D_T$ and the number of RWs generated by $G_T$.  
\end{definition}

\subsection{Upper Bounds on the Drift of $Z_t$}
\begin{lemma}\label{lem:DriftZ}
Consider $Z_t$ independent, aperiodic, active RWs at time $t$, each following the identical law over a finite connected graph $\cG$ with $d$ defined in \eqref{eq:MaxMinPath}. 
Then, 
\begin{align}\label{eqn:MeanRWUB}
\E[Z_{t+d}-Z_t\mid\cF_t] \le  - c\zeta Z_t + (N-1)d.
\end{align} 
\end{lemma}
\begin{IEEEproof} 
From the definition of $D$ and $G$ from Definition \ref{defn:BirthDeathRW}, we can write the difference in the number of active RWs at time $t+d$ and $t$ as 
\begin{align}\label{eqn:DiffRW}
Z_{t+d} = Z_t - D_{(t, t+d]} + G_{(t, t+d]}.
\end{align}
We note that the graph $\cG^\prime$ has $N-1$ benign nodes, and at most one RW can be generated at each node at each time $t$. 
Thus, %we have
\begin{align}\label{eqn:BirthRW}
G_{(t, t+d]} \le \sum_{s=0}^{d-1}\left(Z_{t+s}\wedge (N-1)\right) \le (N-1)d.
\end{align}
At time $t$, there are $Z_t$ active RWs. 
From Definition~\ref{defn:MinPathPacMan}, homogeneity of all RWs with transition probability matrix $P^\prime$, and the definition of $\Prb_u$ in Definition~\ref{defn:Initdistribution}, we have for any active RW $j \in \cZ_t$ %RW $X_j$ ,
\begin{align*}
&\Prb\left(\cup_{n=1}^d\set{X_j(t+n) = 1}\mid\cF_t\right) \\
&\ge \Prb_{X_j(t)}\set{X_j(d_u) = 0} \ge c. 
\end{align*}
%where $\Prb_u$ is defined in Definition~\ref{defn:Initdistribution}. 
That is, $c$ is the uniform lower bound on the probability of ending up at the Pac-Man \emph{within} $d$ steps, over all possible initial positions. 
It follows that the number of deaths for RWs is lower bounded by the number of active RWs at time $t$ hitting Pac-Man (ignoring the RWs generated during this interval and hitting Pac-Man), and hence 
\begin{align}\label{eqn:DeathRW}
\E[D_{(t,t+d]}\mid \cF_t] \ge c\zeta Z_t.
\end{align}
Taking conditional expectation of \eqref{eqn:DiffRW} given history $\cF_t$, substituting the upper bound on the conditional mean number of births \eqref{eqn:BirthRW} and the lower bound on the conditional mean number of deaths \eqref{eqn:DeathRW}, we obtain the result. 
\end{IEEEproof}
\begin{cor}\label{cor:Lyapunov} 
Consider independent aperiodic RWs on a finite connected graph $\cG$ with identical probability laws and $d$ defined in \eqref{eq:MaxMinPath}. 
For any $\epsilon > 0$, there exists positive constants $b, B$, such that $B > b$ and the random sequence $Z$ satisfies the following conditions.
\begin{compactenum}[(a)]
\item If $Z_{t}\le B$, then $\E[Z_{t+d}\mid\cF_{t}]\le b$.
\item If $Z_{t} > B$, then $\E[Z_{t+d}-Z_{t}\mid\cF_{t}]<-\epsilon$. 
\end{compactenum}
\end{cor}
\begin{IEEEproof}
Let $\epsilon > 0$, $N-1$ be the number of benign nodes in $\cG^\prime$, and $c$ be the uniform lower bound on hitting Pac-Mac from any node in minimum number of steps as defined in Definition~\ref{defn:MinPathPacMan}. 
We define positive constants $B \triangleq \frac{1}{c\zeta}\left((N-1)d+\epsilon\right)$ and $b \triangleq (1-c\zeta)B+ (N-1)d$, such that $b < B$. 
\begin{compactenum}[(a)]
\item Let $Z_t \le B$. 
It follows from \eqref{eqn:MeanRWUB}, that
\begin{align*}
\E[Z_{t+d}\mid \cF_t]\le (1-c\zeta)B + (N-1) d = b. 
\end{align*}
\item Let $Z_t > B$. It follows from \eqref{eqn:MeanRWUB} and definition of $B$, that 
\begin{align*}
&\E[Z_{t+d}-Z_t\mid\cF_t] \le - c\zeta Z_t + (N-1)d < -\epsilon.
\end{align*}
\end{compactenum}
\end{IEEEproof}

\subsection{Supermartingale}
\begin{definition}[Supermartingale] 
\label{defn:SupMart}
Consider independent aperiodic RWs on a finite connected graph $\cG$ with identical probability laws and $d$ defined in \eqref{eq:MaxMinPath}, $\epsilon > 0$, and positive constants $b, B$ defined in Corollary~\ref{cor:Lyapunov}. For $t_0\ge 0$ and $k \in \Z_+$, we define periodic samples of number of active RWs and its natural filtration at time $t_0+dk$ as 
We define a Lyapunov function $V:\R_+\to\R_+$ for each random variable $Z\in \R_+$
\begin{align}\label{eqn:Lyapunov}
V(Z) \triangleq Z\Ind{\set{\E[Z]>B}} + B\Ind{\set{\E[Z]\le B}}. 
\end{align}
For $t_0\ge 0$ and $k \in \Z_+$, we define periodic samples of number of active RWs and its natural filtration at time $t_0+dk$ as  
\begin{xalignat*}{2}
&M_k \triangleq V(Z_{t_0+dk}),&&\cH_k \triangleq \cF_{t_0+dk}.
\end{xalignat*}
We define a random sequence $M \triangleq (M_k: k \in \Z_+)$ and filtration $\cH_\bullet \triangleq (\cH_k:k\in \Z_+)$. 
\end{definition}
\begin{lemma}\label{lem:SupMart}
Sequence $M$ is a supermartingale adapted to filtration $\cH_\bullet$.  
\end{lemma}
\begin{IEEEproof}
We first observe that $Z_t$ is $\cF_t$ measurable, and hence $M_k$ is a $\cH_k$ measurable by definition. For each $k \in \Z_+$, we can define the following $\cH_k$ measurable events
\begin{align*}
&\cA_{k,1} \triangleq \set{Z_{t_0+dk}=z_{t_0+dk}\le B, \E[Z_{t_0+dk}] \le B},\\
&\cA_{k,2} \triangleq \set{Z_{t_0+dk} =z_{t_0+dk} > B, \E[Z_{t_0+dk}] \le B},\\
&\cA_{k,3} \triangleq \set{Z_{t_0+dk}=z_{t_0+dk} \le B, \E[Z_{t_0+dk}] > B},\\
&\cA_{k,4} \triangleq \set{Z_{t_0+dk} =z_{t_0+dk} > B, \E[Z_{t_0+dk}] > B}.
\end{align*}
In terms of these events, we can write the conditional mean as 
\begin{align*}
&\E[(M_{k+1}-M_k)\mid\cH_k] \\
&= \E[(M_{k+1}-M_k)(\Ind{\cA_{k,1}}+\Ind{\cA_{k,2}}+\Ind{\cA_{k,3}}+\Ind{\cA_{k,4}})\mid\cH_k]. 
\end{align*}
Let us compute $\E[(M_{k+1}-M_k)\Ind{\cA_{k,1}}\mid\cH_k]$ as an example: from Corollary~\ref{cor:Lyapunov} and definition \eqref{eqn:Lyapunov},
if $\Ind{\cA_{k,1}}=1$, then $z_{t_0+dk}\le B$, thus $\E[Z_{t_0+d(k+1)}\mid\cH_k] = b < B$, we have $M_{k+1} = B$. In addition, since $\Ind{\cA_{k,1}}=1$, then $\E[Z_{t_0+dk}]\le B$, thus $M_{k} = B$, and 
\begin{align*}
\E[(M_{k+1}-M_k)\Ind{\cA_{k,1}}\mid\cH_k] = B - B = 0.
\end{align*}
Similarly, we have:
\begin{align*}
&\E[(M_{k+1}-M_k)\Ind{\cA_{k,2}}\mid\cH_k] < -\epsilon,\\
&\E[(M_{k+1}-M_k)\Ind{\cA_{k,3}}\mid\cH_k] < 0,\\
&\E[(M_{k+1}-M_k)\Ind{\cA_{k,4}}\mid\cH_k] < -\epsilon.
\end{align*}
Combining the these results, we get the result.
\end{IEEEproof}

\subsection{Proof of Theorem~\ref{thm:FiniteRWs}}

% We define a stopping time $\tau_0 \triangleq \inf\set{t \in \Z_+: Z_t = 0}$. If $\tau_0<\infty$, then $Z_t = 0$ for all $t \ge \tau_0$, which implies $\limsup_{t\to\infty}Z_t<\infty$.Therefore, without loss of generality, we consider the case when $Z_t > 0$ for any finite time $t$. 

From the definition of sequence $M$ and filtration $\cH_\bullet$ in Definition \ref{defn:SupMart}, Lemma \ref{lem:SupMart}, and positivity of $Z$, we observe that $M$ is a positive supermartingale adapted to filtration $\cH_\bullet$. 
By the Doob’s supermartingale convergence Theorem \cite{bookprobability}, supermartingale $M$ converges to a limit $M_\infty$ almost surely, i.e.  
\begin{align*}
\lim_{k\to\infty}M_k = M_\infty <\infty,\,\,a.s.
\end{align*}
From the definition of supermartingale $M$ in Definition \ref{defn:SupMart}, it follows that for any $t_0\ge 0$, 
\begin{align*}
\limsup_{k\to\infty}Z_{t_0+d k}\le \max\set{B, M_\infty}<\infty.
\end{align*}
Since the choice of $t_0 \in \Z_+$ was arbitrary, we have 
\begin{align*}
\limsup_{t\to\infty}Z_t <\infty,\,\,a.s.
\end{align*}

\section{Proof of Theorem~\ref{thm:Peak}}\label{Appe:Peak}

In a complete graph and under a uniform target sampling distribution, for any benign nodes $u, v\in\cB$, the transition probability matrix is uniform, $P_{uv}=\frac{1}{N}$.
Given $Z_t=z$ at time slot $t$, for each node $u$, let $I_{u,t}=1$ denote the indicator that node $u$ is \textit{not} visited by any RW at time slot $t$. Then:
\begin{align}
\E\left[I_{u, t} \,\middle|\, Z_t=z\right] = (1-\frac{1}{N})^z.
\end{align}
Denote by $C_t$ the number of benign nodes that are not visited by any RW at time slot $t$. Then,
\begin{align}\label{eq:ExpectedC}
\E\left[C_t \,\middle|\, Z_t=z\right] =& \E\left[\sum_{u\in\cB}I_{u, t} \,\middle|\, Z_t=z\right]\nonumber\\
=&(N-1)(1-\frac{1}{N})^z.
\end{align}
Recall that $A_u=1$ for all $u\in\cB$. According to \eqref{eq:ExpectedC}, the one-step drift of $Z_t$ satisfies:
\begin{align}\label{eq:Driftoneslot}
&\E\left[Z_{t+1} - Z_t \,\middle|\, Z_t=z\right]\nonumber \\
&= q(N-1)(1-\frac{1}{N})^z - z\frac{\zeta}{N}.
\end{align}
%\red{$\E[Z_t\mid Z_t = z] = z$.}
Since $(1-\frac{1}{N})^z\le 1$, then
\begin{align}\label{eq:DriftUB}
&\E\left[Z_{t+1} \,\middle|\, Z_t=z\right]\le q(N-1) + z(1-\frac{\zeta}{N}).
\end{align}
Taking expectation on both sides with respect to $Z_t$ yields
\begin{align}\label{eq:DriftRecur}
\E\left[Z_{t+1}\right] \le q(N-1) + \E\left[Z_{t}\right](1-\frac{\zeta}{N}).
\end{align}
By repeatedly applying the recursion in \eqref{eq:DriftRecur}, we obtain:
\begin{align}\label{eq:recurZt}
\E[Z_t]\le \frac{q}{\zeta}(N-1)N + (1-\frac{\zeta}{N})^t\left(z_0-\frac{q}{\zeta}(N-1)N\right).
\end{align}
Therefore, if $z_0\ge\frac{q}{\zeta}(N-1)N$, the upper bound 
\begin{align*}
\frac{q}{\zeta}(N-1)N + (1-\frac{\zeta}{N})^t\left(z_0-\frac{q}{\zeta}(N-1)N\right)
\end{align*}
decreases with time $t$, then 
\begin{align*}
\bar{Z}^\star=\sup_{t\ge0}\E[Z_t]\le z_0;
\end{align*}
if $z_0<\frac{q}{\zeta}(N-1)N$, the upper bound 
\begin{align*}
\frac{q}{\zeta}(N-1)N + (1-\frac{\zeta}{N})^t\left(z_0-\frac{q}{\zeta}(N-1)N\right)
\end{align*}
increases with time $t$, then 
\begin{align*}
\bar{Z}^\star=\sup_{t\ge0}\E[Z_t]\le \frac{q}{\zeta}(N-1)N.
\end{align*}
It follows that
\begin{align*}
\bar{Z}^\star\le\max\set{z_0, \frac{q}{\zeta}(N-1)N}\le \max\set{z_0, \frac{q}{\zeta}N^2}.
\end{align*}

\section{Proof of Theorem~\ref{thm:Effectiveness}}\label{Appe:Effectiveness}

To clarify the idea of the proof, we condense the time interval between the termination of the last RW and the creation of the next RW. Since each new RW is an identical copy of the last visited RW, such intervals only extend the waiting time horizon and do not affect convergence.

In the following proof, for any benign node $u$, if its visiting time is $t_u$ and a new RW---denoted by index $j'$---is generated at time $t_u+A_u$, i.e., $X_{j'}(t_u+A_u)=u$, \textit{then, based on the above discussion, we remove the waiting time. Consequently, $X_{j'}(t_u+A_u)=u$ degenerates to $X_{j'}(t_u)=u$}.

We consider an infinite chain of RWs $\set{j_s}_s$. Let $\cA$ be the set of absorbing states, i.e., $\cA=\set{1, w}$ if $\zeta=1$ and $\cA=\set{w}$ if $\zeta\in(0, 1)$.
Let $u\in\cB$ denote the initial location of RW $j_0$. Let $\nu$ be a probability measure on $\cB$, and $\Prb_u$, $\Prb_{\nu}$ be defined in 
Definition~\ref{defn:Initdistribution}. The stopping times of RW $j_0$ with respect to $\cA$, starting from $X_{j_0}(0)=u$ and $X_{j_0}(0)\sim\nu$, are defined as:
\begin{align}
K_{u} \triangleq& \inf\set{t > 0: X_{j_0}(t) \in \cA, X_{j_0}(0) = u},\label{eq:StoppingT}\\
K_{\nu} \triangleq& \inf\set{t > 0: X_{j_0}(t) \in \cA, X_{j_0}(0)\sim\nu}.\label{eq:StoppingTProb}
\end{align}
We have $\Prb_{u}(K_{u}<\infty)=1$ and $\Prb_{\nu}(K_{\nu}<\infty)=1$.

\begin{definition}\label{defn:Activej0}
(Active distribution) Consider a strongly connected graph $\cG$ with absorbing states $\cA$, as defined in Definition~\ref{defn:StrongConnect}. Let a chain of RWs $\set{j_s}_{s\ge0}$ be defined in Definition~\ref{defn:ChainRWs}. Let $K_u$ be defined in \eqref{eq:StoppingT}. For any $t$ and $I\subset\cB$, we define the active distribution of RW $j_0$ at time $t$ as
\begin{align}\label{eq:InitialActiveProb}
\xi_{0; t}(I; u) \triangleq \Prb_u\big(X_{j_0}(t)\in I\mid K_u>t\big).
\end{align}
Let $t_s$ denote the birth time of RW $j_s$, and suppose that its initial location $X_{j_s}(t_s)$ is drawn from a distribution $\nu_s$, which depends on $u$. For any $t\ge t_s$,
we define the active distribution of RW $j_s$ as
\begin{align}\label{eq:ActiveProbs}
\xi_{s; t}(I; u) \triangleq \Prb_{\nu_s}\big(X_{j_s}(t)\in I\mid K_{\nu_s}>t-t_s\big),
\end{align}
where the subscript $\nu_s$ emphasizes that the RW is initialized according to $\nu_s$. The dependence of $\nu_s$ on  $u$ is implicit in this notation $\xi_{s; t}(I; u)$.
\end{definition}
Couple the newly created RW with its parent RW such that, after creation, the new RW is independently reinitialized as an \iid replica of the original RW. That is, it evolves independently and has the same probability distribution as the original RW. Consequently, at any time while the parent RW remains active, the probability distribution of the newly created RW coincides with that of an independent copy of the parent RW. \textit{Recall that we remove the waiting time}, applying this argument recursively, at any time $t$, the active probability distribution of any active RW coincides with that of any ancestor, as long as the ancestor remains active.
Therefore, for any $s > 0$, we have
\begin{align}\label{eq:SameDist}
\xi_{s;t}\overset{d}{=}\xi_{0; t}.
\end{align}
At any time $t$, we re-parameterize the active distribution of the the most recently created (i.e., latest-born) RW as $\xi_t$. Since the chain $\set{j_{s}}_{s\ge0}$ is infinite, we now analyze limiting behavior of the probability distribution $\xi_t$. From \eqref{eq:InitialActiveProb} and \eqref{eq:SameDist}, we have
\begin{align}\label{eq:LatestProb}
\lim_{t\to\infty}\xi_t(I; u) =& \lim_{t\to\infty}\xi_{0; t}(I; u)\nonumber\\ =& \lim_{t\to\infty}\Prb_u\big(X_{j_0}(t)\in I\mid K_u>t\big).
\end{align}
If the limit in \eqref{eq:LatestProb}
exists, it is referred to as the Yaglom limit \cite{foley2017yaglom}. This limit depends on the initial location. Intuitively, the Yaglom limit captures the long-term distribution of the process conditioned on survival. It remains to show that the limit $\lim_{t\to\infty}\xi_t(I;u)$ exists and to derive its explicit expression. 

\begin{definition}\label{defn:DSQ}
(Quasi-Stationary Distribution \cite{collet2012quasi}) 
Consider a strongly connected graph $\cG$ with absorbing states $\cA$, as defined in Definition~\ref{defn:StrongConnect}. Let $K_u$ be defined in \eqref{eq:StoppingT}. We say that $\nu$ is a quasi-stationary distribution (QSD) of RW $j_0$ if, for all $t\ge 0$ and any set $I\subset \cB$, 
\begin{align*}
\nu(I) = \Prb_{\nu}(X_{j_0}(t)\in I\mid K_{u}>t).
\end{align*}
\end{definition}

The following Lemma~\ref{lem:QSDchain} shows that the distribution of a chain of RWs converges asymptotically to that of a single RW conditioned on long-term survival. This provides a way to obtain the explicit expression of $\lim_{t\to\infty}\xi_t(I;u)$.

\begin{lemma}\label{lem:QSDchain}
Consider a robustly connected graph $\cG$ with absorbing states $\cA$, as defined in Definition~\ref{defn:StrongConnect}. Let $\set{j_{s}}_{s\ge0}$ be an infinite chain, as defined in Definition~\ref{defn:ChainRWs}. Suppose the initial RW $j_0$ starts at node $u\in\cB$.  We define the distribution of the chain at time $t$ as
\begin{align}\label{eq:ChainDist}
\pi_{\text{chain},t} \triangleq \xi_t.
\end{align} 
Let $t\to\infty$, the distribution of a chain is convergent:
\begin{align}\label{eq:QSD}
\lim_{t\to\infty}\pi_{\text{chain}, t} = \nu^{(\zeta)}
\end{align}
where $\nu^{(\zeta)}$ is the  left normalized leading eigenvector of $Q^{(\zeta)}$ \big(as defined in \eqref{eq:NewTransMatCase1} and \eqref{eq:NewTransMatCase2}\big).
\end{lemma}
\begin{IEEEproof}
From Definition~\ref{defn:StrongConnect}, the submatrix $Q^{(\zeta)}$ is irreducible. By \cite[Section~2]{foley2017yaglom} or \cite[Theorem~16.11]{billingsley1995probability}, the irreducibility of $Q^{(\zeta)}$ ensures the existence of the corresponding Yaglom limits \big(see \eqref{eq:LatestProb}\big), which is convergent in total variation.

Moreover, since $Q^{(\zeta)}$ is irreducible and aperiodic, any existing Yaglom limit (with any initial state $u$)
coincides with a QSD, as established in \cite[Proposition~1]{foley2017yaglom}. Therefore, the Yaglom limit in \eqref{eq:LatestProb} is a QSD for every  $u\in\cB$.

In our case, each initial RW is defined on a finite state space $\cV$ with a nonempty absorbing set $\cA$. The restricted transition matrix $Q^{(\zeta)}$ on the transient states $\cB$ is reducible and aperiodic. According to \cite{darroch1965quasi}, the QSD exists and is unique. As a result, the Yaglom limit in \eqref{eq:LatestProb} converges to the same QSD for all initial states $u\in\cB$.  

Meanwhile, the QSD can be calculated as the leading left eigenvector of $Q^{(\zeta)}$, normalized to sum to one \cite[Eqn.~(10) and the third equation on p.~99]{darroch1965quasi}. Thus, according to \eqref{eq:ChainDist}, the limiting distribution $\lim_{t\to\infty}\pi_{\text{chain}, t}$ is given by: 
\begin{align*}
\lim_{t\to\infty}\pi_{\text{chain}, t}= \nu^{(\zeta)}.
\end{align*}
\end{IEEEproof}

In each chain of RWs, every child inherits the current model state (i.e., ${\bf x}_t$) from its parent. As a result, under the RW-SGD algorithm, each infinite chain asymptotically behaves as if a single effective RW is solving a surrogate optimization problem with a time-varying sampling distribution $\tilde{\pi}_t$. Specifically: 
\begin{compactenum}
\item When $\zeta=1$, the absorbing state $\cA=\set{1, w}$, so $\tilde{\pi}_t = [0, \pi_{\text{chain},t}]$, where $\pi_{\text{chain};t}$ is a discrete distribution supported on a finite set of size $N$, and
\begin{align*}
\lim_{t\to\infty}\tilde{\pi}_t=[0, \nu^{(1)}].
\end{align*}
\item When $0<\zeta<1$, the absorbing state $\cA=\set{w}$, so $\tilde{\pi}_t = \pi_{\text{chain},t}$, where $\pi_{\text{chain};t}$ is a discrete distribution supported on a finite set of size $N+1$, and 
\begin{align*}
\lim_{t\to\infty}\tilde{\pi}_t= \nu^{(\zeta)}.
\end{align*}
\end{compactenum}

\section{Proof of \eqref{eq:UniqueP}}\label{Appe:UniqueP}
In fact, as discussed before, we condense the time interval between the termination of the last RW and the creation of the next RW,
where $\cA$ represents the set of absorbing states (as defined in Appendix~\ref{Appe:Effectiveness}), i.e., $\cA=\set{1, w}$ if $\zeta=1$ and $\cA=\set{w}$ if $\zeta\in(0, 1)$.
For any nodes $u, v$, the transition probability matrix $P_\text{chain}$ can be written as
\begin{align*}
\left[P_\text{chain}\right]_{uv} = \Prb\left(\set{X_j(1)=v}\mid \set{X_j(0)=u, X_j(1)\notin\cA}\right).
\end{align*}
Taking marginal distribution of the chain at time $0$ as $\mu$ and applying Bayes’ rule, we obtain:
\begin{align*}
\left[P_\text{chain}\right]_{uv} =& \frac{\Prb\left(\set{X_j(1)=v, X_j(0)=u, X_j(1)\notin\cA}\right)}{\Prb\left(\set{X_j(0)=u, X_j(1)\notin\cA}\right)}\nonumber\\
=& \frac{\mu_uQ^{(\zeta)}_{uv}}{\mu_u\sum_{v}Q^{(\zeta)}_{uv}} =\frac{Q^{(\zeta)}_{uv}}{\sum_{v}Q^{(\zeta)}_{uv}}.
\end{align*}

\section{Proof of Proposition~\ref{pro:Bounds}}\label{Appe:Bounds}

\noindent\textit{Proof of Part (1)}.

Since a chain of RWs behaves like a single RW that never dies, we can apply the convergence results of RW-SGD. According to \cite{10.5555/3618408.3618786, doi:10.1137/08073038X, 10.5555/3327546.3327656}, the standard RW-SGD algorithm converges to a deterministic limit when the stepsize $\eta_t$ decreases with the number of iterations and tends to $0$. Consequently, under the same stepsize condition, a chain of RWs converges to the optimizer of the surrogate optimization problem \eqref{eq:SurOpt}.

Let $\tilde{\bf x}^*$ be the optimizer of either \eqref{eq:SurOpt}. Applying strong convexity, we obtain:
\begin{align*}
f({\bf x}^*)\geq f(\tilde{\bf x}^*) + \nabla f(\tilde{\bf x}^*)({\bf x}^* - \tilde{\bf x}^*) + \frac{\mu}{2}\|{\bf x}^* - \tilde{\bf x}^*\|^2,
\end{align*}
which implies
\begin{align*}
0\geq f({\bf x}^*)-f(\tilde{\bf x}^*)\geq \nabla f(\tilde{\bf x}^*)({\bf x}^* - \tilde{\bf x}^*) + \frac{\mu}{2}\|{\bf x}^* - \tilde{\bf x}^*\|^2.
\end{align*}
By Cauchy–Schwarz inequality, it follows that
\begin{align*}
\frac{\mu}{2}\|{\bf x}^* - \tilde{\bf x}^*\|^2\leq - \nabla f(\tilde{\bf x}^*)({\bf x}^* - \tilde{\bf x}^*)\leq \|\nabla f(\tilde{\bf x}^*)\|\|{\bf x}^* - \tilde{\bf x}^*\|.
\end{align*}
Therefore
\begin{align*}
\|\tilde{\bf x}^\star - {\bf x}^\star\|\leq\frac{2}{\mu}\|\nabla f(\tilde{\bf x}^\star)\|.
\end{align*}
The equality holds when $\nabla f(\tilde{\bf x}^\star)$ is co-linear with $\tilde{\bf x}^\star - {\bf x}^\star$ \cite{boyd2004convex}.

Similarly, using the $L$-Lipschitz condition, we derive:
\begin{align*}
L\|{\bf x}^* - \tilde{\bf x}^*\|\geq \|\nabla f({\bf x}^*) - \nabla f(\tilde{\bf x}^*)\|.
\end{align*}
By the optimality conditions, we have:
\begin{align*}
\nabla f({\bf x}^*) = 0,
\end{align*}
which implies
\begin{align*}
\|{\bf x}^* - \tilde{\bf x}^*\|\geq \frac{1}{L}\|\nabla f(\tilde{\bf x}^*)\|.
\end{align*}
The equality holds when $\tilde{\bf x}^\star - {\bf x}^\star$ aligns with the directional of maximal curvature of $f({\bf x})$ \cite{boyd2004convex}.

\noindent\textit{Proof of Part (2)}.
This proof follows the same argument as in \cite[Theorem~1]{10619692}, 
with the necessary substitutions under our setting. 
Specifically, by Theorem~\ref{thm:Effectiveness}, the modified stationary distribution of a single effective RW (i.e., the chain of RWs) 
is $\tilde{\nu}^{(\zeta)}$ with $\zeta\in(0,1]$. 
The corresponding transition probability matrix is given in \eqref{eq:UniqueP}. 
Based on Assumption~\ref{assu:Boundedgrad}, we set $w(u)=1$ for all $u\in\cV$.
By substituting $\tilde{\nu}^{(\zeta)}$, $P_{\text{chain}}$, $\eta_{\text{chain}}$, and the original sampling distribution $\pi$ into the proofs of \cite[Lemmas~1, 2, Theorem~1]{10619692}, and by \textit{artificially condensing the time interval between the child and its parent}, we obtain the following bounds:
\begin{align*}
\E\|\tilde{\bf x}_T - {\bf x}^\star\| \le &\, 
2(1-\gamma\mu)^T\|{\bf x}_0 - {\bf x}^\star\|^2 
+ \frac{\gamma L\sigma^2}{\eta_{\text{chain}}\mu^2}\\
+& \frac{\|\tilde{\nu}^{(\zeta)}-\pi\|_{\text{TV}}^2\sigma^2L}{\mu^3}.
\end{align*}
Since, in practice, it is unrealistic to assume that we can artificially condense the time interval between the termination of the last RW and the creation of the next RW, we instead let 
 $T\to\infty$, and we complete the proof.

\section{Proof of Proposition~\ref{pro:NumIter}}\label{Appe:NumIter}

The creation threshold is $A_u=A\ge 1$ at all benign nodes, and a creation trial succeeds with a probability of $q\in(0,1]$ once some node's waiting time reaches $A$.

The upper bound $\E[\itr_t]\le t$ is immediate since at most one learning iteration can occur per time slot.

For the lower bound, we consider worst case, i.e., the \cwl algorithm with a \emph{dominated single RW process} that allows \emph{at most one} active RW at any time. whenever a new RW would be created while one is already active, we suppress the creation; when there is no active RW and a creation trial succeeds, a single RW is created. This case can only reduce the number of learning iterations, so if we denote by $\itr_t^{\mathrm{dom}}$ the number of iterations in the dominated process, then
\begin{align*}
\itr_t^{\mathrm{dom}} \le \itr_t \quad \text{a.s.}\,\,\Rightarrow\,\, \E[\itr_t^{\mathrm{dom}}] \le \E[\itr_t].
\end{align*}
It therefore suffices to lower bound $\E[\itr_t^{\mathrm{dom}}]$.

In the dominated process, define a renewal cycle as the interval from the birth of the $i$-th RW to the birth of the $(i+1)$-th RW. Let $T_{i,1}$ be the lifetime (number of steps) of the $i$-th RW until it hits the Pac-Man node, $T_{i,2}$ the idle time from that absorption until \emph{some} benign node accumulates $A$ waiting slots, and $T_{i,3}$ the additional time until a creation trial succeeds (each slot independently with probability $q$). During a cycle, learning iterations occur exactly in the $T_{i,1}$ slots when the RW is alive; hence the ``reward'' in cycle $i$ is $T_{i,1}$, and the cycle length is
\begin{align*}
L_i = T_{i,1} + T_{i,2} + T_{i,3}.
\end{align*}
We now bound the expectations of these terms:
\begin{compactenum}[(i)]
\item \textit{Lifetime.} At each step the RW hits the Pac-Man node with probability $\frac{1}{N}$, and is terminated with probability $\zeta$, so $T_{i,1} \sim \mathrm{Geom}\left(\frac{\zeta}{N}\right)$
and thus $\E[T_{i,1}] = \frac{N}{\zeta}$.

\item \textit{Threshold wait.} Immediately after absorption there is no active RW; every node's waiting counter increases by $1$ each slot. Let $M_i$ be the maximum counter among benign nodes at the absorption time of the $i$-th RW. Since at the absorption step only one node is visited, at least one other node has a positive counter, so $M_i\ge 1$ (for $N\ge 2$). Therefore the additional time to reach threshold $A$ is at most $A - M_i \le A-1$, giving
\begin{align*}
\E[T_{i,2}] \le A - 1.
\end{align*}

\item \textit{Creation delay.} Once a node is eligible, a new RW is created via an independent Bernoulli$(q)$ trial each slot, so $T_{i,3} \sim \mathrm{Geom}(q)$
and $\E[T_{i,3}] = 1/q$.
\end{compactenum}

By the renewal–reward theorem, the long-run fraction of updating slots in the dominated process equals
\begin{align*}
\frac{\E[T_{i,1}]}{\E[L_i]}=&\frac{\frac{N}{\zeta}}{\frac{N}{\zeta} + \E[T_{1,2}] + \E[T_{1,3}]}\\
\ge&\frac{\frac{N}{\zeta}}{\frac{N}{\zeta} + A-1 + 1/q},
\end{align*} 
which implies
\begin{align*}
\lim_{t\to\infty}\frac{\E[\itr_t^{\mathrm{dom}}]}{t}\ge\frac{\frac{N}{\zeta}}{\frac{N}{\zeta} + A - 1 + 1/q}.
\end{align*}
It follows that
\begin{align*}
\lim_{t\to\infty}\frac{\E[\itr_t]}{t}\ge\frac{\frac{N}{\zeta}}{\frac{N}{\zeta} + A - 1 + 1/q}.
\end{align*}

Finally, each local update is associated with one RW transmission. Note that $C_t$ counts the total number of transmissions generated by all active RWs on the graph. Therefore, $\itr_t \le C_t$, and any lower bound on $\E[\itr_t]$ also yields a lower bound on $\E[C_t]$. In particular,
\begin{align*}
\liminf_{t\to\infty}\frac{\E[C_t]}{t} \ge \frac{\frac{N}{\zeta}}{\frac{N}{\zeta} + A - 1 + 1/q}.
\end{align*}

\section{Discussions of Extensions to Multiple Pac-Man Nodes and General Graphs}\label{Appe:Extension}

\subsection{A Discussion on Extending Theorem~\ref{thm:FiniteRWs}}\label{Appe:Extending Theorem bound}
The argument underlying Theorem~\ref{thm:FiniteRWs} can be extended to multiple Pac-Man nodes with only minor modifications. The inequality $\E[Z_{t+d}-Z_t\mid\cF_t] \le  - c'\zeta Z_t + (N-k)d'$ continues to hold in the multiple Pac-Man setting, although the parameters are modified accordingly. Here, $k$ denotes the number of malicious nodes, and $c'$ is defined as the minimum probability that a RW starting from any node $u\in\cB$ reaches one of the malicious nodes within $d'$ steps. Unlike the single Pac-Man setting, its value depends not only on $d'$ (and the graph topology), but also on both the number and the locations of the malicious nodes. With these modified parameters, the same drift inequality is obtained, and the remainder of the analysis follows the same arguments as in the single Pac-Man case.

\subsection{A Discussion on Extending Theorem~\ref{thm:Peak}}\label{Appe:Extending Theorem Peak}

We first discuss how Theorem~\ref{thm:Peak} can be extended to multiple Pac-Man nodes on a complete graph. The same proof framework applies with only minor modifications. On a complete graph, the locations of the Pac-Man nodes do not affect the RW population dynamics; only the number of Pac-Man nodes matters. Let $k$ denote the number of Pac-Man nodes. By repeating the same drift analysis as in the single Pac-Man case, we obtain
$\E[Z_t]\le \frac{q}{k\zeta}(N-k)N + \left(1-\frac{k\zeta}{N}\right)^t
\left(z_0-\frac{q}{k\zeta}(N-k)N\right)$,
which generalizes the upper bound in the single Pac-Man setting. The effect of multiple Pac-Man nodes is reflected through the parameter $k$, which appears explicitly in the bound. Once this modified upper bound is established, the subsequent analysis proceeds in essentially the same manner as in the single Pac-Man case, leading to the corresponding results for the multiple Pac-Man setting.

We next discuss how the analysis of Theorem~\ref{thm:Peak} can be extended beyond complete graphs. For clarity, the theorem is stated only for complete graphs, where the key ideas can be presented in their simplest form. Nevertheless, the underlying analytical framework is not restricted to complete graphs and extends to arbitrary connected graphs. Specifically, by following the same proof strategy, one can derive an upper bound of the form $\E[Z_t]\le f(q,\zeta,N)+h(q,\zeta,N)^ts(q,\zeta,N)$, where $0<h(q,\zeta,N)<1$. The functions $f$, $h$, and $s$ depend on the graph structure as well as the diameter $d'$ and the constant $c'$ defined analogously to those in Theorem~\ref{thm:FiniteRWs}. As a result, their explicit expressions become considerably more involved; these quantities can still be characterized in principle. Once such an upper bound is established, the remainder of the analysis proceeds essentially unchanged, leading to the corresponding conclusions for general connected graphs.

\subsection{A Discussion on Extending Theorem~\ref{thm:Effectiveness}}\label{Appe:Extending Theorem Optimal}
The analysis underlying Theorem~\ref{thm:Effectiveness} can be adapted to the setting with multiple Pac-Man nodes with only minor modifications. As discussed in Remark~\ref{remark:MultiplePacMan}, the presence of multiple Pac-Man nodes affects the analysis only through the RW chain transition matrix. Specifically, for any connected graph, once Pac-Man configuration (i.e., the number and locations of the Pac-Man nodes) is given, the modified transition probability matrix $P'$ can be constructed from \eqref{eq:NewTransMatCase1} and \eqref{eq:NewTransMatCase2}, which in turn determines the corresponding RW chain transition matrix $P_{\text{chain}}^{(\zeta)}$ in \eqref{eq:UniqueP} and its stationary distribution $\pi_{\text{chain}}^{(\zeta)}$ in \eqref{eq:pichain}. The remainder of the analysis remains unchanged: by substituting the modified $P_{\text{chain}}^{(\zeta)}$ and $\pi_{\text{chain}}^{(\zeta)}$ into the same analytical framework, the corresponding results for the multiple Pac-Man setting can be obtained.

\subsection{A Discussion on Extending Proposition~\ref{pro:Bounds}}\label{Appe:Extending Theorem Shift}
The analysis underlying Proposition~\ref{pro:Bounds} can be adapted to the setting with multiple Pac-Man nodes. Once the corresponding RW chain transition matrix $P_{\text{chain}}^{(\zeta)}$ and its stationary distribution $\pi_{\text{chain}}^{(\zeta)}$ are obtained, the same analytical framework can be applied to derive the corresponding results for the multiple-Pac-Man setting.

\end{document}

%% file: header.tex
\usepackage[utf8]{inputenc} 
\usepackage[T1]{fontenc}  
\usepackage{mathrsfs}
\usepackage{%
algorithm,%
algpseudocode,%
amsfonts,%
amsmath,%
amssymb,%
amsthm,%
bbm,%
caption,%
cite,%
comment,%
extarrows,%
float,%
graphicx,%
hyperref,%
indentfirst,%
listings,%
paralist,%
pgfplots,%
subcaption,%
tikz,%
todonotes,%
url,%
xcolor,%
makecell,
xspace
}
\usepackage{epstopdf}

\pgfplotsset{plot coordinates/math parser=false}
\usepgflibrary{%
patterns,%
}

\usetikzlibrary{%
arrows,%
automata,%
chains,%
decorations.pathreplacing,%
intersections,%
positioning,%
shapes,%
}

\IEEEoverridecommandlockouts  
\newtheorem{theorem}{Theorem}
\newtheorem{lemma}{Lemma}

\theoremstyle{definition}
\newtheorem{definition}{Definition}

\newtheorem{proposition}{Proposition}

\newtheorem{assumption}{Assumption}
\newtheorem{cor}{Corollary}

\theoremstyle{remark}
\newtheorem{remark}{Remark}

\definecolor{darkgreen}{rgb}{0.0, 0.5, 0.0}

\newcommand{\set}[1]{\left\lbrace#1\right\rbrace}

\newcommand{\Ind}[1]{\mathbbm{1}_{#1}}
\newcommand{\abs}[1]{\left\lvert#1\right\rvert}
\newcommand{\norm}[1]{\left\lVert#1\right\rVert}

\newcommand{\E}{\mathbb{E}}
\newcommand{\N}{\mathbb{N}}
\newcommand{\R}{\mathbb{R}}
\newcommand{\Z}{\mathbb{Z}}

\newcommand{\bx}{\mathbf{x}}

\newcommand{\cA}{\mathcal{A}}
\newcommand{\cB}{\mathcal{B}}
\newcommand{\cD}{\mathcal{D}}

\newcommand{\cF}{\mathcal{F}}
\newcommand{\cG}{\mathcal{G}}
\newcommand{\cH}{\mathcal{H}}

\newcommand{\cM}{\mathcal{M}}

\newcommand{\cV}{\mathcal{V}}

\newcommand{\cZ}{\mathcal{Z}}

\newcommand{\Prb}{\mathrm{Pr}}
\newcommand{\itr}{\mathrm{Iter}}
\newcommand{\iid}{\emph{i.i.d.}\ }

\newcommand{\CWL}{\textsc{Create-If-Late} }
\newcommand{\cwl}{\text{CIL}\xspace}
\newcommand{\DeCa}{\textsc{DecAFork}\xspace}
\newcommand{\ac}{\textsc{AC}\xspace}
\newcommand{\MP}{\textsc{MissingPerson}\xspace}
\newcommand{\ER}{\text{Erd\H{o}s--R\'enyi}\xspace}

\renewcommand{\ge}{\geqslant}
\renewcommand{\le}{\leqslant}